\documentclass[onecolumn, 10 pt, journal]{IEEEtran}
\usepackage[utf8]{inputenc}
\usepackage{amsmath,amssymb,euscript,mathrsfs,amsthm}
\usepackage{graphicx,color}
\usepackage{authblk} 
\usepackage{multirow}  
\usepackage{hhline}
\usepackage{caption} 
\usepackage{subcaption}
\usepackage{enumerate,paralist} 
\usepackage{tikz}
\usetikzlibrary{automata,positioning,arrows}
\usetikzlibrary{positioning,chains,fit,shapes,calc}
\usetikzlibrary{graphs,graphs.standard} 
\usepackage{hyperref}
\usepackage{algorithm}
\usepackage{algorithmic}

\newcommand{\cref}[1]{Claim~\ref{#1}}

\newtheorem{theorem}{\bf Theorem }

\newtheorem{lemma}{\bf Lemma}
\newtheorem{corollary}{\bf Corollary}

\newtheorem{assumption}{\bf Assumption}

\newcommand{\vect}[1]{\boldsymbol{#1}}
\newcommand{\mat}[1]{\boldsymbol{#1}}

\newcommand{\CMP}{\mathbb{C}}
\newcommand{\REAL}{\mathbb{R}}
\newcommand{\BIGO}{\mathcal{O}}
\newcommand{\IMG}{{\bf i}}
\newcommand{\TSP}{^{\rm T}}
\newcommand{\HET}{^{\rm H}}
\newcommand{\DIAG}[1]{{\rm diag}(#1)}
\newcommand{\PRO}[1]{\mathbb{P}\left\{#1\right\}}
\newcommand{\PROL}[1]{\mathbb{P}\{#1\}}
\newcommand{\EXP}[1]{\mathbb{E}\left[#1\right]}
\newcommand{\EXPL}[1]{\mathbb{E}[#1]}

\newcommand{\ABS}[1]{\left|#1\right|}
\newcommand{\ABSL}[1]{|#1|}

\newcommand{\SUPP}[1]{{\rm supp}(#1)}

\DeclareMathOperator{\polylog}{polylog}
\DeclareMathOperator{\poly}{poly}
 
\begin{document}
\title{Learning Mixtures of Sparse Linear Regressions Using Sparse Graph Codes}

\author{Dong~Yin,
        Ramtin~Pedarsani,
        Yudong~Chen,
        and~Kannan~Ramchandran%
\thanks{D. Yin and K. Ramchandran are with the Department
of EECS at UC Berkeley, email: \{dongyin, kannanr\}@eecs.berkeley.edu.}%
\thanks{R. Pedarsani is with the Department of ECE at UC Santa Barbara, email: ramtin@ece.ucsb.edu.}%
\thanks{Y. Chen is with the School of ORIE at Cornell University, email: yudong.chen@cornell.edu.}
\thanks{This paper was presented in part at IEEE 55th Annual Allerton Conference on Communication, Control, and Computing, 2017.}}
 
\maketitle

\begin{abstract}
In this paper, we consider the \emph{mixture of sparse linear regressions} model. Let $\vect{\beta}^{(1)},\ldots,\vect{\beta}^{(L)}\in\CMP^n$ be $ L $ unknown sparse parameter vectors with a total of $ K $ non-zero elements. Noisy linear measurements are obtained in the form $y_i=\vect{x}_i\HET \vect{\beta}^{(\ell_i)} + w_i$, each of which is generated randomly from one of the sparse vectors with the label $ \ell_i $ unknown.
The goal is to estimate the parameter vectors efficiently with low sample and computational costs. This problem presents significant challenges as one needs to simultaneously solve the \emph{demixing} problem of recovering the labels $ \ell_i $ as well as the \emph{estimation} problem of recovering the sparse vectors $ \vect{\beta}^{(\ell)} $.

Our solution to the problem leverages the connection between modern coding theory and statistical inference. We introduce a new algorithm, \emph{Mixed-Coloring}, which samples the mixture strategically using query vectors $ \vect{x}_i $ constructed based on ideas from sparse graph codes. Our novel code design allows for both efficient demixing and parameter estimation. To find $K$ non-zero elements, it is clear that we need at least $\Theta(K)$ measurements, and thus the time complexity is at least $\Theta(K)$. In the noiseless setting, for a constant number of sparse parameter vectors, our algorithm achieves the order-optimal sample and time complexities of $\Theta(K)$. In the presence of Gaussian noise,\footnote{The proposed algorithm works even when the noise is non-Gaussian in nature, but the guarantees on sample and time complexities are difficult to obtain.} for the problem with two parameter vectors (i.e., $L=2$), we show that the Robust Mixed-Coloring algorithm achieves near-optimal $\Theta(K\polylog(n))$ sample and time complexities. When $K=\BIGO(n^{\alpha})$ for some constant $\alpha\in(0,1)$ (i.e., $K$ is sublinear in $n$), we can achieve sample and time complexities both sublinear in the ambient dimension. In one of our experiments, to recover a mixture of two regressions with dimension $n=500$ and sparsity $K=50$, our algorithm is more than $300$ times faster than EM algorithm, with about one third of its sample cost.
\end{abstract}

\IEEEpeerreviewmaketitle
  
\section{Introduction}\label{sec:intro} 
Mixture and latent variable models, such as Gaussian mixtures and subspace clustering, are expressive, flexible, and widely used in a broad range of problems including background modeling~\cite{harville2002framework}, speaker identification~\cite{reynolds2000speaker} and recommender systems~\cite{zhang2012guess}. However, parameter estimation in mixture models is notoriously difficult due to the non-convexity of the likelihood functions and the existence of local optima. In particular, it often requires a large sample size and many re-initializations of the algorithms to achieve an acceptable accuracy.

Our goal is to develop provably fast and efficient algorithms for mixture models---with sample and time complexities \emph{sublinear} in the problem's ambient dimension when the parameter vectors of interest are sparse---by leveraging the underlying low-dimensional structures.

In this paper we focus on a powerful class of models called \emph{mixtures of linear regressions}~\cite{de1989mixtures}. We consider the \emph{sparse} setting with a \emph{query-based} algorithmic framework. In particular, we assume that each query-measurement pair $ (\vect{x}_i, y_i) $ is generated from a  sparse linear model chosen randomly from $ L $ possible models:\footnote{We use $ \vect{x}_i\HET $ to denote the conjugate transpose of $ \vect{x}_i $. In this paper, for any positive integer $N$, $[N]$ denotes the set $\{1,2,\ldots,N\}$.}
\begin{align}
y_i=\vect{x}_i\HET\vect{\beta}^{(\ell)}+w_i ~ \text{ with probability } q_{\ell}, \; \text{ for  }  \ell \in [L] ,
\label{eq:model}
\end{align}
where $ w_i $ is noise. Here, the probability $q_{\ell} > 0 $ is also called the \emph{mixture weight} of $\vect{\beta}^{(\ell)}$, and they satisfy  $\sum_{\ell=1}^L q_{\ell} = 1$.
The total number of non-zero elements in the parameter vectors $ \{ \vect{\beta}^{(\ell)} \in \CMP^n, \ell\in [L] \}$ is assumed to be $K$. The goal is to estimate the $\vect{\beta}^{(\ell)}$'s, without knowing which $ \vect{\beta}^{(\ell)}  $ generates each query-measurement pair. We also note that when $L=1$, we recover the compressive sensing problem that has been extensively studied in recent years~\cite{candes2006robust,candes2006stable}.

A mixture of regressions provides a flexible model for various heterogeneous settings where the regression coefficients differ for different subsets of observations. This model has been applied to a broad range of tasks including medicine measurement design~\cite{blackwell2006applying}, behavioral health care~\cite{deb2002estimates} and music perception modeling~\cite{viele2002modeling}. Here, we study the problem when the query vectors $\vect{x}_i$ can be \emph{designed} by the user;  in  Section~\ref{sec:motivation} we discuss several practical applications that motivate the study of this query-based setting.  Our results show that by appropriately exploiting this design freedom, one can achieve significant reduction in the sample and computational costs. 
  
To recover $ K $ unknown non-zero elements, the number of linear measurements needed scales at least as $ \Theta(K) $. The corresponding time complexity is also at least $ \Theta(K) $, which is the time needed to write down $K$ numbers as the solution. We introduce a new algorithm, called the \emph{Mixed-Coloring} algorithm, that \emph{matches these sublinear sample and time complexity lower bounds}. The design of query vectors and decoding algorithm leverages ideas from sparse graph codes such as low-density parity-check (LDPC) codes~\cite{gallager1962low}.
For any $L=\Theta(1)$, our algorithm recovers the parameter vectors with optimal $ \Theta(K) $ sample and time complexities in the noiseless setting, both in theory and empirically. In the noisy setting, for compressive sensing problems (i.e., $L=1$), it is known from an information-theoretic point of view that the optimal sample complexity is $\Theta(K\log(n/K))$~\cite{wainwright2009information,akccakaya2010shannon}. In this work, we show that when the noise is Gaussian distributed, $L=2$, and the non-zero elements take value in a finite quantized set, the Robust Mixed-Coloring algorithm has $ \Theta(K\polylog(n))$ sample and time complexities. Since our problem is harder than compressive sensing, the sample and time complexities of our algorithm are optimal up to polylogarithmic factors. When $K=\BIGO(n^\alpha)$ for some $\alpha\in(0,1)$, the sample and time complexities are sublinear in the ambient dimension $n$. In noisy setting with continuous-valued parameter vectors, we provide experimental results and show that our algorithm can successfully recover the best quantized approximation of the parameter vectors, provided that the continuous-valued parameter vectors are close to the quantized grid in $\ell_\infty$ norm\footnote{In Section~\ref{sec:experiments}, we formally define the \emph{perturbation} of the continuous-valued parameter vector $\vect{\beta}$ with respect to the quantized alphabet. The notion of perturbation essentially measures the distance between the continuous-valued parameter vector and the quantized grid in $\ell_\infty$ norm.}. Prior literature on this problem that does not utilize the design freedom typically have sample and time complexities that are at least polynomial in $n$; we provide a survey of prior work and a more detailed comparison in Section~\ref{sec:related}. Empirically, we find that our algorithm is orders of magnitude faster than standard Expectation-Maximization (EM) algorithms for mixture of regressions. For example,  in one of our experiments, detailed in Section~\ref{sec:experiments}, we consider recovering a mixture of two regressions with dimension $n=500$ and sparsity $K=50$; our algorithm is more than $300$ times faster than EM algorithm, with about $ 1/3 $ of its sample cost.

\subsection{Algorithm Overview}\label{sec:overview}
Our Mixed-Coloring algorithm solves two problems simultaneously: (i) rapid {demixing}, namely identifying the \emph{label} $ \ell_i $ of the vector $ \vect{\beta}^{(\ell_i)} $ that generates each measurement $ y_i $; (ii) efficient identification of the \emph{location} and \emph{value} of the non-zero elements of the $\vect{\beta}^{(\ell)}  $'s. The main idea is to use a divide-and-conquer approach that iteratively reduces the original problem into simpler ones with much sparser parameter vectors. More specifically, we design $\Theta(K)$ sets of sparse query vectors, with each set only associated with a subset of all the non-zero elements. The design of the query vectors ensures that we can first identify the sets which are associated with a single non-zero element (called singletons), and recover the location and value of that element (motivated by a balls-and-bins model that we utilize for designing our measurements, we call them singleton balls, shown as shaded balls in Figure~\ref{subfig:singleton}). We further identify the pairs of singleton balls which have the same (but unknown) label, indicated by the edges in Figure~\ref{subfig:singleton}. Results from random graph theory guarantee that, with high probability, the $L$ largest connected components (giant components) of the singleton graph have different labels, and thus we recover a fraction of the non-zero elements in each $\vect{\beta}^{(\ell)}$, as shown in Figure~\ref{subfig:giant}. We can then iteratively enlarge the recovered fraction with a guess-and-check method until finding all the non-zero elements. We revisit Figure~\ref{fig:algorithm} when describing the details of our algorithm in Section~\ref{sec:algoutline}.
\begin{figure}
	\centering
	\begin{minipage}{\textwidth}
		\centering
		\begin{subfigure}[b]{0.2\textwidth}
			\includegraphics[width=\textwidth]{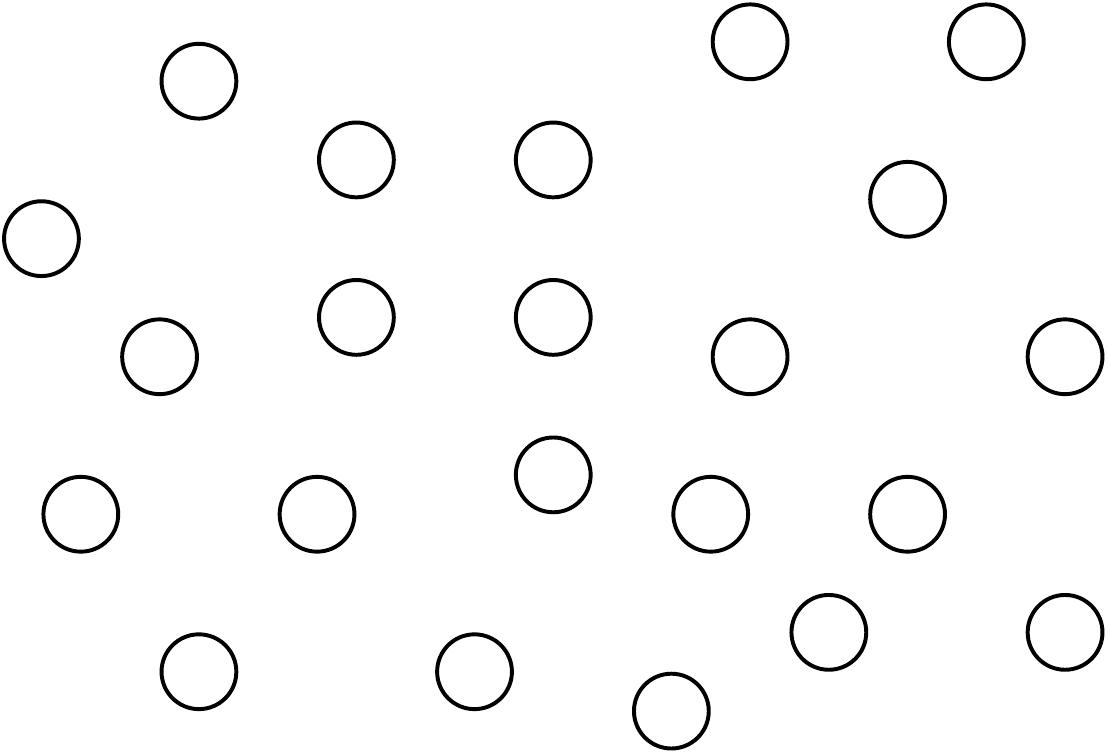}
			\caption{Non-zero elements}
			\label{subfig:nonzero}
		\end{subfigure}
		\quad 
		\begin{subfigure}[b]{0.2\textwidth}
			\includegraphics[width=\textwidth]{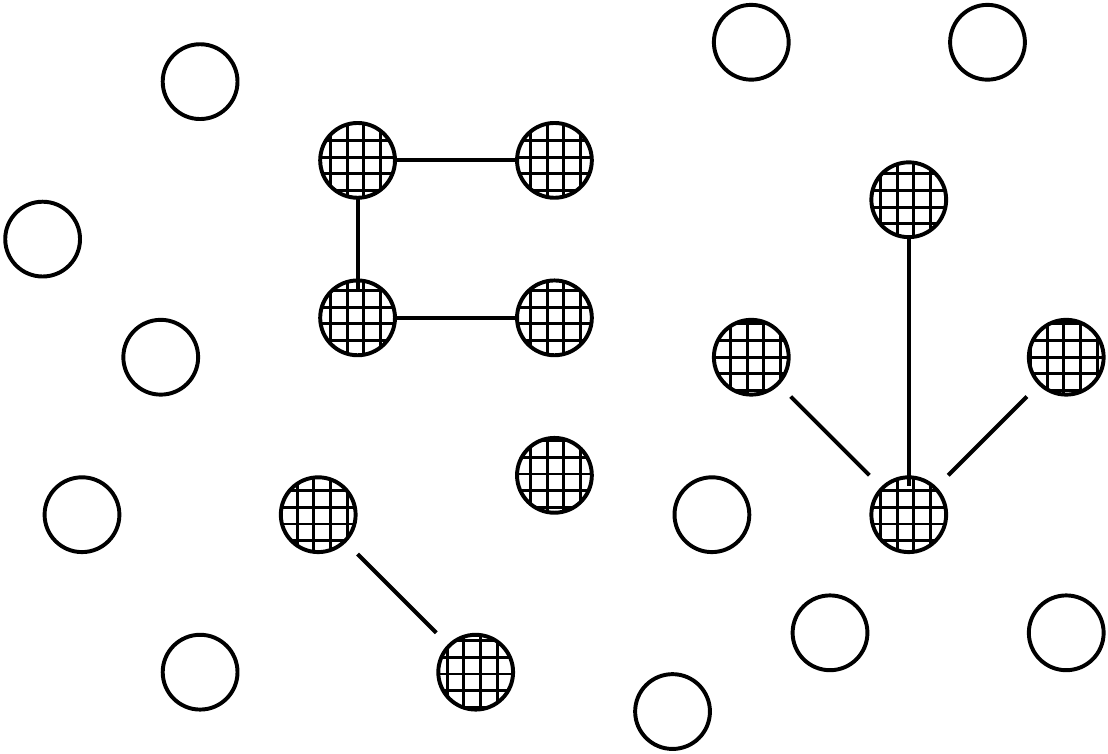}
			\caption{Singleton balls}
			\label{subfig:singleton}
		\end{subfigure}
		\quad 
		\begin{subfigure}[b]{0.2\textwidth}
			\includegraphics[width=\textwidth]{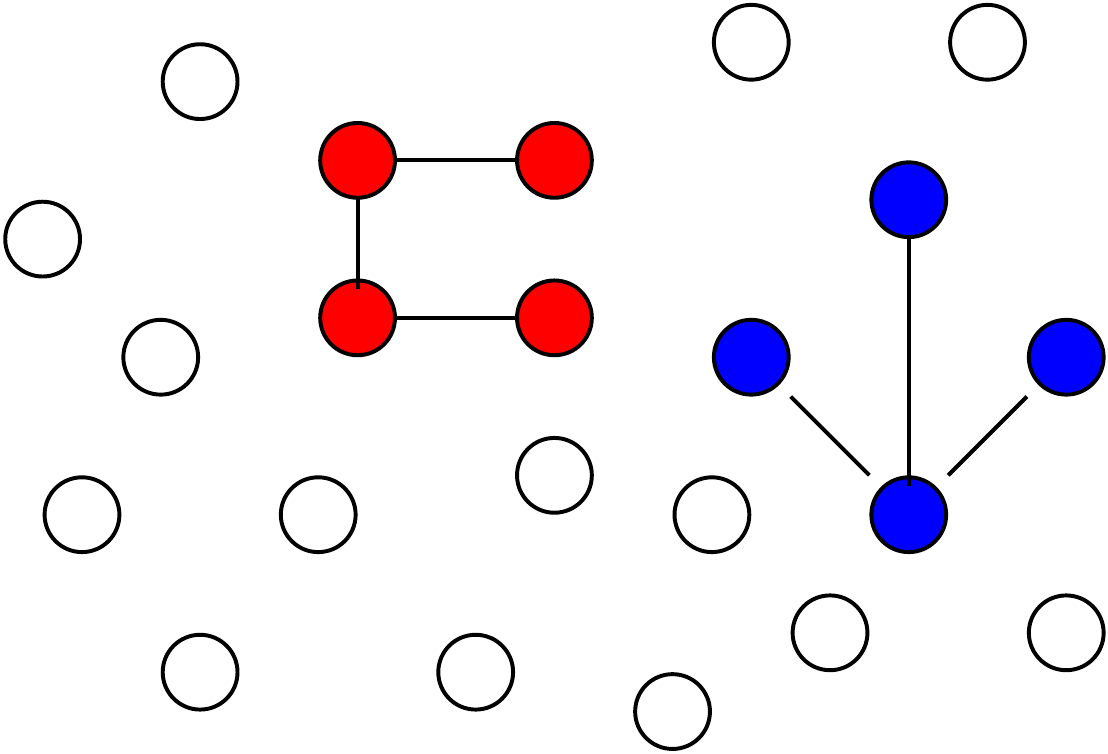}
			\caption{Giant components}
			\label{subfig:giant}
		\end{subfigure}
		\quad 
		\begin{subfigure}[b]{0.2\textwidth}
			\includegraphics[width=\textwidth]{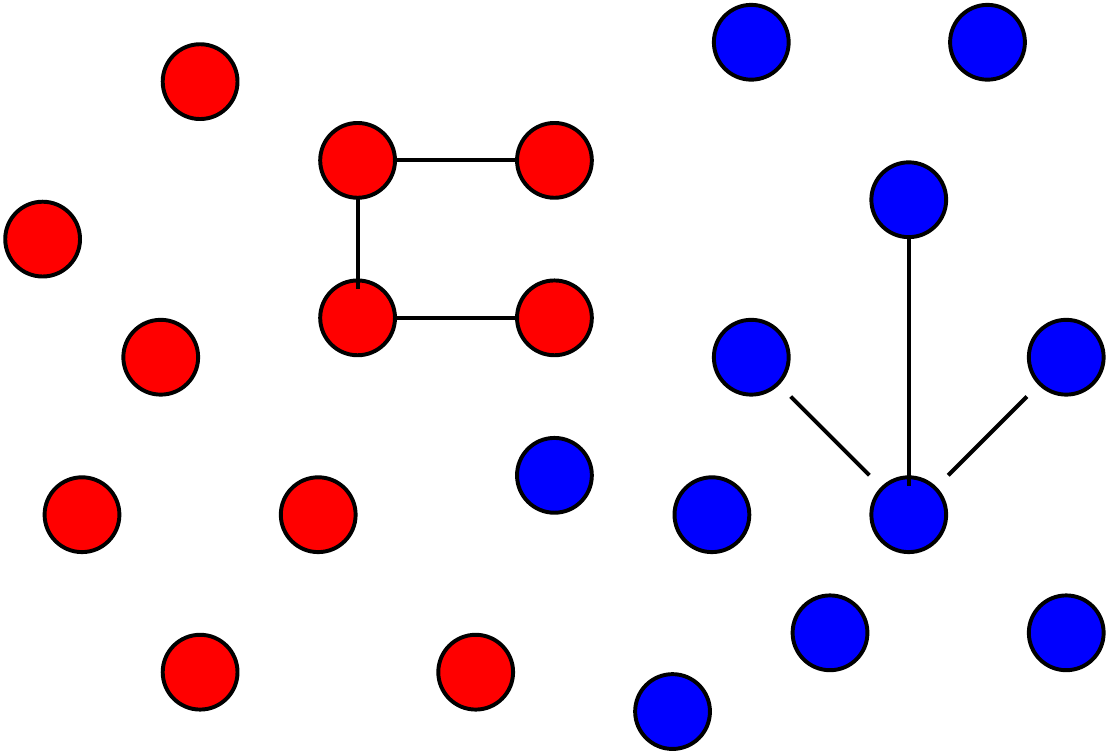}
			\caption{Results}
			\label{subfig:colored}
		\end{subfigure}
		\caption{\small{Mixed-Coloring algorithm with $ L=2 $. (a) A graph with nodes corresponding to the non-zero elements in the two parameter vectors. (b) Detect some non-zero elements without knowing their labels, i.e., singleton balls, shown as shaded balls, and find pairs of singleton balls which must have the same label, shown as the edges between shaded balls. (c) Find the two largest connected components of the graph, and the non-zero elements in these two components must have different labels, shown as red and blue. Thus, we recover a fraction of the non-zero elements in each parameter vector. (d) Recover the whole parameter vectors using iterative decoding. The non-zero elements in the two parameter vectors are shown in blue and red, respectively.}}\label{fig:algorithm}
	\end{minipage}
\end{figure}

\subsection{Motivation}\label{sec:motivation}
Our problem is a natural extension of the setting of compressive sensing, in which one often has full freedom of designing query vectors in order to estimate a sparse parameter vector. In many applications, the unknown sparse parameter vector can be affected by latent variables, leading to a mixture of sparse linear regressions, and these scenarios have been observed in neuroscience~\cite{lewicki1998review}, genetics~\cite{jansen1992general}, psychology~\cite{blackwell2006applying}, etc. Here, we provide a concrete example motivated by neuroscience applications. In neural signal processing, sensors are used to measure the brain activities, represented by an unknown sparse vector $\vect{\beta}$. The sensors can be modeled as digital filters, and one can \emph{design} the linear filter weights ($\vect{x}_i$'s) when measuring the neural signal. Multiple sensors are usually placed in a particular area of the brain in order to acquire enough compressed measurements. However, there may be more than one neuron affecting a particular area of the brain, as shown in Figure~\ref{fig:neuron}, and different neurons may have different activities, corresponding to the $\vect{\beta}^{(\ell)}$'s. Consequently, each sensor may be measuring one of several different sparse signals. Further, if we use the sensors multiple times, a single sensor may even obtain measurements that are generated by different neurons, since neurons are different depths may be active during different time periods. Thus, the problem can be formulated as a mixture of sparse linear regressions. Variants of this problem, such as neural spike sorting~\cite{lewicki1998review}, have been studied in neuroscience. While the common solution is to use clustering algorithms on the spike signals, we believe that our algorithm provides the potential of improving sensor design and reducing sample and time complexities.
\begin{figure}[h]
\centering
\includegraphics[width=0.48\textwidth]{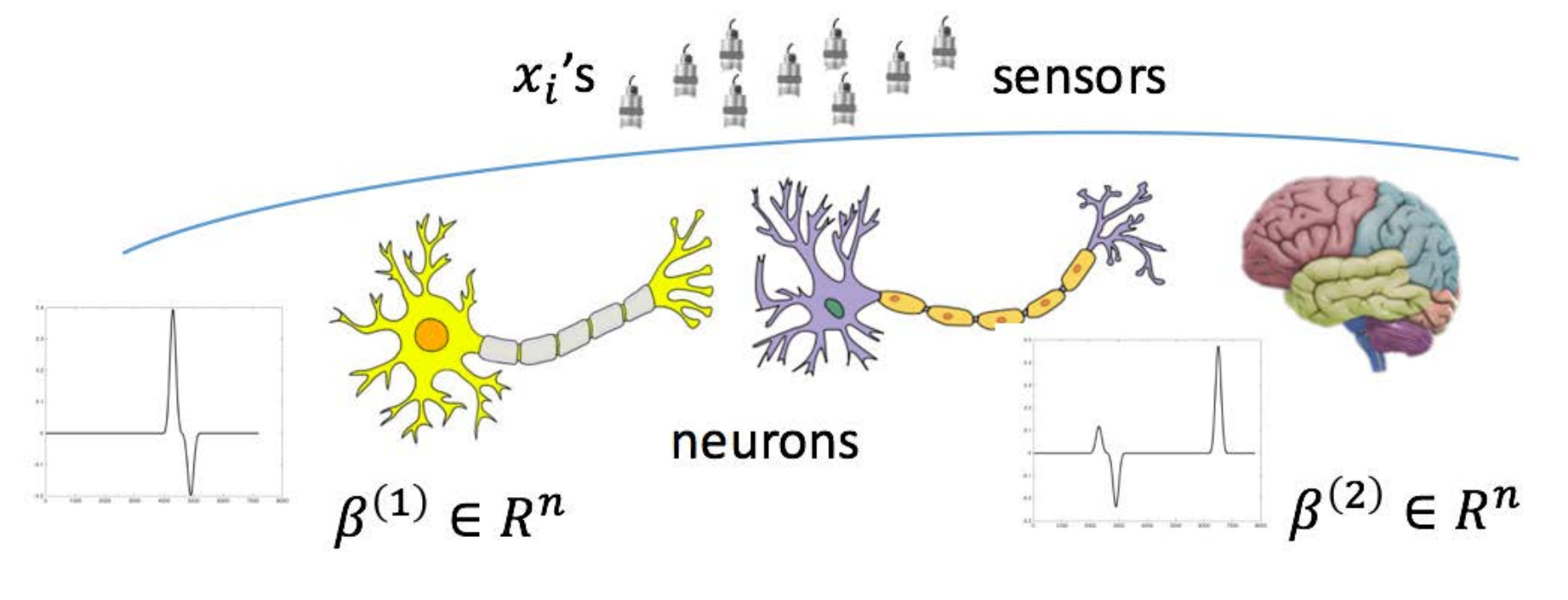}
\caption{\small{Mixture of neural signals. Sensors (modeled as digital filters) are used to measure the brain activities, represented by an unknown sparse vector $\vect{\beta}$. One can \emph{design} the linear filter weights of the sensors ($\vect{x}_i$'s) when measuring the neural signals. Multiple sensors are usually placed in a particular area of the brain in order to acquire enough compressed measurements. However, there may be more than one neuron affecting a particular area of the brain, and different neurons may have different activities, corresponding to the $\vect{\beta}^{(\ell)}$'s. In this case, we have a mixture-of-sparse-linear-regressions problem.}}
\label{fig:neuron}
\end{figure}

In addition, our work adds the intellectual value of the power of design freedom in tackling sparse mixture problems by highlighting the significant performance gap between algorithms that can exploit the design freedom and those that cannot. We also believe that our ideas are applicable more broadly for other latent-variable problems that require experimental designs, such as survey designs in psychology with mixed type of respondents and biology experiments with mixed cell interior environments.

\subsection{Organization}
We summarize our main results in Section~\ref{sec:main_thm}, discuss related works in Section~\ref{sec:related}, present the details of our algorithm in the noiseless and noisy settings in Sections~\ref{sec:algoutline} and~\ref{sec:noise}, respectively, provide experimental results in Section~\ref{sec:experiments}, and make conclusions in Section~\ref{sec:con}.

\section{Main Results}\label{sec:main_thm}
In this section, we present the recovery guarantees for the Mixed-Coloring algorithm, and provide bounds on its sample and time complexities. We assume there are $L$ unknown $ n $-dimensional parameter vectors $\vect{\beta}^{(1)},\ldots,\vect{\beta}^{(L)}$. Each $ \vect{\beta}^{(\ell)} $ has $ K_\ell $ non-zero elements, i.e., $ \ABSL{ \SUPP{\vect{\beta}^{(\ell)} } }  = \ABSL{ \{ j : \beta_j^{(\ell)} \neq 0 \} } = K_\ell$. Let $ K=\sum_{\ell=1}^{L} K_\ell$ be the total number of non-zero elements. 
Using the query vectors $\{ \vect{x}_i \}\in\CMP^n$, the Mixed-Coloring algorithms obtains $m$ measurements $ y_i $, $i\in[m]$ generated independently according to the model~\eqref{eq:model}, and outputs an estimate $\{ \hat{\vect{\beta}}^{(\ell)}$, $\ell\in[L] \}$ of the unknown parameter vectors. 
We defer more details to Sections~\ref{sec:algoutline} and~\ref{sec:noise}.

Our results are stated in the asymptotic regime where $n$ and $K$ approach infinity. 
A constant is a quantity that does not depend on $n$ and $K$, with the associated Big-O notations $ \BIGO(\cdot)$ and $\Theta(\cdot) $. We assume that $ L $ is a known and fixed constant, and the mixture weights satisfy $q_{\ell}=\Theta(1)$ for each $\ell\in [L]$ and thus are of the same order.  Similarly, the sparsity levels of the parameter vectors are also of the same order with $ K_\ell = \Theta(K)$. 

\subsection{Guarantees for the Noiseless Setting}\label{sec:noiseless_results}
In the noiseless case, i.e., $ w_i \equiv 0 $, we consider for generality the complex-valued setting with $ \vect{\beta}^{(\ell)}  \in\CMP^n $ (our results can be easily applied to real case). We make a mild technical assumption, which stipulates that if any pair of parameter vectors have overlapping support, then the elements in the overlap are different. 

\begin{assumption}
\label{asm:identifiability}
For each pair $\ell_1,\ell_2\in[L]$, $\ell_1\neq \ell_2$ and each index $j\in\SUPP{\vect{\beta}^{(\ell_1)}}\cap\SUPP{\vect{\beta}^{(\ell_2)}}$, we have $\beta_j^{(\ell_1)}\neq \beta_j^{(\ell_2)}$.
\end{assumption}
 
We need this assumption due to our element-wise recovery strategy. However, this assumption is mild in practice. In particular, in the noiseless case, if the non-zero elements are generated from some continuous distribution, it is a measure zero event that two elements at the same coordinate share exactly the same value.
Under the above setting, we have the following recovery guarantees for the Mixed-Coloring algorithm.
\begin{theorem}\label{thm:main_noiseless}
Consider the asymptotic regime where $n$ and $K$ approach infinity. Under Assumption~\ref{asm:identifiability}, for any fixed constant $p^* \in (0,1)$, there exists a constant $C>0$ such that if the number of measurements is $ m \ge CK $, then the Mixed-Coloring algorithm guarantees the following three properties for each $\ell\in[L]$ (up to a label permutation):
\begin{compactenum}
	\item (No false discovery) For each $j\in\SUPP{\vect{\beta}^{(\ell)}}$, $\hat{\beta}^{(\ell)}_j$ equals either $\beta^{(\ell)}_j$ or $0$; for each $j\notin\SUPP{\vect{\beta}^{(\ell)}}$, $\hat{\beta}^{(\ell)}_j=0$.
	\item (Support recovery) 
	$$
	\mathbb{P} \big\{\ABSL{\SUPP{\hat{\vect{\beta}}^{(\ell)}}} \ge (1-p^*)\ABSL{\SUPP{\vect{\beta}^{(\ell)}}} \big\} \ge 1- \BIGO(1/K).
	$$
	\item (Element-wise recovery) For each $j\in\SUPP{\vect{\beta}^{(\ell)}}$,
	$$
	\PROL{\hat{\beta}^{(\ell)}_j = \beta^{(\ell)}_j}\ge 1-p^*-\BIGO(1/K).
	$$
\end{compactenum}
Moreover, the computational time of the Mixed-Coloring algorithm is $\Theta(K)$.
\end{theorem}
As we can see, to recover an arbitrarily large fraction of the non-zero elements, our Mixed-Coloring algorithm has \emph{order-optimal} $\Theta(K)$ sample and time complexities. More specifically, the first property ensures that Mixed-Coloring algorithm has no false discovery: for zero elements in the parameter vectors, our algorithm does not produce non-zero estimates, and for non-zero elements, our algorithm outputs either the true value or zero. The second property ensures that the Mixed-Coloring algorithm recovers $(1-p^*)$ fraction of the non-zero elements with high probability. The third property ensures that for each non-zero element, the probability that it can be recovered is asymptotically at least $1-p^*$. In fact, the recovered fraction of the non-zero elements is uniformly distributed on the support of the parameter vectors.

The error fraction $ p^* $ is an input parameter to algorithm, and can be made arbitrarily close to zero by adjusting the oversampling ratio $C \equiv C(p^*, L, \{q_\ell\} )$. By more careful analysis, one can show that the dependence of $C$ on $p^*$ is $C=\BIGO(\log(1/p^*))$ (see the proof of Lemma~\ref{lem:errorfloor} in Appendix~\ref{sec:prf_noiseless}). Thus, when $p^*$ approaches $0$, the sample and time complexities grow slowly as $\log(1/p^*)$. Here, since we set $p^*$ as a constant, we hide this dependence in the constant $C$. Given the number of components $ L $, mixture weights $ \{q_\ell \}$ and the target $ p^* $, the value of the constant $ C $ can be computed numerically. Table~\ref{tab:thmconst} gives some of the $ C $ values for several $p^*$ and~$ L $, under the setting  $q_\ell=1/L,\forall \ell \in [L]$. We see that the value of $ C $ is quite modest. More details of computing the constants in the sample complexity can be found in Appendix~\ref{sec:constant}.

\begin{table}[h]
\caption{Sample complexity of the Mixed-Coloring algorithm}
\begin{center}
  \begin{tabular}{ c | c | c | c }
    \hline
    $L$ & $2$ & $3$ & $4$ \\ \hline
    $p^*$ & $5.1\times10^{-6}$ & $8.8\times 10^{-6}$ & $8.1\times 10^{-6}$ \\ \hline
    $m = CK$ & $33.39K$ & $37.80K$ & $40.32K$\\
    \hline   
  \end{tabular} 
\end{center}
\label{tab:thmconst}
\end{table}

We can in fact boost the above guarantee to recover all the non-zero elements, by running the Mixed-Coloring algorithm $\Theta(\log K)$ times independently and aggregating the results by majority voting. By property 2 in Theorem~\ref{thm:main_noiseless} and a union bound argument, this procedure {\it exactly} recovers all the parameter vectors with probability $1-\BIGO(1/\poly(K))$ with $\Theta(K\log K)$ sample and time complexities.	

\subsection{Guarantees for the Noisy Setting}\label{sec:noisy_results}

An extension of the previous algorithm, \emph{Robust Mixed-Coloring}, handles noise in the measurement model~\eqref{eq:model}, in the case of two parameter vectors which appear equally likely, i.e., $L=2$ and $q_\ell=1/2$, $\ell=1,2$. Many interesting applications have binary latent factors: gene mutation present/not, gender, healthy/sick individual, children/adult, etc; see also the examples in~\cite{de1989mixtures,viele2002modeling,chen2014convex}. We would like to mention that our goal is to design a query-based algorithm that can simultaneously conduct fast demixing and robust estimation in the presence of noise. Even if there are only two possible parameter vectors, achieving this goal is highly non-trivial, and we believe that our framework provides useful intellectual insights to this problem. Extending our results to the setting with $ L > 2 $ is an important and interesting direction, and we leave it to future work.

The noise $ w_i $ is assumed to be i.i.d. Gaussian with mean zero and constant variance $ \sigma^2 $. We note that the Gaussian noise assumption is mainly for theoretical reason. As one can see in Subsection~\ref{sec:noisy_decoding_alg}, our algorithm uses EM algorithm as a subroutine to estimate the component means of a mixture of two random variables. The analysis of EM algorithm is known to be hard due to the non-convexity of the likelihood functions. For simplicity, in this paper we assume that the noise is Gaussian and employ the recent convergence results on EM algorithm for two-component Gaussian mixtures in~\cite{balakrishnan2014statistical}. Since EM algorithm is widely used for non-Gaussian noise and is shown to have good performance in many applications, we believe that our algorithm can work well in practice even if the noise is not Gaussian distributed. 

In the noisy setting, we make an additional assumption that the non-zero elements in the parameter vectors take value in a finite quantized set.  
\begin{assumption}\label{asm:quantize}
The non-zero elements of the parameter vectors satisfy $ \beta_j^{(\ell)} \in \mathbb{D}, \forall \beta_j^{(\ell)}\neq0, \ell\in[L] $, where
$$
\mathbb{D} \triangleq \{\pm\Delta, \pm2\Delta,\ldots,\pm b\Delta\}\subset\REAL,
$$
The positive constants $\Delta$ and $b$ are known to the algorithms.
\end{assumption} 
Here, we note that this assumption is mild in practice. As mentioned in Theorem~\ref{thm:main_noisy}, the quantization step size $\Delta$ can be as small as a constant multiple of the standard error of the noise, and this quantization step size should be small enough for most applications. In fact, for continuous-valued parameter vectors, in the noisy setting, it is fundamental that the non-zero elements can only be recovered up to certain precision. Moreover, in our empirical results in Section~\ref{sec:experiments}, the Robust Mixed-Coloring algorithm works even when the assumption is violated. In this case, the algorithm produces the best quantized approximation to the unknown parameter vectors,  provided that they are not too far off the quantized set. We would also like to mention that it is a major challenge to develop algorithms with sublinear complexity and provable guarantees in noisy settings of sparse mixed regression, even with the assumption of quantized non-zero elements. Thus, even with this mild simplifying assumption, our work demonstrates a significant progress. Establishing strong theoretical guarantees for a fast recovery algorithm with sublinear sample and time complexities for the continuous alphabet setting remains to be an open problem..

When the quantization assumption holds, exact recovery is possible, as guaranteed in the following theorem. The Robust Mixed-Coloring algorithm maintains sublinear sample and time complexities, and recovers the parameter vectors in the presence of i.i.d. Gaussian noise.

\begin{theorem}\label{thm:main_noisy}
Consider the asymptotic regime where $K$ and $n$ approach infinity. Suppose that the noise in the measurements are i.i.d. Gaussian distributed with mean $0$ and variance $\sigma^2$, and that $\Delta/\sigma \ge \frac{4}{\sqrt{3}}$.
When $L=2$ and Assumptions~\ref{asm:identifiability} and~\ref{asm:quantize} hold, if the number of measurements is $ m = \Theta(K\polylog(n)) $, then, the Robust Mixed-Coloring algorithm guarantees the following three properties for each $\ell\in\{1,2\}$ (up to a label permutation):
\begin{compactenum}
	\item (No false discovery) With probability at least $1-\BIGO(1/\poly(n))$, for each $j\in\SUPP{\vect{\beta}^{(\ell)}}$, $\hat{\beta}^{(\ell)}_j$ equals either $\beta^{(\ell)}_j$ or $0$; for each $j\notin\SUPP{\vect{\beta}^{(\ell)}}$, $\hat{\beta}^{(\ell)}_j=0$.
	\item (Support recovery) 
$$
	\mathbb{P} \big\{\ABSL{\SUPP{\hat{\vect{\beta}}^{(\ell)}}} \ge (1-p^*)\ABSL{\SUPP{\vect{\beta}^{(\ell)}}} \big\} \ge 1- \BIGO(1/K).
$$
	\item (Element-wise recovery) For each $j\in\SUPP{\vect{\beta}^{(\ell)}}$, 
$$
	\PROL{\hat{\beta}^{(\ell)}_j = \beta^{(\ell)}_j}\ge 1-p^*-\BIGO(1/K).
$$
\end{compactenum}
Moreover, the computational time of the Robust Mixed-Coloring algorithm is $\Theta(K\polylog(n))$.
\end{theorem} 

We can make similar remarks as in the noiseless case: 1) with high probability, the Robust Mixed-Coloring has no false discovery, 2) the algorithm can recover an arbitrarily large $(1-p^*)$ fraction of the supports, and 3) each element is recovered with probability asymptotically at least $1-p^*$. As we can see, the sample and time complexities of the Robust Mixed-Coloring algorithm are both $\Theta(K\polylog(n))$, and thus, when $K=\BIGO(n^\alpha)$ for some $\alpha\in(0,1)$, we can achieve \emph{sublinear} sample and time complexities in the ambient dimension $n$. We also note that, Assumption~\ref{asm:identifiability} is still needed in the noisy setting, i.e., the two parameter vectors differ at overlapping support. This assumption can still be mild in the sublinear regime where $K=o(n)$. In particular, if the supports of the two parameter vectors are independently drawn from certain distributions, then the probability that the two parameter vectors have overlapping supports \emph{vanishes} as $n$ approaches infinity. As for the dependence on $p^*$, we again note that when $p^*$ approaches $0$, the sample and time complexities grow slowly as $\log(1/p^*)$. Here, since we set $p^*$ as a constant, we hide this dependence in the big-O notation.

Similar to the noiseless case, by running the Robust Mixed-Coloring algorithm $\Theta(\log K)$ times, one can exactly recover the two parameter vectors with probability $1-\BIGO(1/\poly(K))$. In this case, the sample and time complexities are $\Theta(K\log(K)\polylog(n))$, and further, if we assume that $K=\Theta(n^\alpha)$ for some constant $\alpha$, we can still conclude that the sample and time complexities for full recovery are $\Theta(K\polylog(n))$.

\section{Related Work}\label{sec:related}
\subsection{Mixture of Regressions}
Parameter estimation using the expectation-maximization (EM) algorithm is studied empirically in~\cite{faria2010fitting}. In~\cite{stadler2010}, an $ \ell_1 $-penalized EM algorithm is proposed for the sparse setting. Theoretical analysis of the EM algorithm is difficult due to non-convexity. Progress was made in~\cite{yi2014alternating},~\cite{balakrishnan2014statistical}, and~\cite{yi2016solving} under stylized Gaussian settings with dense $ \vect{\beta} $, for which a sample complexity of $\Theta(n \polylog(n))$ is proved given a suitable initialization of EM. The algorithm uses a grid search initialization step to guarantee that the EM algorithm can find the global optimal solution, with the assumption that the query vectors are i.i.d. Gaussian distributed. The time complexity is polynomial in $n$. An alternative algorithm is proposed in~\cite{chen2014convex}, which achieves optimal $ \BIGO(n) $ sample complexity, but has high computational cost due to the use of semidefinite lifting. The algorithm in~\cite{chaganty2013spectral} makes use of tensor decomposing techniques, but suffers from a high sample complexity of $ \BIGO(n^6) $. In comparison, our approach has near-optimal sample and time complexities by utilizing the potential design freedom. The classification version of this problem has also been studied in~\cite{sun2014learning}.

\subsection{Coding-theoretic Methods and Group Testing} 
Many modern error-correcting codes such as LDPC codes and polar codes~\cite{arikan2009channel} with their roots in communication problems, exploit redundancy to achieve robustness, and use structural design to allow for fast decoding. These properties of codes have recently found applications in statistical problems, including graph sketching~\cite{li2015active}, sparse covariance estimation~\cite{pedarsani2015sparse}, low-rank approximation~\cite{ubaru2015low}, and discrete inference~\cite{ermon2014low}.
Most related to our approach is the work in~\cite{li2014sub,yincompressed,pedarsani2014phasecode,yin2015fast}, which apply sparse graph codes with peeling-style decoding algorithms to compressive sensing and phase retrieval problems. In our setting we need to handle a mixture distribution, which requires more sophisticated query design and novel demixing algorithms that go beyond the standard peeling-style decoding.

Another line of work relevant to our scheme is designing measurements in group testing~\cite{du2000combinatorial} via error correcting codes and expander graphs~\cite{d2000new,cheraghchi2010derandomization,lee2015saffron,mazumdar2016nonadaptive}. These results bear some similarities to our algorithm as they also exploit linear sketches of data for efficient sparsity pattern recovery. Our scheme differs from these works since we tackle problems in real and complex fields, whereas in group testing problems one consideres binary OR operations. In addition, we aim to solve the demixing problem in sparse recovery, and this is a more challenging task that has not been studied in the context of group testing.

\subsection{Combinatorial and Dimension Reduction Techniques} 
Our results demonstrate the power of strategic query and coding theoretic tools in mixture problems, and can be considered as efficient linear sketching of a mixture of sparse vectors. In this sense, our work is in line with recent works that make use of combinatorial and dimension reduction techniques in high-dimensional and large scale statistical problems. These techniques, such as locality-sensitive hashing~\cite{dhillon2011nearest}, sketching of convex optimization~\cite{pilanci2014iterative}, and coding-theoretic methods~\cite{achlioptas2015stochastic}, allow one to design highly efficient and robust algorithms applicable to computationally challenging datasets without compromising statistical accuracy.

\section{Mixed-Coloring Algorithm for Noiseless Recovery}\label{sec:algoutline}
In this section, we provide details of the Mixed-Coloring algorithm in the noiseless setting. We first provide some primitives that serve as important ingredients in the algorithm, and then describe the design of query vectors and decoding algorithm in detail.

\subsection{Primitives}\label{sec:primitive}
The algorithm makes use of four basic primitives: \textbf{summation check}, \textbf{indexing}, \textbf{guess-and-check}, and \textbf{peeling}, which are described below.

\textbf{Summation Check:} Suppose that we generate two query vectors $\vect{x}_1$ and $\vect{x}_2$ independently from some continuous distribution on $ \CMP^{n} $, and a third query vector of the form $\vect{x}_1+\vect{x}_2$. Let $y_1$, $y_2$, and $y_3$ be the corresponding measurements. We check the sum of the measurements and in the noiseless case, if $y_3=y_1+y_2$, then we know that these three measurements are generated from the same parameter vector $ \vect{\beta}^{(\ell)} $ almost surely. In this case we call $ \{y_1, y_2\} $ a \emph{consistent pair} of measurements as they are from the same $ \vect{\beta}^{(\ell)} $ (the third measurement $ y_3 $ is now redundant). 

\textbf{Indexing:}
The indexing procedure is to find the locations and values of the non-zero elements by carefully designed query vectors. In the noiseless case, this can be done by suitably designed \emph{ratio tests}. We sketch the idea of ratio test here. Consider a consistent pair of measurements $ \{y_1, y_2\} $ and corresponding query vectors $ \{\vect{x}_1, \vect{x}_2\} $. We design the query vectors such that the information of the locations of the non-zero elements is encoded in the relative phase between $y_1$ and $y_2$. In particular,  we generate  $n$ i.i.d. random variables $r_j, j\in[n]$ uniformly distributed on the unit circle. Letting $W=e^{\IMG\frac{2\pi}{n}}$ where $\IMG$ is the imaginary unit, we set the $ j $-th entries of $ \vect{x}_1 $ and $ \vect{x}_2 $ to be either $x_{1,j}=x_{2,j}=0$, or $x_{1,j}=r_j$ and $x_{2,j}=r_jW^{j-1}$. (The locations of the zeros are determined using sparse-graph codes and discussed later.)  
Below is an example of such a consistent pair of measurements and the corresponding linear system:
\begin{equation}\label{eq:example_ratio_test}
 \left[\begin{array}{c}
y_1\\y_2
\end{array}\right]
=
\left[\begin{array}{c}
\vect{x}_1\HET \\ \vect{x}_2\HET
\end{array}\right]
\vect{\beta}^{(1)} \\
= 
\left[\begin{array}{cccccccc}
0 & r_2 & r_3 & 0 & 0 & r_6 & 0 & 0 \\ 
0 & r_2W & r_3W^2 & 0 & 0 & r_6W^5 & 0 & 0 
\end{array}\right]
\vect{\beta}^{(1)}.
\end{equation}
Suppose that $ \vect{\beta}^{(1)} $ is $ 3 $-sparse and of the form $\vect{\beta}^{(1)} = [0~0~{*}~0~{*}~0~0~{*}]\TSP$. There is only one non-zero element, $\beta_3^{(1)}$, that contributes to the measurements $ y_1 $ and $ y_2 $ .
In this case the consistent measurement pair $ \{y_1, y_2\} $ is called a \emph{singleton}. 
A singleton can be detected by testing the integrality of the relative phase of the ratio $ y_1/y_2 $. 
In the above example, since $y_1=r_3\beta_3^{(1)}$ and $y_2=r_3W^2\beta_3^{(1)}$,  we observe that $\ABSL{y_1}=\ABSL{y_2}$ and the relative phase $\angle(y_2/y_1)=2\cdot  \frac{2\pi}{8}$ is an integral multiple of $\frac{2\pi}{8}$. We therefore know that with probability one, this consistent pair is a singleton, and moreover the corresponding non-zero element is located at the $ 3 $-rd coordinate with value $\beta_3^{(1)}=y_1/r_3$. In general, for a consistent measurement pair $ \{y_1, y_2\} $, if we observe that $\ABSL{y_1}=\ABSL{y_2}$ and the relative phase $\angle(y_2/y_1)=k\cdot  \frac{2\pi}{n}$ for some nonnegative integer $k$, then, we know that this consistent pair is a singleton, and the corresponding non-zero element is located at the $(k+1)$-th coordinate with value $y_1/r_{k+1}$. We would like to remark that the indexing step can also be done using real-valued query vectors.

\textbf{Guess-and-check and Peeling:}
After the ratio tests, we have already found some singletons, i.e., consistent pairs that are only associated with a single non-zero element. Ideally, we would like to iteratively reduce the problem by subtracting off recovered elements, in a Gaussian elimination-like manner, and find other non-zero elements. However, although we have recovered the locations and values of some non-zero elements, we still do not know their \emph{labels}, and the uncertainty in the labels brings additional difficulty to the problem. To resolve this issue, we use a guess-and-check strategy. In the example above, suppose instead that $ \vect{\beta}^{(1)} $ is $ 4 $-sparse, i.e., $\vect{\beta}^{(1)}=[0~{*}~{*}~0~{*}~0~0~{*}]\TSP$, in which case the consistent pair 
\begin{align}
y_i = x_{i,2} \beta^{(1)}_{2} + x_{i,3} \beta^{(1)}_{3}, \quad i = 1,2
\label{eq:doubleton}
\end{align}
is associated with two non-zero elements of $ \vect{\beta}^{(1)} $. 
Suppose that, in a previous iteration of the algorithm we have recovered the location and value of $ \beta^{(1)}_{2} $. At this point, we only know that this non-zero element is located at the second coordinate, and has value $\beta$ ($\beta=\beta^{(1)}_{2}$), but we do not know that this element belongs to the parameter vector $ \vect{\beta}^{(1)} $, nor do we know that the consistent pair $\{y_1, y_2\}$ is generated by $ \vect{\beta}^{(1)} $. Despite the uncertainty in the labels, we \emph{guess} that, this non-zero element belongs to the parameter vector that generates $\{y_1, y_2\}$. Then we can \emph{peel off} (i.e., subtract) this recovered element by
$$
y_i \leftarrow y_i - x_{i,2} \beta \quad i =1,2.
$$
The updated measurement pair satisfies $ y_i =  x_{i,3} \beta^{(1)}_{3}, i=1,2$. Then, we can \emph{check} whether our previous guess is correct, by doing ratio test on the updated pair. In this example, since the updated measurements are only associated with $\beta^{(1)}_{3}$ (i.e., this pair becomes a singleton), they can pass the ratio test. Then, we know that the peeling step is valid, and that the previous non-zero element at the second coordinate (with value $\beta$) and the newly recovered element at the third coordinate belong to the same parameter vector almost surely. If the updated measurement pair cannot pass the ratio test, there are two possibilities: 1) this pair is generated by some other parameter vector, or 2) this pair is associated with more than two non-zero elements. In this case, we keep both the measurement pairs before and after peeling for future usage. In general, the guess-and-check strategy and the peeling step can be combined to detect that two non-zero elements are from the same parameter vectors.

The continuing execution of these four primitives is made possible by the design of the query vectors using sparse-graph codes, which we describe next.

\subsection{Design of Query Vectors}\label{sec:design}

\begin{figure}[t]
    \centering
    \begin{subfigure}[b]{0.15\textwidth}
        \includegraphics[width=\textwidth]{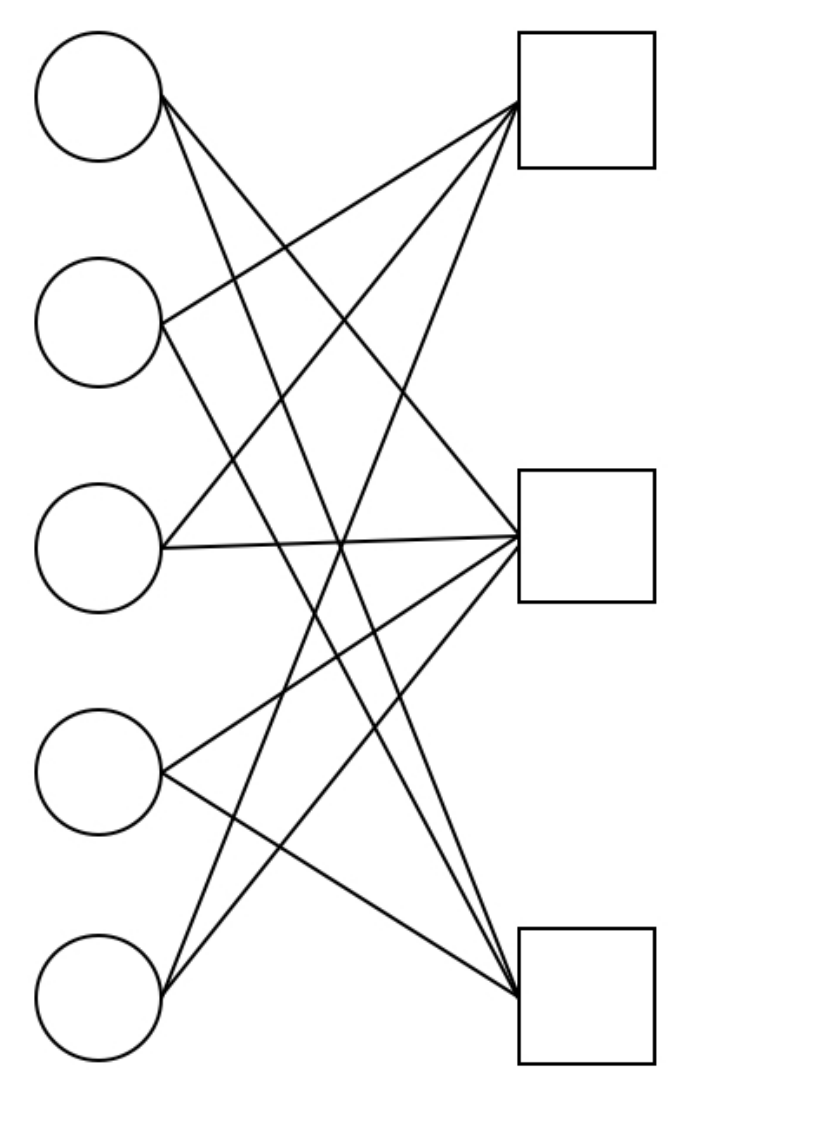}
        \caption{ }
        \label{fig:bipartite}
    \end{subfigure}
    \quad 
    \begin{subfigure}[b]{0.35\textwidth}
        \includegraphics[width=\textwidth]{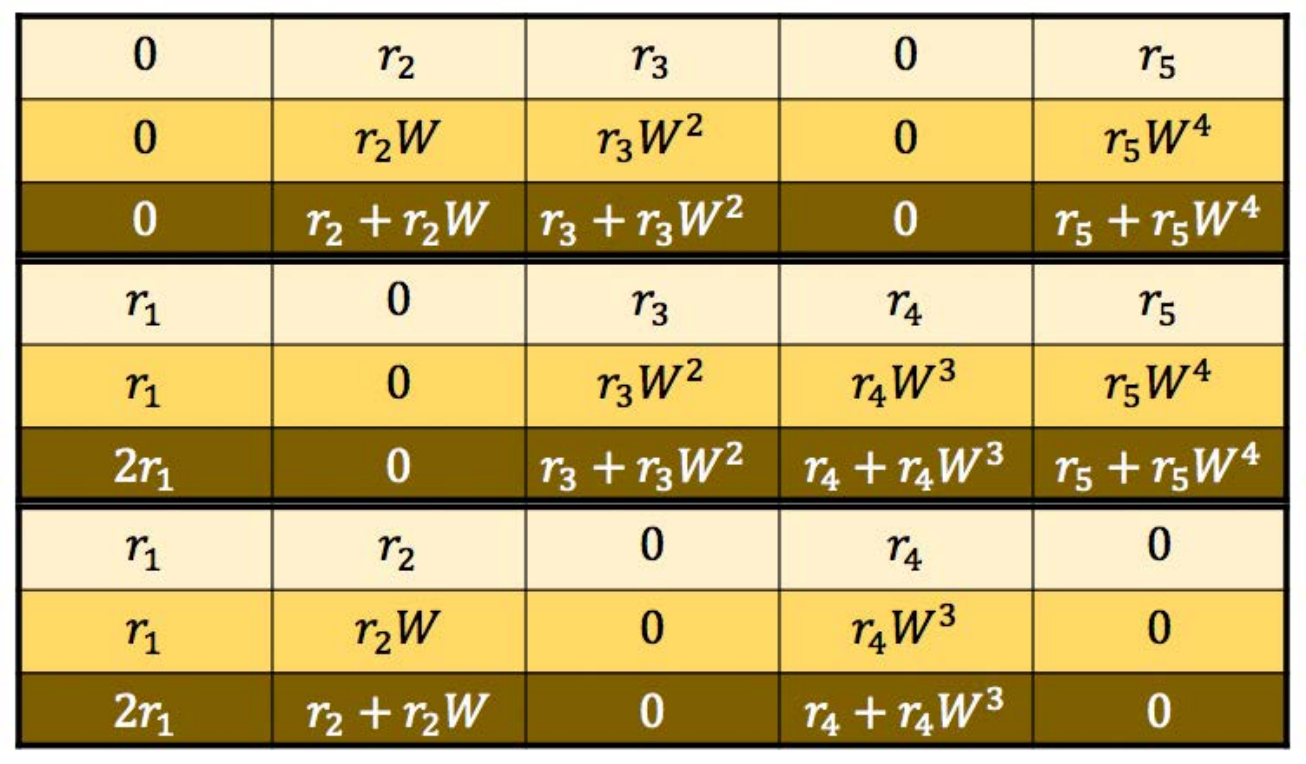}
        \caption{ }
        \label{fig:query_vectors}
    \end{subfigure}
    \caption{\small{Example of query vector design with $n=5$, $M=3$, $d=2$. (a) Bipartite graph with biadjacency given in~\eqref{eq:biadjacency}. (b) Design of query vectors in three bins. In the $i$-th bin, the query vectors are $\vect{r}_1\DIAG{\vect{h}_i}$, $\vect{r}_2\DIAG{\vect{h}_i}$, and $\vect{r}_3\DIAG{\vect{h}_i}$.}}\label{fig:design_example}
\end{figure}

As illustrated in Figure~\ref{fig:query}, we construct $M=\Theta(K)$ sets of query vectors (called \emph{bins}). The query vectors in each bin are associated with some coordinates of the parameter vectors (i.e., the query vectors are non-zero only on those coordinates). The association between the coordinates and bins is determined by a $d$-left regular bipartite graph with $n$ left nodes (coordinates) and $M$ right nodes (bins), where each left node is connected to $d=\Theta(1)$ right nodes chosen independently uniformly at random. Here, we note that other designs of the bipartite graph may also be employed, such as expander graphs~\cite{sipser1996expander,jafarpour2009efficient}. As we see in later sections, as long as the bipartite graph structures allow for a density evolution analysis~\cite{richardson2008modern}, we may be able to use such graphs. In this paper, we choose to use $d$-left regular bipartite graph since it is amenable to a transparent analysis and already achieves order-optimal sample and time complexities in the noiseless setting. In our design, each bin consists of three query vectors. The values of the non-zero elements of the first two query vectors are in the form of~\eqref{eq:example_ratio_test}, enabling the ratio test. The third query vectors equals the sum of the first two and is used for the summation check. 

More precisely, we first design three random vectors $\vect{r}_1, \vect{r}_2, \vect{r}_3\in\CMP^n$, where $\vect{r}_1=[r_1,r_2,\ldots,r_n]\TSP$ consists of elements that are i.i.d. uniformly distributed on the unit circle $\{z:|z|=1\}$, and $\vect{r}_2$ is a vector with elements of $\vect{r}_1$ being modulated by Fourier coefficients $W^j$, $j = 0, \ldots, n-1$, $W=e^{\IMG\frac{2\pi}{n}}$, and $\vect{r}_3 = \vect{r}_1 + \vect{r}_2$. Let $\mat{H}\in\{0,1\}^{M\times n}$ be the biadjacency matrix of the bipartite graph, and $\vect{h}_i\TSP$ be the $i$-th row of $\mat{H}$. Then, the query vectors of the $i$-th bin is $\vect{r}_1\DIAG{\vect{h}_i}$, $\vect{r}_2\DIAG{\vect{h}_i}$, and $\vect{r}_3\DIAG{\vect{h}_i}$. We provide a concrete example with $n=5$, $M=3$, $d=2$ in Figure~\ref{fig:design_example}. The biadjacency matrix is given in~\eqref{eq:biadjacency}.
\begin{equation}\label{eq:biadjacency}
\mat{H} = \left[\begin{array}{ccccc}
0 & 1 & 1 & 0 & 1\\
1 & 0 & 1 & 1 & 1\\
1 & 1& 0 & 1 & 0
\end{array}\right].
\end{equation}

If the query vectors in each bin were used only once, then we would have very few bins passing the summation check and hence few consistent pairs. Instead, we use the first two query vectors repeatedly for $R=\Theta(1)$ times, obtaining two sets of measurements, each of size $ R $ and called \emph{type-I} and \emph{type-II index measurements}. We use the third query vector $V=\Theta(1)$ times to obtain a set of \emph{verification measurements}. We therefore have $2R+V$ measurements associated with each of the $ M $ bins, hence a total of $m=(2R+V)M=\Theta(K)$ measurements, as shown in Figure~\ref{fig:query}.

\begin{figure}[t!]
\centering
\includegraphics[width=0.30\textwidth]{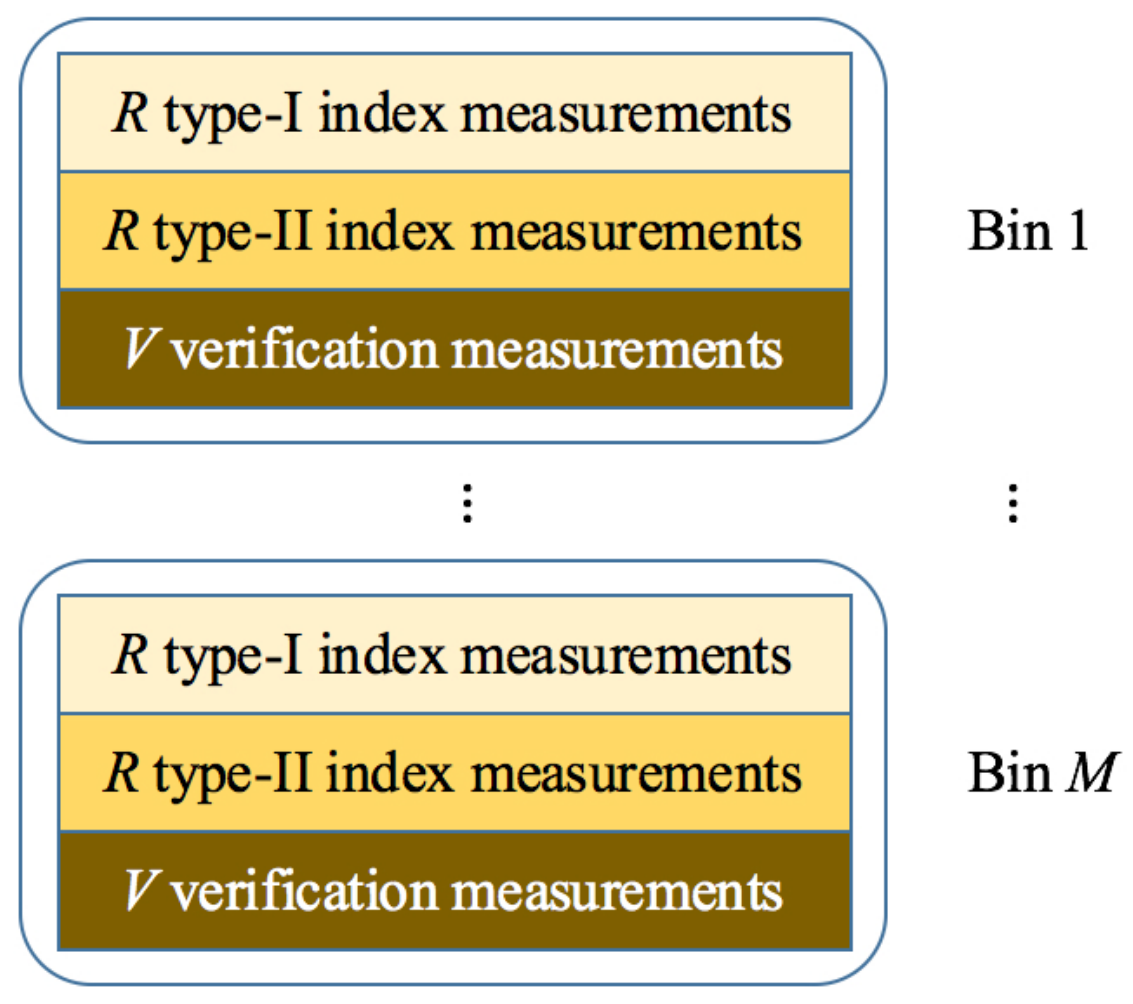}
\caption{ \small{Design of query vectors in noiseless setting. We design $M=\Theta(K)$ bins, and in each bin, we design three query vectors: two for indexing and one for summation check. Each indexing query vector is repeated $R=\Theta(1)$ times, and the verification query vector is repeated $V=\Theta(1)$ times. Thus, the total number of measurements is $ (2R+V)M = \Theta(K)$.} }\label{fig:query}
\end{figure}

\begin{algorithm}[t]
\caption{Mixed-Coloring}\label{alg:mixedcoloring}
\begin{algorithmic} 
\REQUIRE Query vectors $\vect{x}_i$, measurements $y_i$, $i\in[m]$, number of parameter vectors $L$.
\ENSURE Estimates of parameter vectors $\vect{\beta}^{(\ell)}$, $\ell\in[L]$.
\STATE Find consistent pairs via summation check.
\STATE Find singleton balls via ratio test.
\STATE Construct a graph $\mathcal{G}$ with nodes representing all the singleton balls.
\STATE Identify the consistent pairs that contain two singleton balls. Add an edge between the two nodes in $\mathcal{G}$.
\STATE Find $L$ largest connected components in $\mathcal{G}$. Recover a fraction of non-zero elements in each $\vect{\beta}^{(\ell)}$.
\STATE Iteratively find the rest of the non-zero elements in each $\vect{\beta}^{(\ell)}$ via peeling.
\end{algorithmic}
\end{algorithm}

\subsection{Decoding Algorithm}\label{sec:decode}
We provide an outline of the decoding algorithm in Algorithm~\ref{alg:mixedcoloring}. The decoding algorithm first finds consistent pairs (by summation check) in each bin, within which singletons are identified (by the ratio test). The ratio test also recovers the locations and values of several non-zero elements, some of which can then be associated with the same $ \vect{\beta}^{(\ell)} $ by guess-and-check. Using tools from random graph theory, we can separate part of the recovered non-zero elements from different parameter vectors. At this point, for each $ \vect{\beta}^{(\ell)} $, we have recovered some of its non-zero elements (including their locations, values and labels). We iteratively conduct a combined operation of guess-and-check and peeling, so that we can subtract the recovered elements from the remaining consistent pairs, until no more non-zero elements can be found. Below we elaborate on these steps.

\textbf{Finding Consistent Pairs:} 
The decoding procedure starts by finding all the consistent pairs. In each bin, we perform summation checks on all triplets $(y_1,y_2,y_3)$ in which $y_1$, $y_2$, and $y_3$ are the type-I index measurement, type-II index measurement and verification measurement, respectively. If a triplet passes the summation check, then a consistent pair $\{y_1,y_2\}$ is found. Note that in each bin the number of triplets of the above form is a constant, so this step can be done in $\Theta(K)$ time. The subsequent steps of the algorithm are based on the consistent pairs found in this step. We also note that, since for every $\ell$, the probability that each measurement is generated by the parameter vector $\vect{\beta}^{(\ell)}$ is a constant $q_\ell$, the probability that one can find a consistent pair in a particular bin is a constant.

We classify the consistent pairs into a few different types. As we have seen, each consistent pair is only associated with a subset of the non-zero elements of a particular parameter vector, due to the design of the bipartite graph. As before, a consistent pair associated with only one non-zero element is called a \emph{singleton}, and we call this non-zero element a \emph{singleton ball}. The consistent pairs associated with two non-zero elements are called \emph{doubletons}; and those associated with more than one non-zero elements are called \emph{multitons}\footnote{Doubletons are also multitons.}. These terminologies are useful for our following discussions.

\textbf{Finding Singletons:}
Each non-zero element of the parameter vectors can be identified by its label-location-value triplet  $ (\ell, j, \beta^{(\ell)}_j) $. We visualize these triplets (i.e., non-zero elements) as balls, as shown in Figure~\ref{subfig:nonzero}, and initially their labels, locations and values are unknown\footnote{Note that the graph in Figure~\ref{fig:algorithm} differs from the bipartite graph that we use to design the query vectors.}. We run the ratio test on the consistent pairs to identify singletons and their associated singleton balls. The singleton balls found are illustrated in Figure~\ref{subfig:singleton} as shaded balls. The ratio test also recovers the locations and values of these singleton balls, although at this point we do not know the label $ \ell $ of the balls. 

To better understand the algorithm, here we analyze the expected number of singleton balls that belong to parameter vector $\vect{\beta}^{(\ell)}$. We show that a constant fraction of the non-zero elements in $\vect{\beta}^{(\ell)}$ can be found as singleton balls in this stage. First, we analyze the probability $Q_{\ell}$ that a particular bin can produce a consistent pair that is generated by $\vect{\beta}^{(\ell)}$. According to our probabilistic model, the measurements are generated independently, and therefore, we have
$$
Q_{\ell} = [1-(1-q_{\ell})^V][1-(1-q_{\ell})^R]^2.
$$
Denote by $\xi_{k}^{(\ell)}$ the probability of the event that a particular bin produces a consistent pair that is generated by $\vect{\beta}^{(\ell)}$, and is associated with $k$ non-zero elements in $\vect{\beta}^{(\ell)}$. Since each non-zero element is associated with $d$ bins among the $M$ bins independently and uniformly at random, for a consistent pair generated by $\vect{\beta}^{(\ell)}$, the number of non-zero elements associated with this pair is binomial distributed with parameters $K_\ell$ (recall that $K_\ell$ is the number of non-zero elements in $\vect{\beta}^{(\ell)}$) and $\frac{d}{M}$, and we have
$$
\xi_{k}^{(\ell)} = Q_{\ell} {K_{\ell}\choose k} \left(\frac{d}{M}\right)^k \left(1-\frac{d}{M}\right)^{K_{\ell}-k}.
$$
In addition, we can use Poisson distribution to approximate the binomial distribution when $\lambda_{\ell}:=\frac{K_{\ell}d}{M}$ is a constant and $K_{\ell}$ approaches infinity, i.e., we have
$$
\xi_{k}^{(\ell)} \approx Q_{\ell} \frac{\lambda_{\ell}^ke^{-\lambda_{\ell}}}{k!}.
$$
Let us ignore the zero elements in $\vect{\beta}^{(\ell)}$ and consider the bipartite graph representing the association between the $K_\ell$ non-zero elements (left notes) in $\vect{\beta}^{(\ell)}$ and the $M$ bins (right nodes). We know that the total number of edges in this bipartite graph is $K_{\ell}d$, and we denote by $\rho_{k}^{(\ell)}$ the expected fraction of the edges that are connected to a right node (bin) with degree $k$. Thus, we have
$$
\rho_{k}^{(\ell)} = \frac{kM}{K_{\ell}d}\xi_k^{(\ell)}=Q_{\ell}\frac{\lambda_{\ell}^{k-1}e^{-\lambda_{\ell}}}{(k-1)!}.
$$
We then proceed to analyze the expected fraction of the singleton balls. Let $q_s^{(\ell)}$ be the probability that a non-zero element in $\vect{\beta}^{(\ell)}$ becomes a singleton ball in a certain consistent pair. The event is equivalent to the event that at least one of its $d$ associated right nodes (bins) has degree $1$. Then, when $K_\ell$ approaches infinity, we have asymptotically
$$
q_s^{(\ell)} = 1-(1-\rho_1^{(\ell)})^d = \Theta(1).
$$
Thus, the expected number of non-zero elements in $\vect{\beta}^{(\ell)}$ that are found as singleton balls in this stage of the algorithm is $q_s^{(\ell)}K_{\ell} = \Theta(K)$, and this implies that we can recover a \emph{constant} fraction of the non-zero elements in each parameter vector, without knowing their labels. We can further prove high probability bounds for this fraction, and more details of this analysis are relegated to Lemma~\ref{lem:single} in Appendix~\ref{sec:prf_noiseless}.

\textbf{Recovering a Subset of Non-zero Elements:} 
The next step is crucial: for two singleton balls and a consistent pair associated with the locations of these two balls, we run the guess-and-check and peeling operations to detect if these two singleton balls indeed have the same label (or equivalently, the two non-zero elements are in the same parameter vector). If so, we call this consistent pair a \emph{strong doubleton} (i.e., doubletons that contain two singleton balls that we find in the previous stage of the algorithm), and connect these two balls with an edge, as shown in Figure~\ref{subfig:singleton}. Doing so creates a graph over the balls (i.e., non-zero elements), and each connected component of the graph is from a single parameter vector. Since each non-zero element is associated with a constant number of consistent pairs (due to using a $d$-left regular bipartite graph with constant $d$), this step can in fact be done efficiently in $\Theta(K)$ time without enumerating all the combinations of singleton ball pairs. Similar to the analysis of singleton balls, we can analyze the number of strong doubletons. In fact, we can show that with high probability, a \emph{constant} fraction of the consistent pairs are strong doubletons. More details are relegated to Lemma~\ref{lem:edge} in Appendix~\ref{sec:prf_noiseless}.

By carefully choosing the parameters\footnote{To make our main sections concise, here we omit the condition on the design parameters $d$, $M$, $R$, and $V$ in order to form giant components with sizes $\Theta(K)$; the precise statement of this condition is relegated to Lemma~\ref{lem:giant} in Appendix~\ref{sec:prf_noiseless}.} $d$, $M$, $R$, and $V$, and using tools from random graph theory, we can ensure that with high probability the $ L $ \emph{largest} connected components (called \emph{giant components}) correspond to the $ L $ parameter vectors, and each of these components has size $\Theta(K)$. Then, the labels of the balls in these components are identified. This is illustrated in Figure~\ref{subfig:giant} for $ L=2 $, where colors represent the labels. More details of this demixing process are provided in Lemma~\ref{lem:giant} in Appendix~\ref{sec:prf_noiseless}. In summary, at this point we have recovered the labels, locations and values of a \emph{constant} fraction of the non-zero elements (i.e., balls) of each parameter vector.

\textbf{Iterative Decoding:}
The decoding procedure proceeds by identifying the labels of the remaining balls via iteratively applying the guess-and-check and peeling primitives. The connected components in Figure~\ref{subfig:giant} are therefore expanded, until no more changes can be made, as illustrated in Figure~\ref{subfig:colored}.

We provide an example of this iterative  procedure in Figure~\ref{fig:simpleexample2}. Recall that the association between the coordinates of the parameter vectors and the bins (or consistent pairs) is determined by a bipartite graph. Here, we only show one consistent pair for each bin and omit the zero elements. The non-zero elements and the consistent pairs are shown as balls and squares, respectively, as in Figure~\ref{fig:subfigure21}. The steps described in the last part recover a subset of these balls, which are shown in red and blue in Figure~\ref{fig:subfigure22}. For simplicity, let us call the corresponding $\vect{\beta}^{(\ell)}$'s red and blue parameter vectors, respectively. If a consistent pair is generated by the red (blue) parameter vector, we say that this consistent pair is red (blue). Now consider the consistent pair 1, which is associated with the coordinates that the balls $ a  $, $ b $ and $ c $ are located. Although we do \emph{not} know whether this pair is red or blue, we can \emph{guess} that this pair is blue, and peel the blue balls $ a $ and $ b $ off from consistent pair 1. Since this consistent pair is indeed blue, the updated measurements can pass the ratio test, and thus we can recover the label, location and value of the non-zero element represented by blue ball $ c $. Similarly, by guessing that the consistent pair 3 is red and peeling off the recovered red ball $ v $ from the consistent pair $ 3 $, we can recover red ball $ w $, as illustrated in Figure~\ref{fig:subfigure23}. We continue this process iteratively, guessing the color of the consistent pairs and peeling off balls recovered in the previous iterations to recover more balls. For example, we peel off balls $ b $ and $ c $ from the measurement pair $ 2 $ to recover ball $ d $, and ball $ w $ from pair $  4$ to recover ball $ z $, resulting in Figure~\ref{fig:subfigure24}.

\subsection{Choice of Design Parameters}

We have completed the description of the Mixed-Coloring algorithm for the noiseless case. The algorithm involves several design parameters including $ d $, $ R $, $ V $, and $ M $. We need to choose these parameters properly in order to guarantee successful decoding. The precise conditions that these parameters need to satisfy are somewhat technical, and they are relegated to Lemmas~\ref{lem:giant} and~\ref{lem:errorfloor} in Appendix~\ref{sec:prf_noiseless}. More specifically, Lemma~\ref{lem:giant} provides the condition that the $ L $ giant components in Figure~\ref{subfig:giant} correspond to the $ L $ parameter vectors, and each of these components has size $\Theta(K)$; and Lemma~\ref{lem:errorfloor} provides the condition that the iterative decoding process can find an arbitrarily large fraction of the non-zero elements, and the analysis is based on density evolution from modern coding theory~\cite{RU01}. In addition, Lemma~\ref{lem:concentration} in Appendix~\ref{sec:prf_noiseless} provides high probability bound on the recovered fraction of non-zero elements. For concrete settings, the optimal choices of these parameters can actually be computed numerically via the density evolution analysis. In particular, for \emph{any} upper bound $p^* \in (0,1) $ of the error fraction, we can find the proper values of $ d $, $ R $, $ V $, and $ M $ for which the peeling process is guaranteed to proceed successfully and recover all but a fraction $ p^* $ of the non-zero elements. As examples, in Table~\ref{tab:design_para} we list the optimal values of the design parameters for $L=2,3$ or $4$ parameter vectors with equal probability (mixing weights), i.e., $q_\ell = \frac{1}{L}$. The details of these numerical computations are provided in Appendix~\ref{sec:constant}.

\begin{table}[h]
\caption{Design parameters of Mixed-Coloring algorithm}
\begin{center}
  \begin{tabular}{ c | c | c | c | c }
    \hline
    Parameter & Description & $L=2$ & $L=3$ & $L=4$ \\ \hline
    $p^*$  &  error fraction & $5.1\times 10^{-6}$ & $8.8\times 10^{-6}$ & $8.1\times 10^{-6}$ \\ \hline
    $d$ & left degree of bipartite graph & $15$ & $15$ & $13$  \\ \hline 
    $R$ & number of type-I / type-II index measurements in each bin & $3$ & $5$ & $8$ \\ \hline
    $V$ & number of verification measurements in each bin & $3$ & $5$ & $8$ \\  \hline
    $M$ & number of bins & $3.71K$ & $2.52K$ & $1.68K$ \\ \hline
    $m$ & total number of measurements, $m=(2R+V)M$ & $33.39K$ & $37.80K$ & $40.32K$ \\ \hline
  \end{tabular} 
\end{center}
\label{tab:design_para}
\end{table}

\begin{figure}
    \centering
    \begin{subfigure}[b]{0.15\textwidth}
        \includegraphics[width=\textwidth]{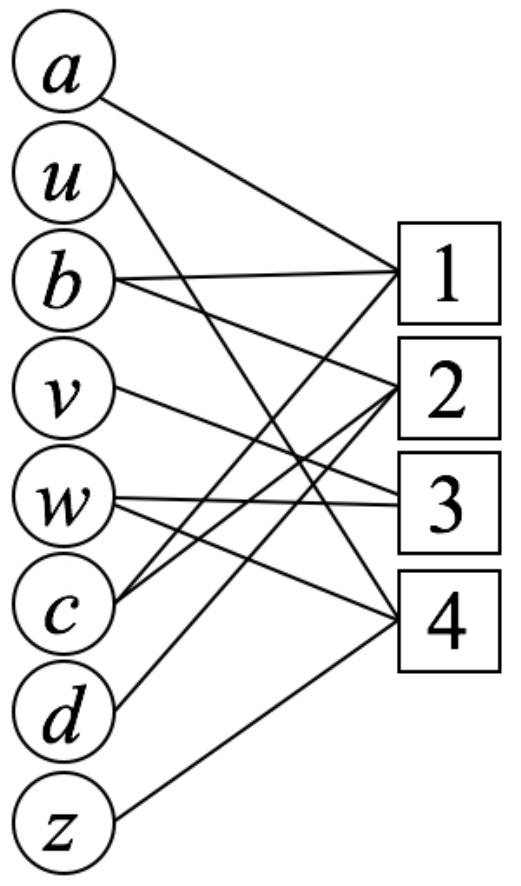}
        \caption{ }
        \label{fig:subfigure21}
    \end{subfigure}
    \quad 
    \begin{subfigure}[b]{0.15\textwidth}
        \includegraphics[width=\textwidth]{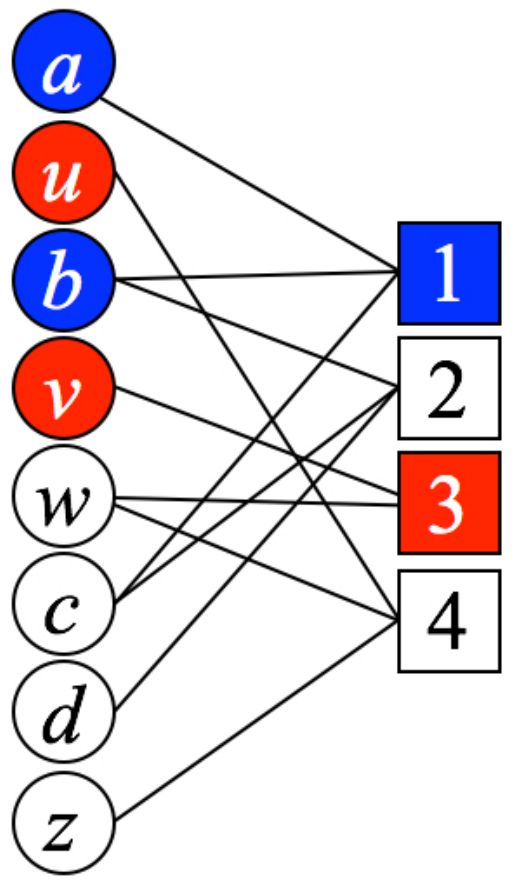}
        \caption{ }
        \label{fig:subfigure22}
    \end{subfigure}
    \quad
    \begin{subfigure}[b]{0.15\textwidth}
        \includegraphics[width=\textwidth]{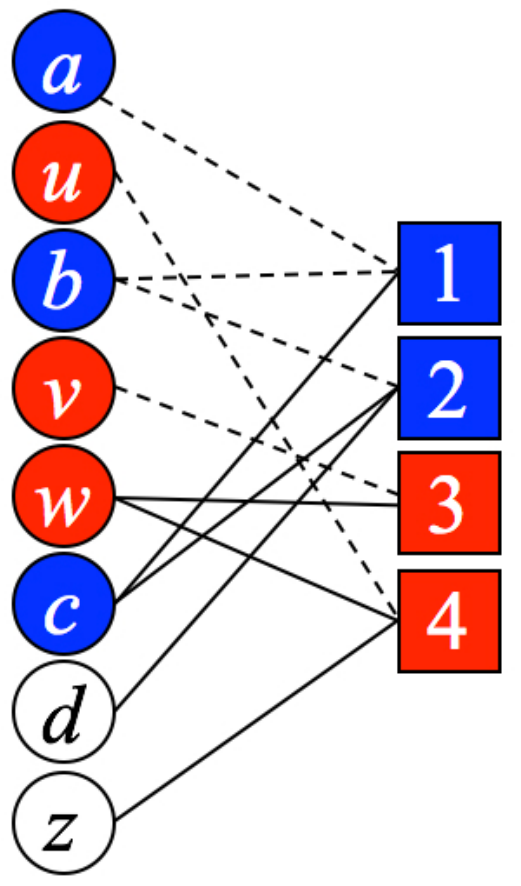}
        \caption{ }
        \label{fig:subfigure23}
    \end{subfigure}
    \quad 
    \begin{subfigure}[b]{0.15\textwidth}
        \includegraphics[width=\textwidth]{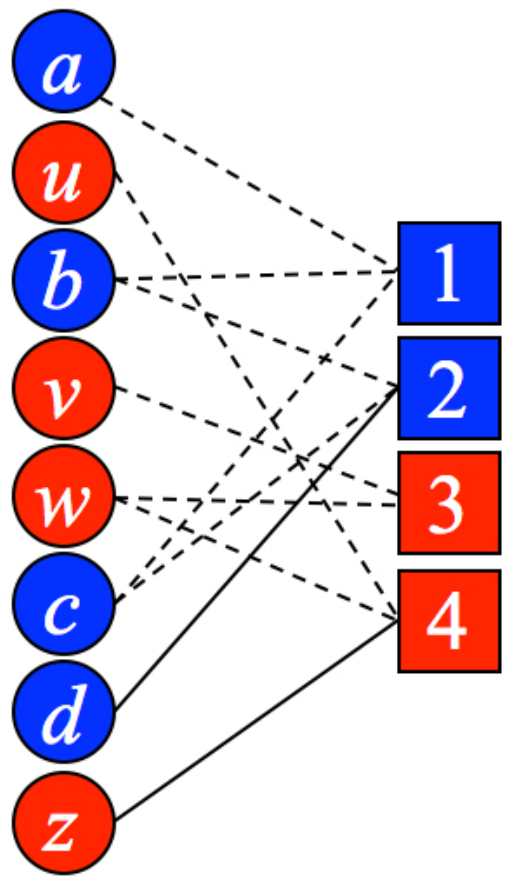}
        \caption{ }
        \label{fig:subfigure24}
    \end{subfigure}
    \caption{\small{Iterative decoding. In each of the bipartite graph, the left nodes (balls) denote the non-zero elements in the two parameter vectors, and the right nodes (squares) denote the consistent pairs. The balls and consistent pairs associated with different parameter vectors are shown in different colors. Here, elements $a$, $b$, $c$, $d$ and consistent pairs $1$, $2$ are associated with $\vect{\beta}^{(1)}$, and elements $u$, $v$, $w$, $z$ and consistent pairs $3$, $4$ are associated with $\vect{\beta}^{(2)}$. If a ball is peeled off, the edges connected to it are shown in dashed lines. (a) The connection between balls and consistent pairs. (b) Using the giant component method, we find balls $a$ and $b$ are in the same color (i.e., in the same parameter vector), and balls $u$ and $v$ are in the other color. (c) Using the guess-and-check approach, we peel $a$ and $b$ from their connected pairs and recover $c$ via ratio rest; similarly, we peel $u$ and $v$ from their connected pairs and recover $w$. (d) Continue the process. Peel $c$ and $w$ from their associated pairs and recover all the non-zero elements.}}\label{fig:simpleexample2}
\end{figure}

\section{Robust Mixed-Coloring Algorithm for Noisy Recovery}\label{sec:noise}
The key idea of Robust Mixed-Coloring algorithm is to turn the noisy problem to a noiseless one. We keep the overall structure of the Robust Mixed-Coloring algorithm the same as its noiseless counterpart. We still use a balls-and-bins model to design the query vectors. In particular, we keep using a $d$-left regular bipartite graph to represent the association between coordinates and bins (sets of measurements), and the algorithm still proceeds as shown in Algorithm~\ref{alg:mixedcoloring}. However, the steps in Mixed-Coloring algorithm that rely on the fact that there is no noise in the measurements should be robustified. In particular, in the presence of noise, the ratio test method for indexing (finding the location and value of non-zero elements) and the summation check for finding consistent measurements (measurements that are generated by the same parameter vector) need to be modified. To this end, we use a new design of query vectors, and employ an EM-based noise reduction scheme to effectively obtain the noiseless measurements of these query vectors. We provide more details of the Robust Mixed-Coloring algorithm in the following.

\subsection{Design of Query Vectors}
We keep the high-level design of query vectors as in the noiseless setting. This means that we still use a $d$-left regular bipartite graph with $n$ left nodes and $M$ right nodes to represent the association between the coordinates and the bins. However, we change the design of query vectors within each bin, and in particular, we design three types of query vectors. The first type, called \emph{binary indexing} vectors, encodes the location information using binary representations with $\lceil \log_2(n) \rceil$ bits (as opposed to using the relative phases in the noiseless case). The second type is called \emph{verification} vectors, which are used to verify the singleton balls (or equivalently, non-zero elements) found by the binary indexing vectors. We robustify the indexing process by replacing the ratio test query vectors with these two types of query vectors. A similar approach is considered in~\cite{yin2015fast} for compressive phase retrieval. The third type of query vectors is used for \emph{consecutive summation check}, which finds \emph{consistent sets} of measurements, and robustifies the summation check step in the noiseless case.

We now provide details of the design. Let $\mat{H}$ denote the biadjacency matrix of the bipartite graph. For a particular bin (we omit the label of the bin for simplicity), let $\vect{h}\in\{0,1\}^n$ denote the association between this bin and the coordinates. We design $P=\Theta(\log^2(n))$ query vectors $\vect{x}_i\in\REAL^{n}$, $i\in[P]$ for this bin as follows:
$$
[\vect{x}_1~\cdots~\vect{x}_P]\TSP = [\mat{B}\TSP~\mat{V}\TSP~\mat{C}\TSP]\TSP\DIAG{\vect{h}}.
$$
where $\mat{B}\in\{0,1\}^{P_1\times n}$, $\mat{V}\in\{1,-1\}^{P_2\times n}$, and $\mat{C}\in\mathbb{Z}^{P_3\times n}$ correspond to the three types of new query vectors, meaning that they are used for binary indexing, verification, and consecutive summation check, respectively. The matrix $\mat{B}$ has $P_1=\lceil \log_2(n) \rceil$ rows, and the $i$-th column of $\mat{B}$ is the binary representation of integer $i-1$. The matrix $\mat{V}$ has $P_2=\Theta(\log(n))$ rows and consists of i.i.d. Rademacher entries, i.e., the entries of $\mat{V}$ are equally likely to be either $1$ or $-1$. The matrix $\mat{C}$ contains $P_3={P_1+P_2\choose 2}$ rows, and the rows of $\mat{C}$ are indexed by pairs $(j,k)$, $1\le j<k\le P_1+P_2$. Let $\mat{D}=[\mat{B}\TSP~\mat{V}\TSP]\TSP$ be a collection of the first two matrices. The row of $\mat{C}$ indexed by $(j,k)$ (denoted by $\vect{c}_{(j,k)}\TSP$) is the summation of the $j$-th and the $k$-th row of $\mat{D}$, i.e., $\vect{c}_{(j,k)}\TSP=\vect{d}_j\TSP+\vect{d}_k\TSP$. Here, we give a simple example with $n=4$, $P_1 = 2$, $P_2=2$, and $P_3 = 6$ in Figure~\ref{fig:example_noisy}.

\begin{figure}[h]
\centering
\includegraphics[width=0.6\textwidth]{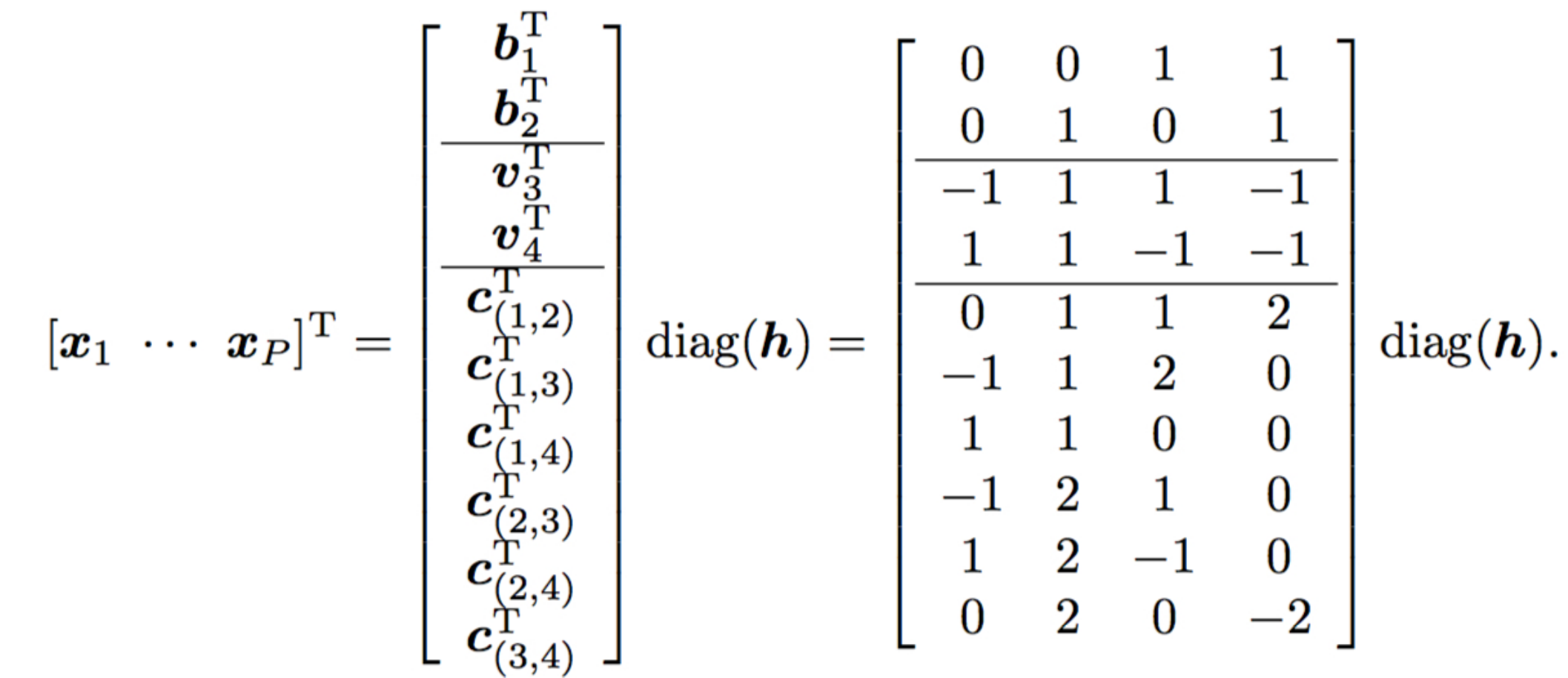}
\caption{\small{Example of query vectors in noisy setting with $n=4$, $P_1 = 2$, $P_2=2$, and $P_3 = 6$. Each row represents a query vector (in most cases we use column vectors, but here we use row vectors for ease of description). Here, $\vect{b}_1\TSP\DIAG{\vect{h}}$ and $\vect{b}_2\TSP\DIAG{\vect{h}}$ are used for binary indexing; $\vect{v}_3\TSP\DIAG{\vect{h}}$ and $\vect{v}_4\TSP\DIAG{\vect{h}}$ are used for verification; and $\vect{c}_{(1,2)}\TSP\DIAG{\vect{h}}, \ldots, \vect{c}_{(3,4)}\TSP\DIAG{\vect{h}}$ are used for consecutive summation check. As we can see, the vectors $\vect{c}_{(j,k)}$ are summations of all the pairs in the first two sets vectors, i.e., $\vect{c}_{(1,2)} = \vect{b}_{1} + \vect{b}_{2}, \vect{c}_{(1,3)} = \vect{b}_{1} + \vect{v}_{3}, \vect{c}_{(1,4)} = \vect{b}_{1} + \vect{v}_{4}, \vect{c}_{(2,3)} = \vect{b}_{2} + \vect{v}_{3}, \vect{c}_{(2,4)} = \vect{b}_{2} + \vect{v}_{4}, \vect{c}_{(3,4)} = \vect{v}_{3} + \vect{v}_{4}$.}}
\label{fig:example_noisy}
\end{figure}

\subsection{Decoding Algorithm}\label{sec:noisy_decoding_alg}
We now describe the decoding part of the Robust Mixed-Coloring algorithm. As mentioned, we use an EM-based noise reduction scheme to find the noiseless measurements of the query vectors, and also conduct robustified summation check and indexing process. Other parts of the algorithm are kept the same. We elaborate the details in the following.

\textbf{Noise Reduction:}
Due to the presence of noise, the first step that we need to take is a noise reduction operation. More specifically, we use each query vector $N=\Theta(\polylog(n))$ times, repeatedly. According to our model, in the presence of Gaussian noise, one can see that if $\vect{x}_i\TSP\vect{\beta}^{(1)} = \vect{x}_i\TSP\vect{\beta}^{(2)}$, the $N$ measurements are i.i.d. Gaussian distributed; otherwise the $N$ measurements are independently distributed as a mixture of two equally weighted Gaussian random variables. Therefore, the problem becomes a standard estimation problem for a one dimensional Gaussian mixture distribution. We propose an EM algorithm with an initialization step using method of moments to estimate the two centers of the mixture. The performance of our proposed EM algorithm can be characterized by Theorem~\ref{thm:em}, proved in Appendix~\ref{sec:em}.
\begin{theorem}\label{thm:em}
Suppose that $\Delta/\sigma \ge \frac{4}{\sqrt{3}}$. Then, by using $N=\Theta(\polylog(n))$ measurements, the proposed EM algorithm, with initialization via method of moments, can recover the exact value of the two centers\footnote{Note that in our problem, the centers take quantized values and the quantization step is known to the decoder, so the estimation can take the exact value of the true centers.} of the mixture of Gaussian distributions with probability at least $1-\BIGO(1/{\rm poly}(n))$.
\end{theorem}

In addition, we can see that since each query vector in each bin is repeated $N=\Theta(\polylog(n))$ times, and there is $P=\Theta(\log^2(n))$ query vectors in each bin, the total number of measurements we get for each bin is $NP=\Theta(\polylog(n))$. Since there are $\Theta(K)$ bins, the total number of measurements of the Robust Mixed-Coloring algorithm is $\Theta(K\polylog(n))$.

\textbf{Consecutive Summation Check:}
After the noise reduction operations, for each query vector $\vect{x}_i$, $i\in[P]$, we get at most two ``centers'' $\{y_{i,1},y_{i,2}\}$ (called \emph{denoised measurements}), which correspond to the inner product of the query vector and the parameter vectors in the noiseless case. However, we do not know the correspondence between the denoised measurements and the two parameter vectors. This means that we can have either $(y_{i,1},y_{i,2})=(\vect{x}_i\TSP\vect{\beta}^{(1)}, \vect{x}_i\TSP\vect{\beta}^{(2)})$ or $(y_{i,1},y_{i,2})=(\vect{x}_i\TSP\vect{\beta}^{(2)}, \vect{x}_i\TSP\vect{\beta}^{(1)})$. Therefore, we need to use the consecutive summation check method to find the denoised measurements which are generated by the same parameter vector.

We illustrate the consecutive summation check process using a simple example in Figure~\ref{fig:summation}. Assume that we have three query vectors $\vect{x}_1$, $\vect{x}_2$, $\vect{x}_3$, and two summation check query vectors $\vect{x}_1+\vect{x}_2$ and $\vect{x}_2+\vect{x}_3$. Suppose that the denoised measurements that we get for $\vect{x}_i$, $i=1,2,3$ are $(y_{1,1},y_{1,2})=(1,5)$, $(y_{2,1},y_{2,2})=(2,4)$, and $(y_{3,1},y_{3,2})=(2,3)$, and the denoised measurements for the summation check query vectors are $(y_{(1,2),1}, y_{(1,2), 2})=(5,7)$ and $(y_{(2,3),1}, y_{(2,3),2})=(5,6)$. By \emph{matching} summations, one can easily find that the only possible case that we can observe these denoised measurements is that $(y_{1,1}, y_{2,2}, y_{3,1})$ and $(y_{1,2}, y_{2,1}, y_{3,2})$ are generated by the same parameter vector (we call them \emph{consistent sets}), respectively, as shown in different colors in Figure~\ref{fig:summation}. In our algorithm, we need to conduct consecutive summation check on all the denoised indexing and verification measurements, using the denoised summation check measurements. We also mention that the reason that we need summations of all the ${P_1+P_2\choose 2}$ pairs of the first $P_1+P_2$ query vectors is that we might have the two denoised measurements taking the same value, i.e., $\vect{x}_i\TSP\vect{\beta}^{(1)} = \vect{x}_i\TSP\vect{\beta}^{(2)}$, and then we have to conduct summation check on two query vectors which are not adjacent. We provide the precise procedures of consecutive summation check in Algorithm~\ref{alg:consecutive}.

\begin{figure}[h]
\centering
\includegraphics[width=0.4\textwidth]{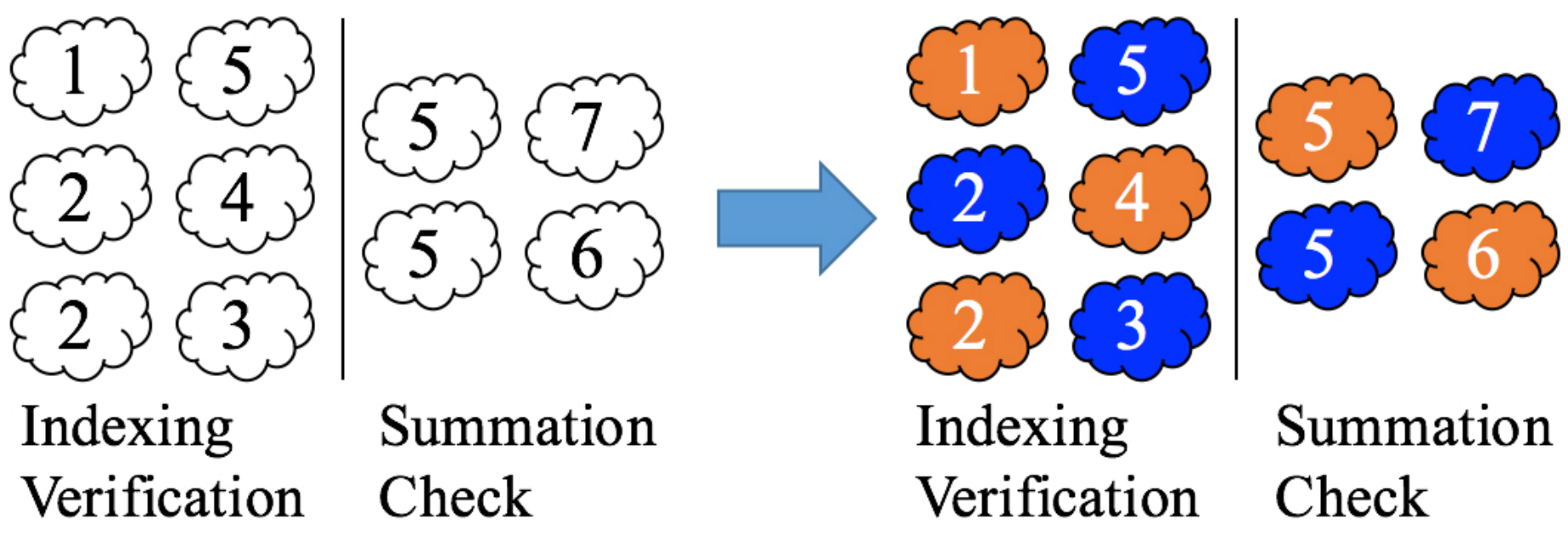}
\caption{\small{Consecutive summation check. In this example, we have three indexing and verification query vectors $\vect{x}_1$, $\vect{x}_2$, $\vect{x}_3$, and two summation check query vectors $\vect{x}_1+\vect{x}_2$ and $\vect{x}_2+\vect{x}_3$. The numbers in the center of each ``cloud'' are the denoised measurements. Here, the denoised measurements that we get for $\vect{x}_i$, $i=1,2,3$ are $(1,5)$, $(2,4)$, and $(2,3)$, respectively, and the denoised measurements for $\vect{x}_1+\vect{x}_2$ and $\vect{x}_2+\vect{x}_3$ are $(5,7)$ and $(5,6)$, respectively. By matching summations, one can easily find the denoised measurements that are generated by the same parameter vector, i.e., consistent sets. In the figure on the right, the numbers shown with the same color are generated by the same parameter vector.}}
\label{fig:summation}
\end{figure}

\begin{algorithm}[t]
\caption{Consecutive Summation Check}\label{alg:consecutive}
\begin{algorithmic} 
\REQUIRE Denoised binary indexing and verification measurements $y_{i,\ell}$, $i=1,\ldots,P_1+P_2$, $\ell=1,2$. \\
\STATE Denoised summation check measurements $y_{(j,k),\ell}$, $1\le j < k \le P_1+P_2$, $\ell=1,2$.
\ENSURE Consistent denoised indexing and verification measurements $y_{i,\ell}$, $i=1,\ldots,P_1+P_2$, $\ell=1,2$.
\STATE $t\leftarrow 0$
\WHILE{$t < P_1+P_2$}
\STATE $s\leftarrow \arg\min\{i>t:y_{i,1}\neq y_{i,2}\}$
\IF{$(y_{t,1}+y_{s,2}, y_{t,2}+y_{s,1}) = (y_{(t,s), 1}, y_{(t,s), 2})$ or $(y_{t,1}+y_{s,2}, y_{t,2}+y_{s,1}) = (y_{(t,s), 2}, y_{(t,s), 1})$}
\STATE{swap $y_{s,1}$ and $y_{s,2}$}
\ENDIF
\STATE $t\leftarrow s$
\ENDWHILE
\end{algorithmic}
\end{algorithm}

\textbf{Indexing:}
We conduct indexing process on the consistent sets. The purpose of the indexing process is to check whether there is a single non-zero element associated with a set of consistent measurements (i.e., whether these measurements form a singleton), and find the location and value of the non-zero element. Recall that after the swapping procedures in Algorithm~\ref{alg:consecutive}, we obtain two consistent sets of denoised measurements $(y_{1,1},\ldots,y_{P_1+P_2,1})$ and $(y_{1,2},\ldots,y_{P_1+P_2,2})$. Without loss of generality, we omit the second subscript and use $(y_{1},\ldots,y_{P_1+P_2})$ to denote one of the consistent set of denoised measurements.

We check the first $P_1$ denoised measurements, which correspond to the binary indexing query vectors. We can see that for the consistent set to be a singleton, it is necessary that all the non-zero denoised binary indexing measurements take the same value in $\mathbb{D}$, say $a\Delta$.
The only possible location index $j$ of the non-zero element satisfies the fact that integer $j-1$ has binary representation $\{\frac{1}{a\Delta}y_{i}\}_{i=1}^{P_1}\in\{0,1\}^{P_1}$. For instance, in the example in Figure~\ref{fig:example_noisy}, suppose that we find a consistent set of measurements generated by $\vect{\beta}^{(1)}$, and the quantization step size $\Delta=1$, i.e., the non-zero elements take integer values. Assume that we observe the denoised consistent measurements $(\vect{b}_1\TSP\DIAG{\vect{h}}\vect{\beta}^{(1)}, \vect{b}_2\TSP\DIAG{\vect{h}}\vect{\beta}^{(1)}) = (2,2)$. Then, it is possible that this is a singleton, and $\beta_4^{(1)} = 2$ is the only non-zero element associated with these consistent measurements.

However, the procedure above is not enough to guarantee that the consistent measurements form a singleton. We continue the example in Figure~\ref{fig:example_noisy}. Suppose that the bipartite graph gives us $\vect{h} = [0~1~1~1]\TSP$. Then, when we observe measurements $(\vect{b}_1\TSP\DIAG{\vect{h}}\vect{\beta}^{(1)}, \vect{b}_2\TSP\DIAG{\vect{h}}\vect{\beta}^{(1)}) = (2,2)$, we can have either $\DIAG{\vect{h}}\vect{\beta}^{(1)} = [0~0~0~2]\TSP$ or $\DIAG{\vect{h}}\vect{\beta}^{(1)} = [0~2~2~0]\TSP$; and in the latter case, this set of measurements does not form a singleton any more. To verify that this consistent set is a singleton, we need to use the next $P_2$ denoised verification measurements. Recall that for the verification query vectors, we design a Rademacher matrix $\mat{V}\in\{-1,1\}^{P_2\times n}$ with elements $V_{i,j}$, $i\in[P_2],j\in[n]$, and use the rows of $\mat{V}\DIAG{\vect{h}}$ as query vectors. Here, we make the following claim on the singleton verification procedure.
\begin{lemma}\label{lem:verification}
Suppose that all the denoised binary indexing measurements $y_{1},\ldots, y_{P_1}$ take value in $\{0, a\Delta\}$, and the sequence $\{\frac{1}{a\Delta}y_{i}\}_{i=1}^{P_1}$ form the binary representation of integer $j-1$. Then, if the verification measurements satisfy $y_{i}=a\Delta V_{i-P_1,j}$ for all $i=P_1+1,\ldots,P_1+P_2$, and $P_2=\Theta(\log(n))$, with probability $1-\BIGO(1/\poly(n))$, this consistent set is indeed a singleton with the non-zero element located at the $j$-th coordinate and taking value $a\Delta$.
\end{lemma}
This result is a corollary of the Johnson-Lindenstrauss Lemma~\cite{johnson1984extensions}, and we provide the proof in Appendix~\ref{sec:prf_verification}. If the denoised measurements pass the verification in Lemma~\ref{lem:verification}, we know that with high probability, the consistent set of measurement is indeed a singleton, and we also obtain the location and value of the non-zero element. We provide the precise procedures of the indexing algorithm in Algorithm~\ref{alg:indexing}.

\begin{algorithm}[t]
\caption{Indexing Algorithm in Noisy Setting}\label{alg:indexing}
\begin{algorithmic} 
\REQUIRE Denoised consistent binary indexing measurements $y_i$, $i=1,\ldots,P_1$.
\STATE Denoised consistent verification measurements $y_{i,\ell}$, $i=P_1+1,\ldots,P_1+P_2$, $\ell=1,2$.
\STATE Rademacher matrix $\mat{V}\in\{-1,1\}^{P_2\times n}$ for verification, quantized set $\mathbb{D}$.
\ENSURE Singleton/non-singleton, singleton location $j$, singleton non-zero value $\beta_j$.
\IF{ $\forall~i\in[P_1]$, $y_i\in\{0,a\Delta\}$ for some $a\Delta\in\mathbb{D}$}
\STATE $j\leftarrow 1+ \sum_{i=1}^{P_1} 2^{P_1-i} \frac{y_i}{a\Delta}$
\IF{ $\forall~i = P_1+1,\ldots,P_1+P_2$, $y_{i}=a\Delta V_{i-P_1,j}$}
\STATE return singleton, $j$, $\beta_j = a\Delta$ 
\ELSE
\STATE return non-singleton
\ENDIF
\ELSE
\STATE return non-singleton
\ENDIF
\end{algorithmic}
\end{algorithm}

So far, we have demonstrated how we robustify the summation check and indexing process. Once these two parts are robustified, other parts of the algorithm, such as finding giant components, guess-and-check, and peeling-style iterative decoding can proceed as in the noiseless case. We relegate the analysis of Robust Mixed-Coloring algorithm to Appendix~\ref{sec:prf_noisy}. Again, there are a few design parameters in the Robust Mixed-Coloring algorithm, and we summarize the choices of these parameters in Table~\ref{tab:design_para_noisy} for a particular target error fraction $p^*$.

Finally, we point out that extending the Robust Mixed-Coloring algorithm to cases where $L>2$ is an important future direction. Although the summation check technique does not provably work in the noisy setting when $L>2$, we believe that using similar but more sophisticated design, we may still be able to obtain consistent sets of measurements.

\begin{table}[h]
\caption{Design parameters of Robust Mixed-Coloring algorithm}
\begin{center}
  \begin{tabular}{ c | c | c  }
    \hline
    Parameter & Description & Choice  \\ \hline
    $p^*$  &  error fraction & $5.1\times 10^{-6}$ \\ \hline
    $d$ & left degree of bipartite graph & $15$  \\ \hline 
    $N$ & number of repetitions of each query vector & $\Theta(\polylog(n))$ \\ \hline
    $P_1$ & number of binary indexing vectors in each bin & $\lceil \log_2(n) \rceil$  \\  \hline
    $P_2$ & number of verification vectors in each bin & $\Theta(\log(n))$ \\ \hline
    $P_3$ & number of summation check vectors in each bin & ${P_1 + P_2 \choose 2} = \Theta(\log^2(n))$  \\ \hline
    $M$ & number of bins & $3.71K$  \\ \hline
    $m$ & total number of measurements & $m=KN(P_1+P_2+P_3) = \Theta(K\polylog(n))$  \\ \hline
  \end{tabular} 
\end{center}
\label{tab:design_para_noisy}
\end{table}

\section{Experimental Results}\label{sec:experiments}
In this section, we test the sample and time complexities of the Mixed-Coloring algorithm in both noiseless and noisy cases to verify our theoretical results. All simulations are done on a laptop with 2.8 GHz Intel Core i7 CPU and 16 GB memory using Python.

We first investigate the sample complexity of Mixed-Coloring algorithm in the noiseless case. The goal of this experiment is to show that in numerical experiments, the number of measurements that we need to successfully recover the parameter vectors \emph{matches} the predictions of the density evolution analysis. We use the optimal parameters $(d,R,V)$ from numerical calculations of the density evolution, presented in Table~\ref{tab:design_para}. We generate instances with different number of measurements $m$ by choosing different number of bins $M$. Recall that $m = (2R+V)M$, and thus varying the number of bins is equivalent to varying the total number of measurements. The parameter vectors that we use have equal sparsity, i.e., $K_\ell = \frac{1}{L}K$, and the mixing weights are equal for all the parameter vectors, i.e., $q_\ell = \frac{1}{L}$. The supports of the parameter vectors are chosen uniformly at random, and the values of the non-zero elements are generated from Gaussian distribution. We choose a few pairs of $L$ and $K$, increase the total number of measurements, and record the empirical success probability and running time averaged over $100 $ trials. Here, we use a sufficiently small $p^*$ so that the success event is equivalent to recovery of \emph{all} the non-zero elements. The results are shown in Figure \ref{subfig:noiseless_sample}. The phase transition occurs at some $ C=m/K $ that matches the values in Table~\ref{tab:thmconst}, predicted by our theory. More specifically, when $L=2$, $L=3$, and $L=4$, we need about $33K$, $38K$, and $40K$ measurements for successful recovery, respectively.

We also test the time complexity of our algorithm in the noiseless case. We use the design parameters that can guarantee successful recovery, as we find in the experiment on sample complexity. More specifically, for $L=2$, we choose $(d, R, V) = (15, 3, 3)$, and $m=34.2K$ (i.e., $M=3.8K$); for $L = 3$, we choose $(d, R, V)=(15, 5, 5)$, and $m=39K$ (i.e., $M=2.6K$); and for $ L=4 $, we choose $(d,R,V)=(13, 8, 8)$, and $m=43.2K$ (i.e., $M=1.8K$).  As shown in Figure~\ref{subfig:noiseless_time}, the running time is linear in $K$ and does not depend on $n$.

\begin{figure}[h]
    \centering
    \begin{subfigure}[b]{0.48\textwidth}
        \includegraphics[width=\textwidth]{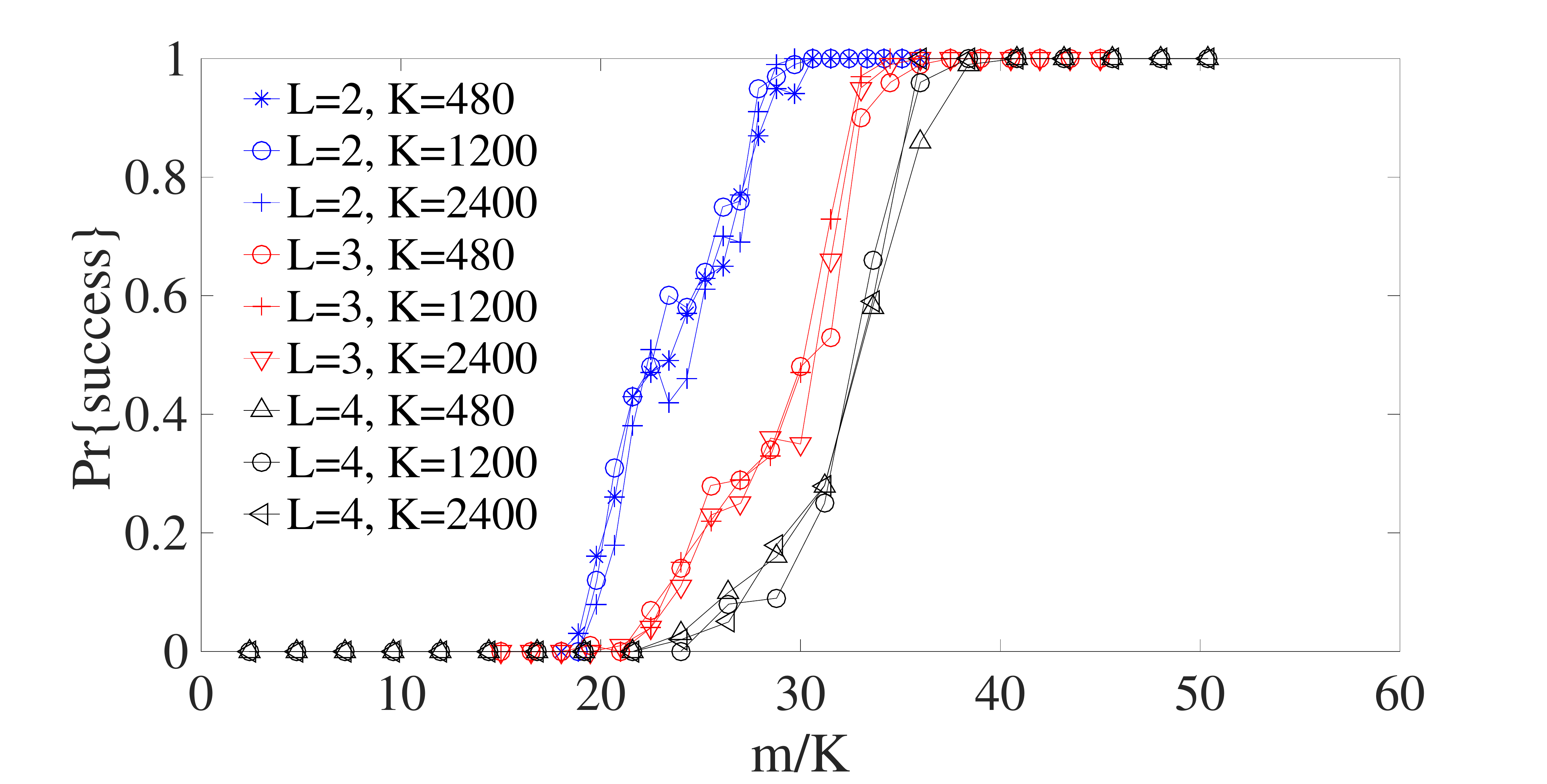}
        \caption{Probability of success $(n=10^5)$}
        \label{subfig:noiseless_sample}
    \end{subfigure}
    \quad 
    \begin{subfigure}[b]{0.48\textwidth}
        \includegraphics[width=\textwidth]{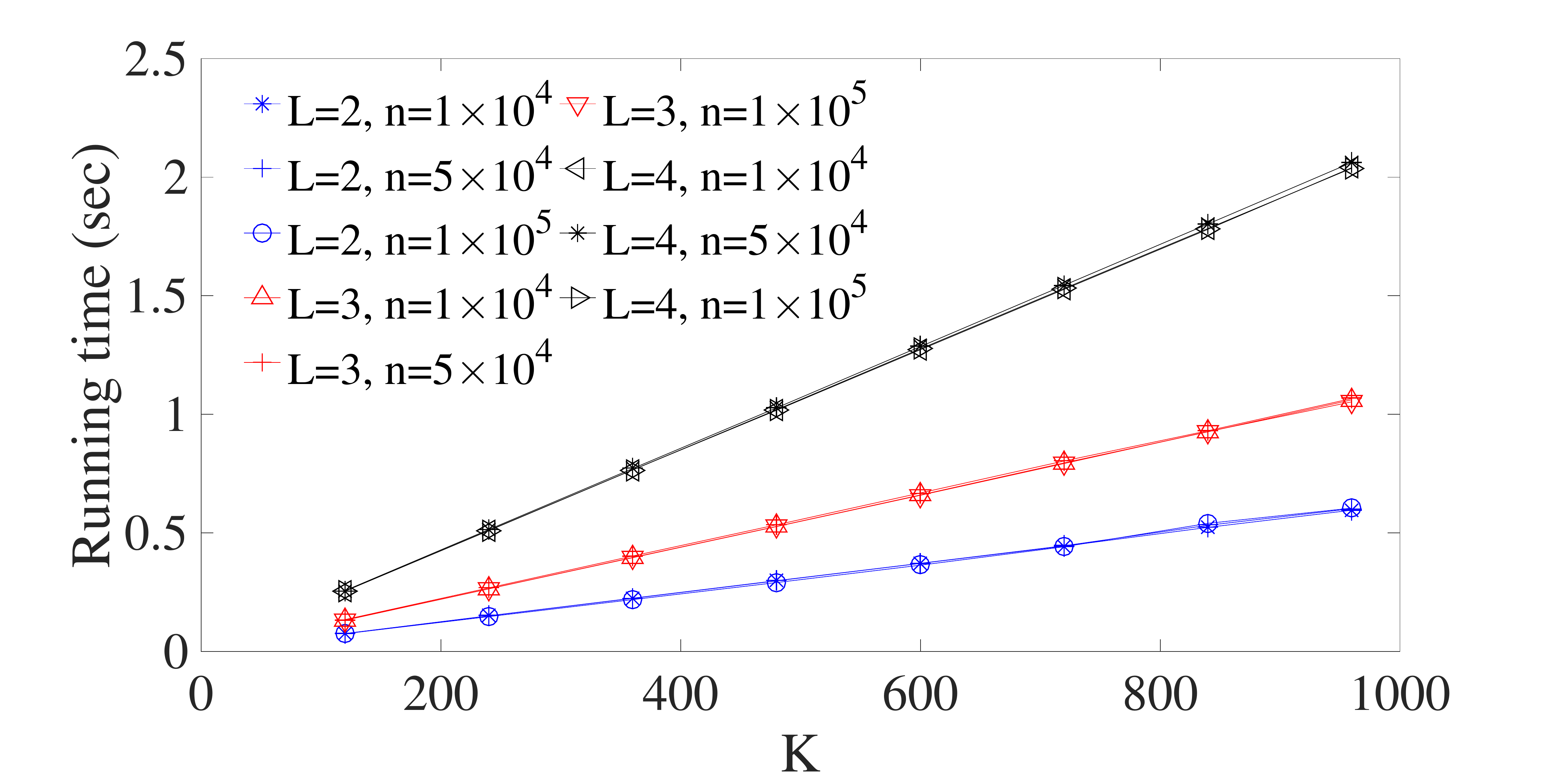}
        \caption{Time complexity}
        \label{subfig:noiseless_time}
    \end{subfigure}
	\caption{\small{Success probability and running time in the noiseless case. In both (a) and (b), for $L = 2$, we use $(d, R, V) = (15, 3, 3)$, for $ L=3 $, we use $(d, R, V)=(15, 5, 5)$, and for $ L=4 $, we use $(d,R,V)=(13, 8, 8)$. We increase the number of measurements $m$ by increasing the number of bins $M$. As we can see, the total number of measurements needed for successful recovery matches the sample complexities predicted by our theory in Table~\ref{tab:thmconst}. In (b), for $L=2$, we use $m=34.2K$ (i.e., $M=3.8K$); for $ L=3 $, we use $m=39K$ (i.e., $M=2.6K$); and for $L=4$, we use $m=43.2K$ (i.e., $M=1.8K$).}}
\label{fig:noiseless}
\end{figure}

Similar experiments are performed for the noisy case using the Robust Mixed-Coloring algorithm, under the quantization assumption. We still focus on the case where the two parameter vectors appear equally like and have the same sparsity. We use quantization step size $\Delta = 1$ and the quantized alphabet $\mathbb{D}=\{\pm1, \pm2, \ldots, \pm5\}$, and the values of the non-zero elements are chosen uniformly at random from $\mathbb{D}$. Figure~\ref{subfig:noisy_sample} shows the minimum number of queries $ m $ required for 100 consecutive successes, for different $ n $ and $ K $. We observe that the sample complexity is \emph{linear} in $K$ and \emph{sublinear} in $n$. The running time exhibits a similar behavior, as shown in Figure~\ref{subfig:noisy_time}.  Both observations agree with the prediction of our theory.

\begin{figure}[h]
    \centering
    \begin{subfigure}[b]{0.48\textwidth}
        \includegraphics[width=\textwidth]{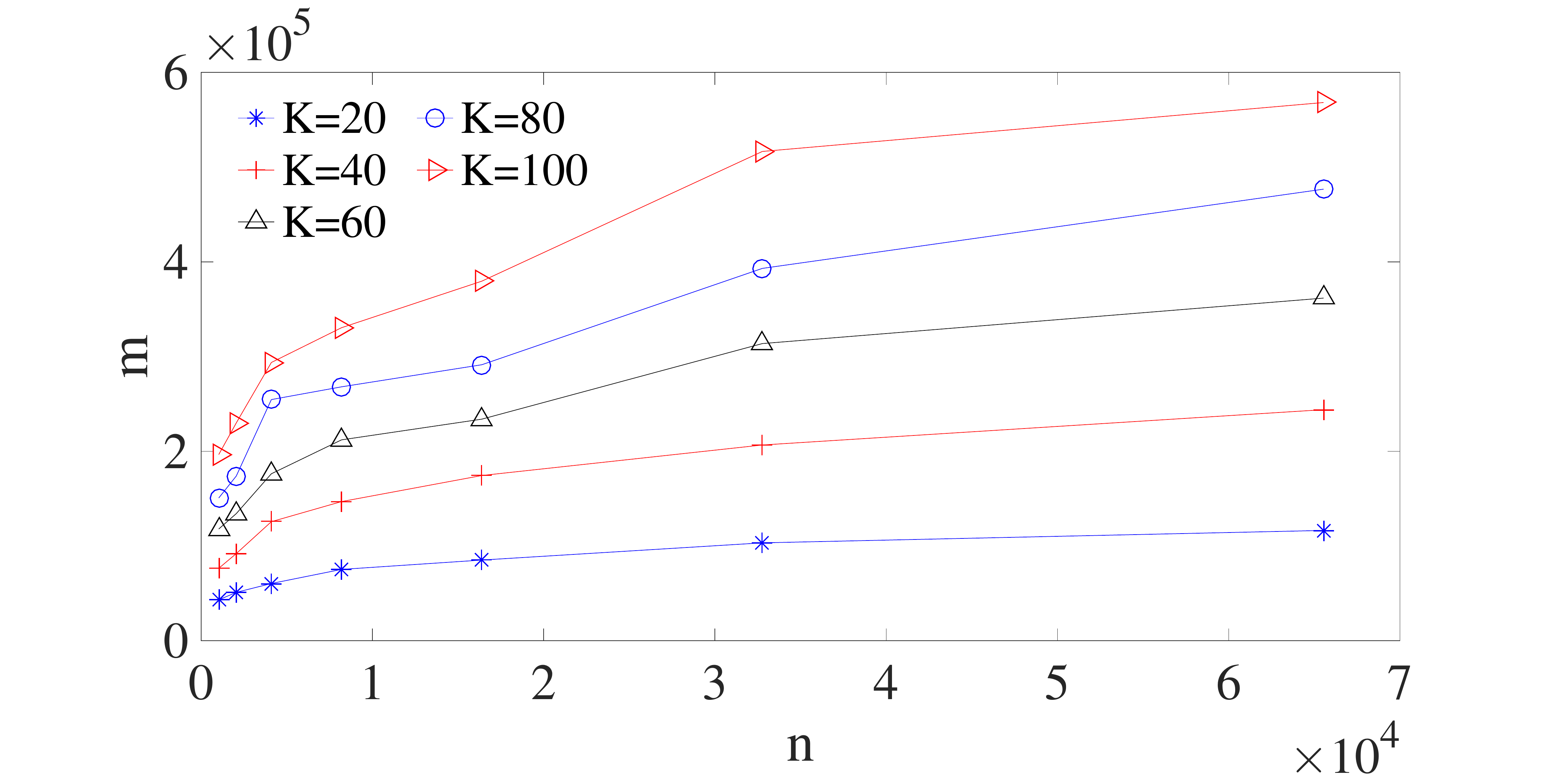}
        \caption{Sample complexity $(\Delta=1, \sigma=0.2)$}
        \label{subfig:noisy_sample}
    \end{subfigure}
    \quad 
    \begin{subfigure}[b]{0.48\textwidth}
        \includegraphics[width=\textwidth]{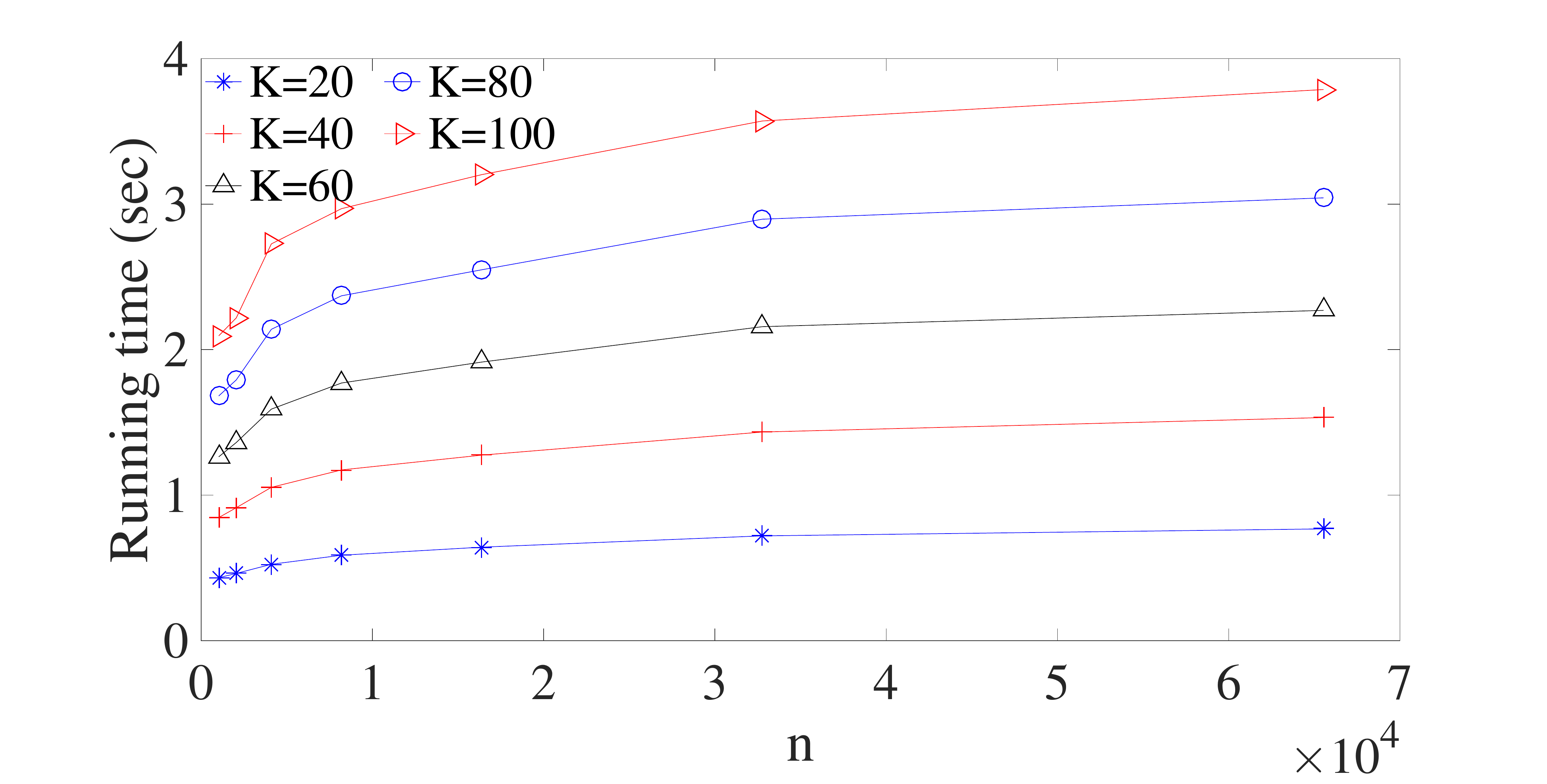}
        \caption{Time complexity $(\Delta=1, \sigma=0.2)$}
        \label{subfig:noisy_time}
    \end{subfigure}
	\caption{\small{Sample and time complexities of the Robust Mixed-Coloring algorithm. In both (a) and (b), we choose the quantization step size $\Delta=1$, the quantize alphabet $\mathbb{D}=\{\pm1, \pm2, \ldots, \pm5\}$, standard deviation of noise $\sigma=0.2$. The design parameters are chosen as follows: left degree $d=15$, number of bins $M=3K$, number of singleton verification query vectors: $0.3\log_2(n)$. In (a), we vary the number of repetitions of each query vector $N$ to find the minimum number of query vectors needed for successful recovery. In (b), we fix $N=\log_2(n)$ and measure the time cost. The experiment on time complexity is conducted in the setting where successful recovery is guaranteed.}}
\label{fig:noisy}
\end{figure}

We also compare the Mixed-Coloring algorithm with a state-of-the-art EM-style algorithm (equivalent to alternating minimization in the noiseless setting) from~\cite{yi2014alternating}. These comparisons are not entirely fair, since our algorithm is based on carefully designed query vectors, while the algorithm in~\cite{yi2014alternating} uses random design, i.e., the entries of $\vect{x}_i$'s are i.i.d. Gaussian. However, this is exactly where the intellectual value of our work lies: we expose the gains available by careful design. We consider four test cases with $(L, n, K) = (2, 100, 20), (2, 500, 50), (2, 100, 100), (2, 500, 500)$, with the first two cases being sparse problems and the last two being relatively dense problems. We find the minimum number of queries that leads to a 100\% successful rate in 100 trials, and the average running time. For the Mixed-Coloring algorithm, we use $d=15$, $R=V=3$ and $M=3.8K$. The parameters of the EM-style algorithm are chosen as suggested in the original paper~\cite{yi2014alternating}. As shown in Table~\ref{tab:noiseless_comp}, in both sparse and dense problems, our Mixed-Coloring algorithm is several orders of magnitude faster. As for the sample complexity, our algorithm requires smaller number of samples in the sparse cases, while in dense problems, the sample complexity of our algorithm is within a constant factor (about 3) of that of the alternating minimization algorithm. For the noisy setting, our algorithm is most powerful in the high dimensional setting, i.e., large $n$, due to the $\text{polylog}(n)$ factors. However, in this setting, it takes prohibitively long time for the state-of-the-art algorithms such as~\cite{stadler2010} to converge, and thus, we do not present the comparison in the noisy setting.
 
\begin{table}[h]
\centering
\begin{center}
  \begin{tabular}{ c | c | c | c }
    \hline
Problem type  &  $(n,K)$  &  $\frac{\text{sample(M-C)}}{\text{sample(EM)}}$ & $\frac{\text{run-time(M-C)}}{\text{run-time(EM)}}$  \\  \hline
   \multirow{2}{*}{Sparse} & $(100,20)$  &  $0.57$  &  $0.00806$  \\ 
 &   $(500,50)$  &  $0.33$  &  $0.00272$  \\ \hline
\multirow{2}{*}{Dense}   &  $(100,100)$  &  $2.78$  &  $0.0526$  \\
   & $(500,500)$  &  $3.00$  &  $0.0270$  \\ \hline   
\end{tabular}
\end{center}
\caption{Comparison of the Mixed-Coloring algorithm (M-C) and the EM-style algorithm (EM). Mixed-Coloring algorithm is advantageous in time complexity for both sparse and dense problems, and is advantageous in sample complexity for sparse problems.}
\label{tab:noiseless_comp}
\end{table}

\begin{figure}[h]
\centering
\includegraphics[width=0.48\textwidth]{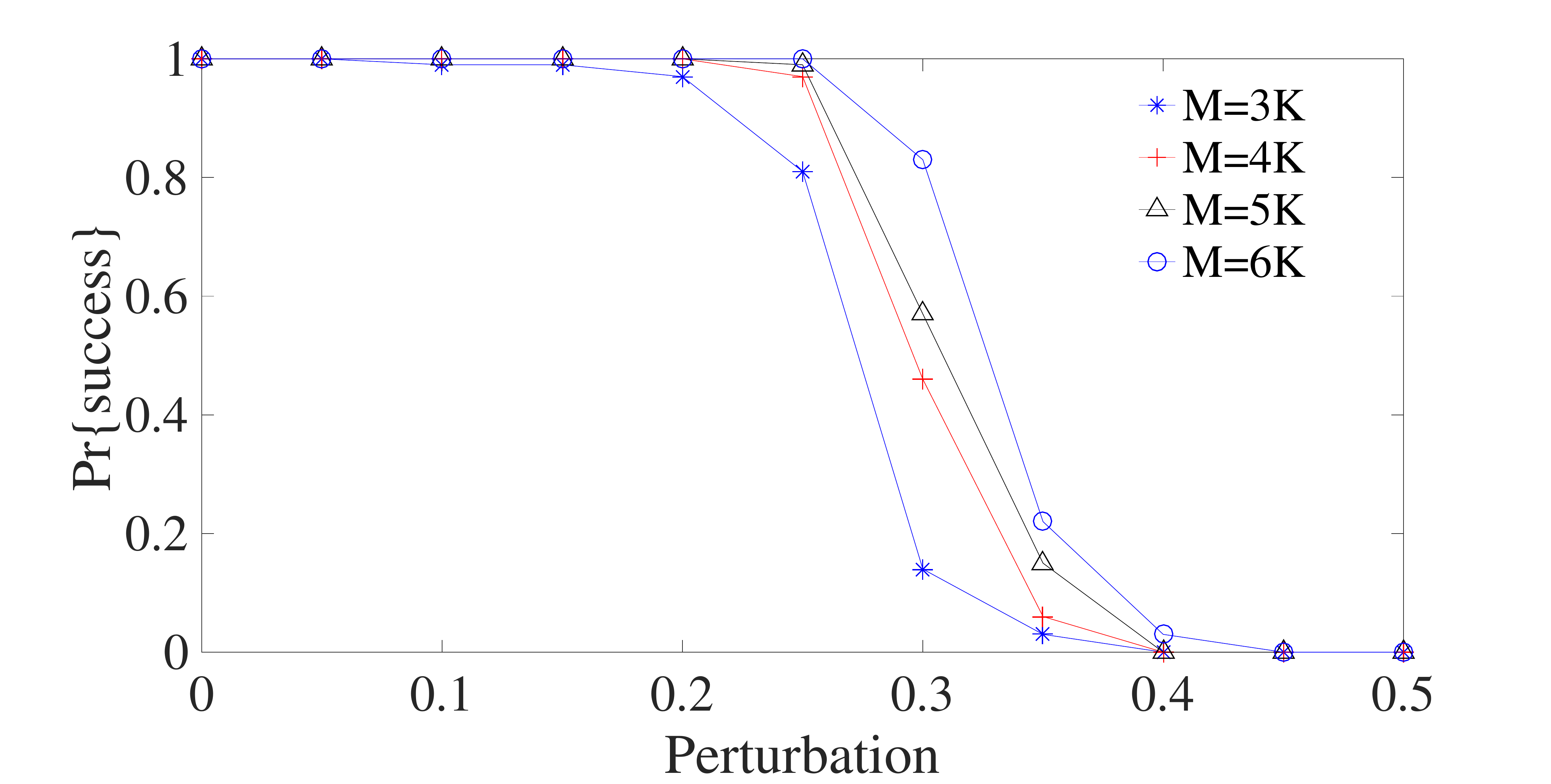}
\caption{\small{Performance of Robust Mixed-Coloring algorithm with quantization assumption violated. We vary the number of bins $M$ to test the empirical probability of success, and also keep $d=5M/K$. Other parameters: $n=4096$, $K=50$, quantization level $\Delta=1$, standard deviation of noise $\sigma=0.1$, number of singleton verification query vectors: $0.3\log_2(n)$, and $R=\log_2(n)$.}}
\label{fig:perturbation}
\end{figure}

We further test the Robust Mixed-Coloring algorithm when the quantization assumption is violated. For any $\beta\in\REAL$, we define $D(\beta)=\arg\min_{a\in\mathbb{D}}\ABSL{a-\beta}\vect{1}(\beta\neq 0)$, where $\vect{1}(\cdot)$ denotes the indicator function. This means that $D(\beta)$ is the element in $\mathbb{D}$ which is the closest one to $\beta$, when $\beta\neq0$. For a vector $\vect{\beta}\in\REAL^n$, we define $D(\vect{\beta}) = \{D(\beta_j)\}_{j=1}^n$. We define the \emph{perturbation} of a vector $\vect{\beta}$ as $\text{Perturbation}(\vect{\beta}) = \max_{j\in[n]} \ABSL{\beta_j-D(\beta_j)}/\Delta$. 

In this experiment, we generate sparse parameter vectors $\vect{\beta}^{(\ell)}$, $\ell\in[L]$ with a total number of $ K $ non-zero elements. These non-zero elements are generated randomly while keeping the perturbation of the parameter vectors under a certain level by adding bounded noise to the quantized non-zero elements. We record the probability of success for different number of bins $ M $ and different perturbation level. Here the success event is defined as recovery of $D(\vect{\beta}^{(\ell)})$ for all $\ell\in[L]$. The result is shown in Figure~\ref{fig:perturbation}. We see that the Robust Mixed-Coloring algorithm works without the quantization assumption as long as the perturbations are not too large.

\section{Conclusions}\label{sec:con}
We proposed the Mixed-Coloring algorithm as a query based learning algorithm for mixtures of sparse linear regressions. Our algorithm leverages the connection between modern coding theory and statistical inference. The design of the query vectors and the recovery algorithm are based on ideas from sparse graph codes. Our novel code design allows for both efficient demixing and parameter estimation. In the noiseless setting, for a constant number of sparse parameter vectors, our algorithm achieves the order-optimal sample and time complexities of $\Theta(K)$. In the presence of Gaussian noise, for the problem with two parameter vectors (i.e., $L=2$), we show that the Robust Mixed-Coloring algorithm achieves near-optimal $\Theta(K\polylog(n))$ sample and time complexities. Our experiments justified the theoretical results, and we observe that the run-time of our algorithm can be orders of magnitudes smaller than that of the state-of-the-art algorithms. In the noisy scenario, studying the Robust Mixed-Coloring algorithm with more than two parameter vectors and obtaining theoretical results for the continuous alphabet case are two important future directions.

\ifCLASSOPTIONcaptionsoff
  \newpage
\fi
\bibliographystyle{ieeetr}
\bibliography{refs}

\begin{thebibliography}{10}

\bibitem{harville2002framework}
M.~Harville, ``A framework for high-level feedback to adaptive, per-pixel,
  mixture-of-{G}aussian background models,'' in {\em Proceedings of the 7th
  European Conference on Computer Vision}, pp.~543--560, Springer, 2002.

\bibitem{reynolds2000speaker}
D.~A. Reynolds, T.~F. Quatieri, and R.~B. Dunn, ``Speaker verification using
  adapted {G}aussian mixture models,'' {\em Digital Signal Processing},
  vol.~10, no.~1, pp.~19--41, 2000.

\bibitem{zhang2012guess}
A.~Zhang, N.~Fawaz, S.~Ioannidis, and A.~Montanari, ``Guess who rated this
  movie: identifying users through subspace clustering,'' in {\em Proceedings
  of the 28th Conference on Uncertainty in Artificial Intelligence},
  pp.~944--953, AUAI Press, 2012.

\bibitem{de1989mixtures}
R.~De~Veaux, ``Mixtures of linear regressions,'' {\em Computational Statistics
  and Data Analysis}, vol.~8, no.~3, 1989.

\bibitem{candes2006robust}
E.~J. Cand{\`e}s, J.~Romberg, and T.~Tao, ``Robust uncertainty principles:
  Exact signal reconstruction from highly incomplete frequency information,''
  {\em IEEE Transactions on Information Theory}, vol.~52, no.~2, pp.~489--509,
  2006.

\bibitem{candes2006stable}
E.~J. Cand{\`e}s, J.~K. Romberg, and T.~Tao, ``Stable signal recovery from
  incomplete and inaccurate measurements,'' {\em Communications on Pure and
  Applied Mathematics}, vol.~59, no.~8, pp.~1207--1223, 2006.

\bibitem{blackwell2006applying}
E.~Blackwell, C.~F.~M. de~Leon, and G.~E. Miller, ``Applying mixed regression
  models to the analysis of repeated-measures data in psychosomatic medicine,''
  {\em Psychosomatic Medicine}, vol.~68, no.~6, 2006.

\bibitem{deb2002estimates}
P.~Deb and M.~Holmes, ``Estimates of use and costs of behavioural health care:
  a comparison of standard and finite mixture models,'' {\em Econometric
  Analysis of Health Data}, pp.~87--99, 2002.

\bibitem{viele2002modeling}
K.~Viele and B.~Tong, ``Modeling with mixtures of linear regressions,'' {\em
  Statistics and Computing}, vol.~12, no.~4, pp.~315--330, 2002.

\bibitem{gallager1962low}
R.~Gallager, ``Low-density parity-check codes,'' {\em IRE Transactions on
  Information Theory}, vol.~8, no.~1, pp.~21--28, 1962.

\bibitem{wainwright2009information}
M.~J. Wainwright, ``Information-theoretic limits on sparsity recovery in the
  high-dimensional and noisy setting,'' {\em IEEE Transactions on Information
  Theory}, vol.~55, no.~12, pp.~5728--5741, 2009.

\bibitem{akccakaya2010shannon}
M.~Ak{\c{c}}akaya and V.~Tarokh, ``Shannon-theoretic limits on noisy
  compressive sampling,'' {\em IEEE Transactions on Information Theory},
  vol.~56, no.~1, pp.~492--504, 2010.

\bibitem{lewicki1998review}
M.~S. Lewicki, ``A review of methods for spike sorting: the detection and
  classification of neural action potentials,'' {\em Network: Computation in
  Neural Systems}, vol.~9, no.~4, pp.~R53--R78, 1998.

\bibitem{jansen1992general}
R.~Jansen, ``A general mixture model for mapping quantitative trait loci by
  using molecular markers,'' {\em Theoretical and Applied Genetics}, vol.~85,
  no.~2-3, pp.~252--260, 1992.

\bibitem{chen2014convex}
Y.~Chen, X.~Yi, and C.~Caramanis, ``A convex formulation for mixed regression
  with two components: Minimax optimal rates,'' in {\em Proceedings of
  Conference on Learning Theory}, pp.~560--604, 2014.

\bibitem{balakrishnan2014statistical}
S.~Balakrishnan, M.~J. Wainwright, and B.~Yu, ``Statistical guarantees for the
  {EM} algorithm: From population to sample-based analysis,'' {\em The Annals
  of Statistics}, vol.~45, no.~1, pp.~77--120, 2017.

\bibitem{faria2010fitting}
S.~Faria and G.~Soromenho, ``Fitting mixtures of linear regressions,'' {\em
  Journal of Statistical Computation and Simulation}, vol.~80, no.~2,
  pp.~201--225, 2010.

\bibitem{stadler2010}
N.~St{\"a}dler, P.~B{\"u}hlmann, and S.~Van De~Geer, ``$\ell_1$-penalization
  for mixture regression models,'' {\em Test}, vol.~19, no.~2, pp.~209--256,
  2010.

\bibitem{yi2014alternating}
X.~Yi, C.~Caramanis, and S.~Sanghavi, ``Alternating minimization for mixed
  linear regression,'' in {\em Proceedings of the 31st International Conference
  on Machine Learning}, pp.~613--621, 2014.

\bibitem{yi2016solving}
X.~Yi, C.~Caramanis, and S.~Sanghavi, ``Solving a mixture of many random linear
  equations by tensor decomposition and alternating minimization,'' {\em arXiv
  preprint arXiv:1608.05749}, 2016.

\bibitem{chaganty2013spectral}
A.~T. Chaganty and P.~Liang, ``Spectral experts for estimating mixtures of
  linear regressions,'' in {\em Proceedings of the 30th International
  Conference on Machine Learning}, pp.~1040--1048, 2013.

\bibitem{sun2014learning}
Y.~Sun, S.~Ioannidis, and A.~Montanari, ``Learning mixtures of linear
  classifiers,'' in {\em Proceedings of the 31st International Conference on
  Machine Learning}, pp.~721--729, 2014.

\bibitem{arikan2009channel}
E.~Arikan, ``Channel polarization a method for constructing capacity-achieving
  codes for symmetric binary-input memoryless channels,'' {\em IEEE
  Transactions on Information Theory}, vol.~55, no.~7, pp.~3051--3073, 2009.

\bibitem{li2015active}
X.~Li and K.~Ramchandran, ``An active learning framework using sparse-graph
  codes for sparse polynomials and graph sketching,'' in {\em Proceedings of
  Advances in Neural Information Processing Systems}, pp.~2161--2169, 2015.

\bibitem{pedarsani2015sparse}
R.~Pedarsani, K.~Lee, and K.~Ramchandran, ``Sparse covariance estimation based
  on sparse-graph codes,'' in {\em Proceedings of IEEE Annual Allerton
  Conference on Communication, Control, and Computing}, pp.~612--619, 2015.

\bibitem{ubaru2015low}
S.~Ubaru, A.~Mazumdar, and Y.~Saad, ``Low rank approximation using error
  correcting coding matrices,'' in {\em Proceedings of the 32nd International
  Conference on Machine Learning}, pp.~702--710, 2015.

\bibitem{ermon2014low}
S.~Ermon, C.~Gomes, A.~Sabharwal, and B.~Selman, ``Low-density parity
  constraints for hashing-based discrete integration,'' in {\em Proceedings of
  the 31st International Conference on Machine Learning}, pp.~271--279, 2014.

\bibitem{li2014sub}
X.~Li, D.~Yin, S.~Pawar, R.~Pedarsani, and K.~Ramchandran, ``Sub-linear time
  support recovery for compressed sensing using sparse-graph codes,'' {\em
  arXiv preprint arXiv:1412.7646}, 2014.

\bibitem{yincompressed}
D.~Yin, R.~Pedarsani, X.~Li, and K.~Ramchandran, ``Compressed sensing using
  sparse-graph codes for the continuous-alphabet setting,'' {\em Proceedings of
  IEEE Annual Allerton Conference on Communication, Control, and Computing},
  pp.~758--765, 2016.

\bibitem{pedarsani2014phasecode}
R.~Pedarsani, D.~Yin, K.~Lee, and K.~Ramchandran, ``Phasecode: Fast and
  efficient compressive phase retrieval based on sparse-graph codes,'' {\em
  IEEE Transactions on Information Theory}, vol.~63, no.~6, pp.~3663--3691,
  2017.

\bibitem{yin2015fast}
D.~Yin, K.~Lee, R.~Pedarsani, and K.~Ramchandran, ``Fast and robust compressive
  phase retrieval with sparse-graph codes,'' in {\em Proceedings of IEEE
  International Symposium on Information Theory}, pp.~2583--2587, 2015.

\bibitem{du2000combinatorial}
D.~Du, F.~K. Hwang, and F.~Hwang, {\em Combinatorial group testing and its
  applications}, vol.~12.
\newblock World Scientific, 2000.

\bibitem{d2000new}
A.~G. D’yachkov, A.~J. Macula, and V.~V. Rykov, ``New applications and
  results of superimposed code theory arising from the potentialities of
  molecular biology,'' in {\em Numbers, Information and Complexity},
  pp.~265--282, Springer, 2000.

\bibitem{cheraghchi2010derandomization}
M.~Cheraghchi, ``Derandomization and group testing,'' in {\em Proceedings of
  IEEE Annual Allerton Conference on Communication, Control, and Computing},
  pp.~991--997, 2010.

\bibitem{lee2015saffron}
K.~Lee, R.~Pedarsani, and K.~Ramchandran, ``{S}{A}{F}{F}{R}{O}{N}: A fast,
  efficient, and robust framework for group testing based on sparse-graph
  codes,'' in {\em Proceedings of IEEE International Symposium on Information
  Theory}, pp.~2873--2877, 2016.

\bibitem{mazumdar2016nonadaptive}
A.~Mazumdar, ``Nonadaptive group testing with random set of defectives,'' {\em
  IEEE Transactions on Information Theory}, vol.~62, no.~12, pp.~7522--7531,
  2016.

\bibitem{dhillon2011nearest}
I.~S. Dhillon, P.~K. Ravikumar, and A.~Tewari, ``Nearest neighbor based greedy
  coordinate descent,'' in {\em Proceedings of Advances in Neural Information
  Processing Systems}, pp.~2160--2168, 2011.

\bibitem{pilanci2014iterative}
M.~Pilanci and M.~J. Wainwright, ``Iterative {H}essian sketch: Fast and
  accurate solution approximation for constrained least-squares,'' {\em Journal
  of Machine Learning Research}, vol.~17, no.~1, pp.~1842--1879, 2016.

\bibitem{achlioptas2015stochastic}
D.~Achlioptas and P.~Jiang, ``Stochastic integration via error-correcting
  codes,'' in {\em Proceedings of the 31st Conference on Uncertainty in
  Artificial Intelligence}, pp.~22--31, 2015.

\bibitem{sipser1996expander}
M.~Sipser and D.~A. Spielman, ``Expander codes,'' {\em IEEE Transactions on
  Information Theory}, vol.~42, no.~6, pp.~1710--1722, 1996.

\bibitem{jafarpour2009efficient}
S.~Jafarpour, W.~Xu, B.~Hassibi, and R.~Calderbank, ``Efficient and robust
  compressed sensing using optimized expander graphs,'' {\em IEEE Transactions
  on Information Theory}, vol.~55, no.~9, pp.~4299--4308, 2009.

\bibitem{richardson2008modern}
T.~Richardson and R.~Urbanke, {\em Modern coding theory}.
\newblock Cambridge University Press, 2008.

\bibitem{RU01}
T.~Richardson and R.~Urbanke, ``The capacity of low-density parity-check codes
  under message-passing decoding,'' {\em IEEE Transactions on Information
  Theory}, vol.~47, pp.~599--618, 2001.

\bibitem{johnson1984extensions}
W.~B. Johnson and J.~Lindenstrauss, ``Extensions of {L}ipschitz mappings into a
  {H}ilbert space,'' {\em Contemporary Mathematics}, vol.~26, no.~189-206,
  p.~1, 1984.

\bibitem{erd6s1960evolution}
P.~Erdos and A.~Renyi, ``On the evolution of random graphs,'' {\em Publications
  of the Mathematical Institute of the Hungarian Academy of Sciences}, vol.~5,
  pp.~17--61, 1960.

\bibitem{richardson2001capacity}
T.~J. Richardson and R.~L. Urbanke, ``The capacity of low-density parity-check
  codes under message-passing decoding,'' {\em IEEE Transactions on Information
  Theory}, vol.~47, no.~2, pp.~599--618, 2001.

\bibitem{baraniuk2008simple}
R.~Baraniuk, M.~Davenport, R.~DeVore, and M.~Wakin, ``A simple proof of the
  restricted isometry property for random matrices,'' {\em Constructive
  Approximation}, vol.~28, no.~3, pp.~253--263, 2008.

\end{thebibliography}

\appendices
\section{Proof of Theorem \ref{thm:main_noiseless}} \label{sec:prf_noiseless}
\subsection{Proof Outline}\label{sec:prfoutline}
We prove Theorem~\ref{thm:main_noiseless} in this section. The proof includes two major steps: (i) show that the expectation of the fraction of non-zero elements which are not recovered can be arbitrarily small; (ii) show that this fraction concentrates around its mean with high probability. The first part mainly uses density evolution techniques which are commonly used in coding theory, and the second part uses Doob's martingale argument.

\subsection{Notation}\label{sec:notation}
We briefly recall the Mixed-Coloring algorithm in the noiseless case and declare the notation that we use for the rest of the proof. 

Recall that the parameter vector $\vect{\beta}^{(\ell)}$ has $K_\ell$ non-zero elements. We call these $K_\ell$ non-zero elements \emph{balls} in \emph{color} $\ell$. We design a $d$-left regular bipartite graph with $n$ left nodes and $M$ right nodes, representing the $n$ coordinates and the $M$ \emph{bins}, respectively. We denote the $i$-th bin by $\mathcal{B}_i$. We use the matrix $\mat{H}\in\{0,1\}^{M\times n}$ to represent the biadjacency matrix of the bipartite graph, i.e., $H_{i,j}=1$ if and only if the $i$-th bin is associated with the $j$-th coordinate. Recall that we design three query vectors in the form of~\eqref{eq:example_ratio_test}, for the purpose of ratio test. The third query vectors is the summation of the first two and is used for summation check. We repeat the first two query vectors $R$ times, respectively, and get $R$ type-I and $R$ type-II index measurements. We repeat the third query vector $V$ times and get $V$ verification measurements. For the $j$-th verification measurement of the $i$-th bin, we define a \emph{sub-bin} $\mathcal{B}_i^j$. If we can find one type-I index measurement and one type-II index measurement such that the summation of the two measurements is equal to the $j$-th verification measurement, we know that these three measurements are generated by the same parameter vector, say $\vect{\beta}^{(\ell)}$. The two index measurements are called a consistent pair. Then, we say that the sub-bin $\mathcal{B}_i^j$ has color $\ell$. We define the \emph{color set} $\mathcal{C}_i^j$ of $\mathcal{B}_i^j$. If we can find a consistent pair corresponding to the $j$-th verification measurement, we let $\mathcal{C}_i^j= \{\ell\}$, otherwise $\mathcal{C}_i^j = \emptyset$. We further define the color set of bin $\mathcal{B}_i$ as $\mathcal{C}_i = \cup_{j=1}^V\mathcal{C}_i^j$.

\subsection{Number of Singleton Balls}\label{sec:prf_singleton}
In this section, we analyze the number of singleton balls in color $\ell$ found in the first stage of the algorithm. We can show that this number is concentrated around a constant fraction of $K_\ell$ with high probability.
\begin{lemma}\label{lem:single}
Let $K_s^{(\ell)}$ be the number of singleton balls in color $\ell$ found in the first stage. Then, there exists a constant\footnote{Recall that in our paper, constants are defined as quantities which do not depend on $n$ and $K$.} $q_s^{(\ell)}$ such that for any constant $\delta>0$,
\begin{equation}\label{eq:num_single}
\PROL{\ABSL{K_s^{(\ell)}-K_{\ell}q_s^{(\ell)}} \le \delta K_{\ell}} \ge 1-2\exp(-2\delta^2K_{\ell}).
\end{equation}
\end{lemma} 
 
\begin{proof}
We first specify some terminologies here. For a bin $\mathcal{B}_i$, we say that this bin \emph{has} color $\ell$ when $\ell\in\mathcal{C}_i$. One should notice that if there are more than one sub-bins in color $\ell$ in bin $\mathcal{B}_i$, these sub-bins are identical. Therefore, we can say that a bin $\mathcal{B}_i$ \emph{contains} $k$ balls in color $\ell$, when $\mathcal{B}_i$ has at least one sub-bin $\mathcal{B}_i^j$ in color $\ell$, and the sub-bin is associated with $k$ non-zero elements in $\vect{\beta}^{(\ell)}$. Equivalently, the coded parameter vector $\tilde{\vect{\beta}}_i^\ell=\DIAG{\vect{h}_i}\vect{\beta}^{(\ell)}$ satisfies $\ABSL{\SUPP{\tilde{\vect{\beta}}_i^\ell}}=k$, $k\ge 0$.

First, we analyze the probability $Q_{\ell}$ that a particular bin $\mathcal{B}_i$ has color $\ell$. According to our model, the measurements are generated independently, therefore, we have
$$
Q_{\ell} = [1-(1-q_{\ell})^V][1-(1-q_{\ell})^R]^2.
$$
Then, we use $\xi_{k}^{(\ell)}$ to denote the probability of the event that a particular bin contains $k$ balls in color $\ell$. Since each ball is associated with $d$ bins among the $M$ bins independently and uniformly at random, the number of balls in color $\ell$ that a bin contains is binomial distributed with parameters $K_\ell$ and $\frac{d}{M}$, and we have
$$
\xi_{k}^{(\ell)} = Q_{\ell} {K_{\ell}\choose k} \left(\frac{d}{M}\right)^k \left(1-\frac{d}{M}\right)^{K_{\ell}-k}.
$$
In addition, we can use Poisson distribution to approximate the binomial distribution when $\lambda_{\ell}:=\frac{K_{\ell}d}{M}$ is a constant and $K_{\ell}$ approaches infinity. In the following analysis, we use the approximation
$$
\xi_{k}^{(\ell)} \approx Q_{\ell} \frac{\lambda_{\ell}^ke^{-\lambda_{\ell}}}{k!}.
$$
Consider the bipartite graph representing the association between the balls in color $\ell$ and the $M$ bins. We know that there are $K_{\ell}d$ edges connected to the balls in color $\ell$, and we use $\rho_{k}^{(\ell)}$ to denote the expected fraction of these $K_{\ell}d$ edges which are connected to a bin which contains $k$ balls in color $\ell$, $k\ge 1$. Then, we have
$$
\rho_{k}^{(\ell)} = \frac{kM}{K_{\ell}d}\xi_k^{(\ell)}=Q_{\ell}\frac{\lambda_{\ell}^{k-1}e^{-\lambda_{\ell}}}{(k-1)!},
$$
and equivalently, $\rho_{k}^{(\ell)}$ is also the probability that an edge, which is chosen from the $K_{\ell}d$ edges uniformly at random, is connected to a bin $\mathcal{B}_i$ containing $k$ balls in color $\ell$.

Let $q_s^{(\ell)}$ be the probability that a ball in color $\ell$ is a singleton ball. The event that this ball is a singleton ball is equivalent to the event that at least one of its $d$ associated bins contains one ball with color $\ell$. Then, when $K_\ell$ approaches infinity, we have
$$
q_s^{(\ell)} = 1-(1-\rho_1^{(\ell)})^d,
$$
and this is because in the limit $K_{\ell}\rightarrow\infty$, the correlations between the $d$ edges connected to a ball become negligible; this technique is often used in the theoretical analysis of density evolution in coding theory, and we use this type of asymptotic argument several times in the proofs. Let $K_s^{(\ell)}$ be the number of singleton balls in color $\ell$, then we have $\EXPL{K_s^{(\ell)}}=K_{\ell}q_s^{(\ell)}$.
Using the asymptotic argument and by Hoeffding's inequality, we also have for any constant $\delta>0$,
$$
\PROL{\ABSL{K_s^{(\ell)}-K_{\ell}q_s^{(\ell)}} \le \delta K_{\ell}} \ge 1-2\exp(-2\delta^2K_{\ell}),
$$
and this means that the number of singleton balls in color $\ell$ is highly concentrated around $K_{\ell}q_s^{(\ell)}$.
\end{proof}

\subsection{Initial Fractions}\label{sec:prf_init}
We construct the graph $\mathcal{G}_\ell$ whose nodes correspond to the singleton balls in color $\ell$ found in the previous stage, and analyze the number of edges in $\mathcal{G}_\ell$, which is equal to the number of strong doubletons in color $\ell$. Here, for clarification, we emphasize that the graph $\mathcal{G}_\ell$ corresponds to a sub-graph with the same color in Figure~\ref{fig:algorithm}, rather than the bipartite graph that we use to design the query vectors. Then, we can show that the number of strong doubletons is concentrated around a constant fraction of $M$ with high probability.
\begin{lemma}\label{lem:edge}
Let $M_s^{(\ell)}$ be the number of strong doubletons in color $\ell$ found in the second stage. Then, there exists a constant $\nu_\ell>0$ such that for any constant $\delta>0$,
\begin{equation}\label{eq:num_strong}
\PROL{\ABSL{ M_s^{(\ell)}-M\nu_\ell } \le \delta M} \ge1-2\exp(-2\delta^2 M). 
\end{equation}
\end{lemma}

\begin{proof}
We know that the expected number of doubletons in color $\ell$ is $M\xi_{2}^{(\ell)}$. Then, we analyze the probability that a doubleton is a strong doubleton. Similar to the analysis in \cite{pedarsani2014phasecode}, for a particular ball in color $\ell$, we let $B$ denote the event that this ball is in a singleton, and $D$ denote the event that this ball is in a doubleton. We have the conditional probability that a ball in a doubleton is also a singleton ball:
$$
\begin{aligned}
q_d^{(\ell)} &:= \PROL{B|D} = \frac{\PROL{D\bigcap B}}{\PROL{D}} \nonumber\\
&= \frac{ 1-\PROL{\bar{B}}-\PROL{\bar{D}}+\PROL{\bar{B}\bigcap \bar{D}} }{ 1-\PROL{\bar{D}} } \nonumber\\
&= \frac{ 1-(1-\rho_1^{(\ell)})^d-(1-\rho_2^{(\ell)})^d+(1-\rho_1^{(\ell)}-\rho_2^{(\ell)})^d }{ 1-(1-\rho_2^{(\ell)})^d }. \nonumber
\end{aligned}
$$
Then we know the probability that a doubleton is a strong doubleton is $(q_d^{(\ell)})^2$, and the expected number of strong doubletons in color $\ell$ is $M\xi_{2}^{(\ell)}(q_d^{(\ell)})^2$.
Let $\nu_\ell=\xi_{2}^{(\ell)}(q_d^{(\ell)})^2$ and $M_s^{(\ell)}$ be the number of edges in graph $\mathcal{G}_\ell$. The expectation of $M_s^{(\ell)}$ is $\EXPL{M_s^{(\ell)}}=M\nu_\ell$, and according to Hoeffding's inequality, we have for any $\delta>0$
$$
\PROL{\ABSL{ M_s^{(\ell)}-M\nu_\ell } \le \delta M}\ge1-2\exp(-2\delta^2M), 
$$
meaning that the number of edges is highly concentrated around $M\nu_\ell$. 
\end{proof}

Then, we get the following result on the size of the giant component of $\mathcal{G}_\ell$, using the asymptotic behavior of the Erdos-Renyi random graphs.
\begin{lemma}\label{lem:giant}
Let $K^{(\ell)}_G$ be the size of the largest connected component (giant component) of $\mathcal{G}_\ell$.
If the parameters of the Mixed-Coloring algorithm satisfy
\begin{equation}\label{eq:con_comp}
\frac{2M\nu_\ell}{K_{\ell}q_s^{(\ell)}}>1,
\end{equation}
then, for any constant $\delta>0$, with probability $1-\BIGO(1/K_{\ell})$, initial fraction of the balls in color $\ell$ which are recovered after the second stage satisfies
\begin{equation}\label{eq:init}
\ABS{\frac{K^{(\ell)}_G}{K_\ell}-\zeta_{\ell}q_s^{(\ell)}} \le \delta,
\end{equation}
where the constant $\zeta_{\ell}$ is the unique solution of the equation 
$$
\zeta_{\ell}+\exp\left( -2\frac{\zeta_{\ell}M\nu_\ell}{K_\ell q_s^{(\ell)}}\right)=1,
$$
and other connected components in $\mathcal{G}_\ell$ are of sizes $\BIGO(\log(K_{\ell}))$. 
\end{lemma} 

\begin{proof}
This result is a direct corollary of the asymptotic behavior of the Erdos-Renyi random graphs~\cite{erd6s1960evolution}, 
and we only give a brief proof here. First, we condition on the number of singleton balls that we find in the first stage, i.e., $K_s^{(\ell)}$ and the number of edges in $\mathcal{G}_\ell$, i.e., $M_s^{(\ell)}$. By symmetry, we know that the $M_s^{(\ell)}$ edges are uniformly chosen from the $K_s^{(\ell)}\choose 2$ possible edges. Therefore, the graph $\mathcal{G}_\ell$ is an Erdos-Renyi random graph. According to the results on the giant component of Erdos-Renyi random graphs, we know that if the limit
$$
\theta:=\lim_{K_s^{(\ell)}\rightarrow \infty} K_s^{(\ell)} \frac{M_s^{(\ell)}}{{{K_s^{(\ell)}}\choose 2}}>1,
$$
then with probability at least $1-\BIGO(1/K_s^{(\ell)})$, the size of the giant component of graph $\mathcal{G}_\ell$ is linear in $K_s^{(\ell)}$, and other connected components have sizes $\BIGO(\log(K_s^{(\ell)}))$. By (\ref{eq:num_single}) and (\ref{eq:num_strong}), we know that for any constant $\epsilon_1>0$, there exists a constant $\alpha_1>0$, such that, with probability at least $1-\BIGO(\exp(-\alpha_1 K_{\ell}))$, 
$$
K_s^{(\ell)} \in I_K = [(q_s^{(\ell)} - \epsilon_1)K_\ell, (q_s^{(\ell)} + \epsilon_1)K_\ell],
$$
and that, for any constant $\epsilon_2>0$, there exists a constant $\alpha_2>0$, such that, with probability at least $1-\BIGO(\exp(-\alpha_2 M))$,
$$
M_s^{(\ell)} \in I_M = [(\nu_\ell - \epsilon_2)M_\ell, (\nu_\ell + \epsilon_2)M_\ell].
$$
We also know that when $K_s^{(\ell)} \in I_K$ and $M_s^{(\ell)} \in I_M$ happen, the limit $\theta$ approaches $\frac{2M\nu_\ell}{K_{\ell}q_s^{(\ell)}}$. Let $A$ be the event that the size of the largest connected component (giant component) of $\mathcal{G}_\ell$, i.e., $K_G^{(\ell)}$ satisfies~\eqref{eq:init}, and other connected components in $\mathcal{G}_\ell$ are of sizes $\BIGO(\log(K_{\ell}))$. Then, according to the aforementioned property of Erdos-Renyi random graphs, conditioned on $K_s^{(\ell)} \in I_K$ and $M_s^{(\ell)} \in I_M$, we have
$$
\PROL{A\mid K_s^{(\ell)} \in I_K, M_s^{(\ell)} \in I_M} \ge 1-\BIGO(1/K_s^{(\ell)}).
$$
Then, we have
\begin{align*}
\PROL{\bar{A}} = &\PROL{\bar{A}\mid K_s^{(\ell)} \in I_K, M_s^{(\ell)} \in I_M}\PROL{K_s^{(\ell)} \in I_K, M_s^{(\ell)} \in I_M} \\
&+ \PROL{\bar{A}\mid K_s^{(\ell)} \notin I_K \text{ or } M_s^{(\ell)} \notin I_M}\PROL{K_s^{(\ell)} \notin I_K \text{ or } M_s^{(\ell)} \notin I_M} \\
\le & \BIGO(1/K_s^{(\ell)}) + \BIGO(\exp(-\alpha_1 K_{\ell})) + \BIGO(\exp(-\alpha_2 M)) \\
\le & \BIGO(1/K_s^{(\ell)}),
\end{align*}
which completes the proof.
\end{proof}

\subsection{Tree-like Assumption}\label{sec:tree}

By Lemma \ref{lem:giant}, we know that we can recover a constant fraction of the non-zero elements with probability $1-\BIGO(1/K_\ell)$. Then, we study the iterative decoding process. The analysis is based on density evolution, which is a common and powerful technique in coding theory. Similar to the density evolution analysis of many modern error-correcting codes~\cite{richardson2001capacity}, 
our derivation of density evolution is based on a tree-like assumption. Here, we state the tree-like assumption first and provide the results on the probability that the tree-like assumption holds.

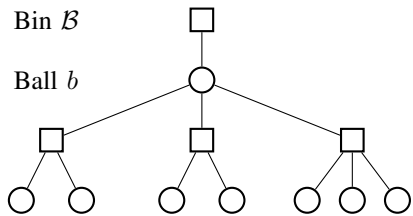
\begin{figure}[!h]
\centering
\begin{tikzpicture}[node distance=0.7cm,>=stealth',bend angle=45,auto]
\tikzstyle{ball}=[circle,thick,draw=black,fill=white,minimum size=3mm]
\tikzstyle{bin}=[rectangle,thick,draw=black,fill=white,minimum size=3mm]
\begin{scope}		
\node [bin] (C1) at (0.0, 0.0) {};
\node [text width=4cm] at (-0.5, 0.0) {Bin $\mathcal{B}$};
\node [ball] (V1) at (0.0, -0.8) {};
\node [text width=4cm] at (-0.5,-0.8) {Ball $b$};
\node [bin] (C2) at (-2.0, -1.6) {};
\node [bin] (C3) at (0.0, -1.6) {};
\node [bin] (C4) at (2.0, -1.6) {};
\node [ball] (V2) at (-2.4, -2.4) {};
\node [ball] (V3) at (-1.6, -2.4) {};
\node [ball] (V4) at (-0.4, -2.4) {};
\node [ball] (V5) at (0.4, -2.4) {};
\node [ball] (V6) at (1.4, -2.4) {};
\node [ball] (V7) at (2.0, -2.4) {};
\node [ball] (V8) at (2.6, -2.4) {};
\path (V1) edge (C1);
\path (V1) edge (C2);
\path (V1) edge (C3);
\path (V1) edge (C4);
\path (C2) edge (V2);
\path (C2) edge (V3);
\path (C3) edge (V4);
\path (C3) edge (V5);
\path (C4) edge (V6);
\path (C4) edge (V7);
\path (C4) edge (V8);
\end{scope}
\end{tikzpicture}
\caption{\small{Level-2 neighborhood of edge $(b, \mathcal{B})$.}}
\label{fig:density}
\end{figure}

As we have mentioned, the association between the balls in color $\ell$ (non-zero elements in $\vect{\beta}^{(\ell)}$) and the bins can be represented by a $d$-left regular bipartite graph. We label the edges by an ordered pair of a ball $b$ and a bin $\mathcal{B}$, denoted by $e=(b, \mathcal{B})$. We define the \emph{level}-$C^*$ neighborhood of $e$, denoted by $N_e^{C^*}$ as the subgraph of all the edges and nodes on paths with length less than or equal to $C^*$, which start from $b$ and the first edge of the paths are not $e$~\cite{pedarsani2014phasecode}. We have the following results on the probability that $N_e^{C^*}$ is a tree, or equivalently, cycle-free, for a constant $C^*$.
\begin{lemma}\label{lem:tree}
\cite{pedarsani2014phasecode} For a fixed constant $C^*$, $N_e^{2C^*}$ is a tree with probability at least $1-\BIGO(\log(K_\ell)^{C^*}/K_\ell)$.
\end{lemma}
We conduct the density evolution analysis conditioned on the event that $N_e^{2C^*}$ is a tree for an edge $e$ which is chosen from the $K_\ell d$ edges uniformly at random. Then, we take the complementary event into consideration and complete the analysis.

\subsection{Density Evolution}
Recall that in the first iteration, we find all the singletons, and in the second iteration, we find the strong doubletons and form the giant component. Let $p_j^{(\ell)}$ be the probability that at the $j$th iteration of the learning algorithm, a ball in color $\ell$, which is chosen from the $K_\ell$ balls uniformly at random, is not recovered, $j\ge 2$. Here, $p_2^{(\ell)}$ corresponds to the probability that after the second iteration, a randomly chosen ball in color $\ell$ is not in the giant component. According to the previous section, we know that by choosing parameters which satisfy (\ref{eq:con_comp}), we have $p_2^{(\ell)}=\frac{K_{\ell}-K_G^{(\ell)}}{K_{\ell}}=\Theta(1)$ with probability $1-\BIGO(1/K_\ell)$. Now we analyze the relationship between $p_{j+1}^{(\ell)}$ and $p_j^{(\ell)}$ for $j\ge 2$ .

Consider the iterative decoding process as a \emph{message passing} process. First, we know that at iteration $j+1$, a ball in color $\ell$ passes a message to a bin through an edge claiming that it is colored, if and only if at least one of the other $d-1$ neighborhood bins contains a resolvable multiton in color $\ell$. Second, a sub-bin in color $\ell$ becomes a resolvable multiton if and only if all the other balls in this sub-bin are colored. This message passing process is illustrated in Figure \ref{fig:density}. Under the tree-like assumption, the messages passed among the balls and bins are independent, we have
$$
p_{j+1}^{(\ell)}=(1-\sum_{i=2}^\infty \rho_i^{(\ell)}(1-p_j^{(\ell)})^{i-1} )^{d-1},
$$
which gives us
\begin{equation}\label{eq:density}
p_{j+1}^{(\ell)}=(1-Q_{\ell}(e^{-\lambda_{\ell} p_j^{(\ell)}}-e^{-\lambda_{\ell}}))^{d-1}.
\end{equation}
As we can see, the major difference between the density evolution of the Mixed-Coloring algorithm and the PhaseCode algorithm in~\cite{pedarsani2014phasecode} (for compressive phase retrieval via sparse-graph codes) is that there is a constant probability $Q_{\ell}$ that a bin has a sub-bin in color $\ell$. 

Next, we show that after a constant number of iterations, $p_j^{(\ell)}$ can be arbitrarily small.
\begin{lemma}\label{lem:errorfloor}
If we choose parameters satisfying
\begin{equation}\label{eq:derive_cond}
(d-1)Q_{\ell}\lambda_{\ell} e^{-\lambda_{\ell} t}>1,
\end{equation}
then for any constant $\delta>0$, there exists a constant $T$, such that $p_T^{(\ell)}<\delta$.
\end{lemma}
\begin{proof}
Let $f_{\ell}(t)=(1-Q_{\ell}(e^{-\lambda_{\ell} t}-e^{-\lambda_{\ell}}))^{d-1}$, then we have $p_{j+1}^{(\ell)}=f_{\ell}(p_j^{(\ell)})$. It is easy to see that $f_{\ell}(1)=1$, $f_{\ell}(0)>0$, and $f_{\ell}$ is a monotonically increasing function. We also have
$$
f_{\ell}^\prime(t) = (d-1)Q_{\ell}\lambda_{\ell} e^{-\lambda_{\ell} t}(1-Q_{\ell}(e^{-\lambda_{\ell} t}-e^{-\lambda_{\ell}}))^{d-2}.
$$
We know that if there is
\begin{equation}\label{eq:derivitive}
f_{\ell}^\prime(1) = (d-1)Q_{\ell}\lambda_{\ell} e^{-\lambda_{\ell} t}>1,
\end{equation}
then there exists at least one fixed point $t\in(0,1)$ such that $f_{\ell}(t)=t$. We use $p_{\ell}^*$ to represent the largest fixed point of $f_{\ell}(t)$ in $(0,1)$. 
Now we argue that the fixed point can be made arbitrarily small by choosing proper parameters. Suppose that for a certain set of parameters $\lambda_{\ell}$ and $d$, the fixed point is $p_{\ell}^*$, then if we keep $\lambda_{\ell}$ and increase $d$ to $\tilde{C}d$, where $\tilde{C}>1$ is a constant, then we can see that the new fixed point is upper bounded by $(p_{\ell}^*)^{\tilde{C}}$, and in this way, the fixed point can be made an arbitrarily small constant.
As shown in \cite{pedarsani2014phasecode}, as long as we can choose parameters to make the fixed point $p_{\ell}^*<\delta/2$, then, there exists a constant number of iterations $T$, depending on $\delta$, such that $p_T^{(\ell)}<\delta$.

Then, we investigate how the sample complexity depends on $p_{\ell}^*$. First, since $p_{\ell}^*$ is a fixed point of the iteration~\eqref{eq:density}, we have
$$
(1-Q_{\ell}(e^{-\lambda_{\ell} p_{\ell}^*}-e^{-\lambda_{\ell}}))^{d-1} = p_{\ell}^*.
$$
Since $p_{\ell}^*$ is usually very small, we use the approximation $e^{-\lambda_{\ell} p_{\ell}^*} \approx 1$, and thus we have
$$
(1-Q_{\ell}(1-e^{-\lambda_{\ell}}))^{d-1} \approx p_{\ell}^*,
$$
which gives us $d =\BIGO( \log(1/p_{\ell}^*) )$. Further, since we keep $\lambda_\ell = \frac{K_\ell d}{M}$ as a constant, we know that $M = \BIGO( \log(1/p_{\ell}^*) )$ as  a function of $p_{\ell}^*$.

\end{proof}

Then, we can prove the following lemma showing that the number of uncolored balls in color $\ell$ is concentrated around $K_{\ell}p_T^{(\ell)}$ with high probability.
\begin{lemma}\label{lem:concentration}
Let $Z_{\ell}$ be the number of uncolored balls in color $\ell$ after $T$ iterations. Then for any $\delta>0$, there exists constant $c_1$, such that when conditioned on the event that $p_2^{(\ell)}=\Theta(1)$, and $K_\ell$ is large enough,
\begin{equation}\label{eq:finalexp}
\ABS{\EXP{Z_{\ell}}-K_{\ell}p_T^{(\ell)}}<K_{\ell}\delta/2,
\end{equation}
\begin{equation}\label{eq:finalconcentration}
\PRO{\ABS{Z_{\ell}-K_{\ell}p_T^{(\ell)}}>K_{\ell}\delta}<2\exp\{-c_1 \delta^2 K_{\ell}^{1/(4T+1)}\}.
\end{equation}
\end{lemma}
The proof of Lemma \ref{lem:concentration} is the same as in \cite{pedarsani2014phasecode}, and uses Doob's martingale argument and Azuma's concentration bound. We should also notice that the event that the tree-like assumption does not hold is already considered in (\ref{eq:finalexp}). Now combining Lemmas \ref{lem:giant}, \ref{lem:errorfloor}, and \ref{lem:concentration}, we have shown that for a specific $\ell\in[L]$, there exists proper parameters of the algorithm such that after a constant number of iterations, the Mixed-Coloring algorithm can recover an arbitrarily large fraction of the balls in color $\ell$ with probability $1-\BIGO(1/K_{\ell})$. Since $L$ is a constant and $K_\ell = \Theta(K)$, the results above implies that for an arbitrarily small constant $p^*\in(0,1)$,
$$
\mathbb{P} \big\{\ABSL{\SUPP{\hat{\vect{\beta}}^{(\ell)}}} \ge (1-p^*)\ABSL{\SUPP{\vect{\beta}^{(\ell)}}} \big\} \ge 1- \BIGO(1/K).
$$
Then, we turn to the first and the third properties in Theorem~\ref{thm:main_noiseless}. According to our ratio test scheme, as long as we have a singleton, we find the exact location and value of the non-zero element, and thus our algorithm has no false discovery. As for the element-wise recovery, one can see that due to the use of $d$-left regular random bipartite graph (each left node is connected to $d$ right nodes uniformly at random), the recovered $(1-p^*)$ fraction of the support is also uniformly distributed on the support of $\vect{\beta}^{(\ell)}$. Thus, for each $j\in\SUPP{\vect{\beta}^{(\ell)}}$, $\PROL{\hat{\beta}^{(\ell)}_j = \beta^{(\ell)}_j~|~\ABSL{\SUPP{\hat{\vect{\beta}}^{(\ell)}}} \ge (1-p^*)\ABSL{\SUPP{\vect{\beta}^{(\ell)}}} }\ge 1-p^*$. Then, by total law of probability, we have
\begin{equation}\label{eq:element_wise}
\begin{aligned}
\PROL{\hat{\beta}^{(\ell)}_j \neq \beta^{(\ell)}_j} =& \PROL{\hat{\beta}^{(\ell)}_j \neq \beta^{(\ell)}_j~|~\ABSL{\SUPP{\hat{\vect{\beta}}^{(\ell)}}} \ge (1-p^*)\ABSL{\SUPP{\vect{\beta}^{(\ell)}}} } \PROL{\ABSL{\SUPP{\hat{\vect{\beta}}^{(\ell)}}} \ge (1-p^*)\ABSL{\SUPP{\vect{\beta}^{(\ell)}}}} \\
&+\PROL{\hat{\beta}^{(\ell)}_j \neq \beta^{(\ell)}_j~|~\ABSL{\SUPP{\hat{\vect{\beta}}^{(\ell)}}} < (1-p^*)\ABSL{\SUPP{\vect{\beta}^{(\ell)}}} } \PROL{\ABSL{\SUPP{\hat{\vect{\beta}}^{(\ell)}}} < (1-p^*)\ABSL{\SUPP{\vect{\beta}^{(\ell)}}}} \\
\le& p^* + \BIGO(1/K).
\end{aligned}
\end{equation}
Thus, we have proved the three properties in Theorem~\ref{thm:main_noiseless}.

\subsection{Time Complexity}\label{sec:prf_time_complexity}
In this section, we analyze the time complexity of the algorithm. First, note that there are $M=\Theta(K)$ bins and each bin has a constant number of sub-bins. Since refining the measurements of each bin takes $\Theta(1)$ operations, the time complexity of refining measurements is $\Theta(K)$. Next, to find all the singletons, we need to check all the colored sub-bins, and checking each sub-bin takes $\Theta(1)$ operations, the time complexity of this stage is $\Theta(K)$. In the third stage, we find all the strong doubletons. We know that there are $\Theta(K)$ singleton balls and for each singleton ball, there are $d$ bins connected to it. For each of the bins, we subtract the measurements contributed by the singleton ball from the refined measurements in the sub-bins, and do the ratio test to see if it is a strong doubleton. Therefore, processing each bin takes $\Theta(1)$ operations and since $d$ is also a constant, the time complexity of finding strong doubletons is also $\Theta(K)$. Then, we get the graph with $\Theta(K)$ nodes and $\Theta(K)$ edges, corresponding to the singleton balls and strong doubletons, respectively. Using breadth-first search algorithm, the time complexity of finding the connected components is $\Theta(K)$.
In the last stage, we iteratively find other uncolored balls. For each unprocessed sub-bin, since we do not know the color of the sub-bin, there are $L$ possible remaining measurements. Each time when we find a new ball, we update at most $dV$ remaining measurements and do the ratio test. Therefore, it takes $\Theta(1)$ operations when coloring a new ball. Since there are $\Theta(K)$ uncolored balls after finding the giant components, the time complexity of the last stage is also $\Theta(K)$. Thus, we have shown that the time complexity of Mixed-Coloring algorithm is $\Theta(K)$, which completes the proof of Theorem \ref{thm:main_noiseless}.

\section{Computing the Constants in the Sample Complexity}\label{sec:constant}
In this section, we give exact constants in the sample complexity results. For simplicity, we assume that $K_{\ell}=K/L$ and $q_{\ell}=1/L$ for all $\ell\in[L]$. We let $c:=M/K$, and thus we have $\lambda_{\ell}=\frac{K_{\ell}d}{M} =\frac{d}{Lc}$. We analyze the minimum number of measurements that we need to reach a certain reliability target. More precisely, we set the maximum error floor to be $p_{\max}^*$, and numerically calculate the error floor for different values of $d$, $c$, $R$, and $V$. Then, we minimize the number of total measurements, which is proportional to $(2R+V)c$ with the constraint that the error floor $p^*\le p_{\max}^*$. As we have shown in previous parts, the parameters should also satisfy (\ref{eq:con_comp}) and (\ref{eq:derive_cond}). We know that if (\ref{eq:con_comp}) is satisfied, when $K$ is large enough, there should be a giant component with size linear in $K$ for each color, where $\theta>1$ is a threshold that we can choose. Therefore, we select optimal parameters with three constraints, which are (\ref{eq:derive_cond}), (\ref{eq:con_comp}), and $p^*\le p_{\max}^*$.

\begin{table*}[!t]
\caption{Constants in the results of sample complexity.}\label{tab:const}
\centering
\begin{tabular}{c|c|cccccccc}
\hline
\multirow{6}{0.08\linewidth}{$L=2$} & $d$  & 11 & 12 & 13 & 14 & \textbf{15} & 16 & 17 & 18 \\
& $p^*/10^{-6}$  & 6.7 & 8.7 & 1.9 & 3.1 & \textbf{5.1} & 1.6 & 0.5 & 7.4 \\ 
& $M/K$  & 2.95 & 3.17 & 3.23 & 3.46 & \textbf{3.71} & 3.78 & 3.86 & 4.37 \\
& $R$  & 4 & 4 & 4 & 3 & \textbf{3} & 3 & 3 & 3\\
& $V$ & 4 & 3 & 3 & 4 & \textbf{3} & 3 & 3 & 2 \\
& $m/K$ & 35.4 & 34.87 & 35.53 & 34.6 & \textbf{33.39} & 34.02 & 34.74 & 34.96\\
\hline
\multirow{6}{0.08\linewidth}{$L=3$} & $d$  & 11 & 12 & 13 & 14 & \textbf{15} & 16 & 17 & 18\\
& $p^*/10^{-6}$  & 4.4 & 5.2 & 2.7 & 9.2 & \textbf{8.8} & 2.8 & 6.2 & 2.3\\ 
& $M/K$  & 1.94 & 2.08 & 2.17 & 2.39 & \textbf{2.52} & 2.56 & 2.76 & 2.81\\
& $R$ & 7 & 6 & 6 & 5 & \textbf{5} & 5 & 5 & 5\\
& $V$  & 7 & 7 & 6 & 6 & \textbf{5} & 5 & 4 & 4\\
& $m/K$ & 40.74 & 39.52 & 39.06 & 38.24 & \textbf{37.80} & 38.4 & 38.64 & 39.34\\
\hline
\multirow{6}{0.08\linewidth}{$L=4$} & $d$ & 11 & 12 & \textbf{13} & 14 & 15 & 16 & 17 & 18\\
& $p^*/10^{-6}$ & 7.8 & 8.7 & \textbf{8.1} & 5.6 & 4.2 & 3.3 & 4.0 & 5.0\\ 
& $M/K$ & 1.48 & 1.59 & \textbf{1.68} & 1.76 & 1.85 & 1.93 & 2.04 & 2.16\\
& $R$  & 9 & 9 & \textbf{8} & 8 & 7 & 7 & 7 & 6\\
& $V$ & 11 & 8 & \textbf{8} & 7 & 8 & 7 & 6 & 7\\
& $m/K$ & 42.92 & 41.34 & \textbf{40.32} & 40.48 & 40.7 & 40.53 & 40.8 & 41.04\\
\hline
\end{tabular}
\end{table*}

The results of the numerical calculation are shown in Table \ref{tab:const}. In these experiments, we set $p_{\max}^*=10^{-5}$, $\theta=2$, and we fix the left degree $d$ and choose different values of $c$, $R$, and $V$ to minimize the number of measurements with the three constraints. Then we compare the optimal number of measurements over different choices of $d$ and find the optimal $d$. As we can see, to reach the same reliability level, for $L=2,3,4$, the optimal number of measurements we need is $33.39K$, $37.80K$, and $40.32K$, respectively. The number of measurements we need only increases slightly with $L$, and the optimal $d$ is around 13 and 15.
\section{Proof of Theorem~\ref{thm:main_noisy}}\label{sec:prf_noisy}
In this section, we analyze the performance of the Robust Mixed-Coloring algorithm and prove Theorem~\ref{thm:main_noisy}. Recall that the overall structure of the Robust Mixed-Coloring algorithm is the same as its noiseless counterpart. Suppose that one can always perfectly find the consistent sets of measurements, and the correct location and value of the non-zero elements, then, the recovery guarantee in the noisy setting will be exactly the same as in the noiseless setting. Further, in the noisy setting, finding the correct consistent sets of measurements, and the correct location and value of the non-zero elements relies on the success of two events: 1) the EM-based algorithm has to always find the correct denoised measurements, and 2) the verification procedure has to identify all the singletons, and cannot misclassify other consistent sets as singletons.

We provide details below. We define error events: $E_\ell^{1}$, as the event that there exists one incidence that the EM algorithm does not find the correct denoised measurements, and $E_\ell^2$, as the event that there exists one incidence where the verification query vectors make a misclassification between singleton and non-singleton. According to Theorem~\ref{thm:em}, the failure probability of each EM operation is $\BIGO(1/\poly(n))$. According to the proof in Appendix~\ref{sec:prf_time_complexity}, there are $\Theta(K)$ bin-level operations during the algorithm, and when processing each bin, we need $\Theta(\log^2(n))$ EM operations, since there are $\Theta(\log^2(n))$ query vectors in each bin. Therefore, the total amount of EM operations is $\Theta(K\log^2(n))$. By union bound, we know that $\PROL{E_\ell^{1}} \le \BIGO(K\log^2(n)/\poly(n)) = \BIGO(1/\poly(n))$. According to Lemma~\ref{lem:verification}, we know that each verification has failure probability $\BIGO(1/\poly(n))$, and using a similar union bound argument, we know that $\PROL{E_\ell^{2}} \le \BIGO(1/\poly(n))$. When both error events $E_\ell^{1}$ and $E_\ell^{2}$ do not happen, the algorithm always find the correct location and value of the recovered elements, and in this case there is no false discovery. By union bound, we know that $\PROL{\overline{E_\ell^{1}} \cap \overline{E_\ell^{2}}} \ge 1 - \BIGO(1/\poly(n))$. Therefore, the probability that there is no false discovery is $1 - \BIGO(1/\poly(n))$, which proves the first property in the theorem.

Then, we turn to prove the second property. We define the error event $E_\ell$ that fewer than $1-p^*$ fraction of the $K_\ell$ non-zero elements of the parameter vector $\vect{\beta}^{(\ell)}$ are recovered by the algorithm. Suppose that none of $E_\ell^{1}$ and $E_\ell^{2}$ happens, then, the analysis of the robustified algorithm becomes exactly the same as in the noiseless setting. Therefore, according to Theorem~\ref{thm:main_noiseless}, we know that $\PROL{E_\ell | \overline{E_\ell^{1}} \cap \overline{E_\ell^2}} \le \BIGO(1/K_\ell)$. Then, we can apply total law of probability and get
\begin{align*}
\PROL{E_\ell} &= \PROL{E_\ell | \overline{E_\ell^{1}} \cap \overline{E_\ell^2}} \PROL{\overline{E_\ell^{1}} \cap \overline{E_\ell^2}}+\PROL{E_\ell | E_\ell^1\cup E_\ell^2}\PROL{E_\ell^1\cup E_\ell^2}  \\
& \le \PROL{E_\ell | \overline{E_\ell^{1}} \cap \overline{E_\ell^2}} + \PROL{E_\ell^1\cup E_\ell^2} \\
& \le \PROL{E_\ell | \overline{E_\ell^{1}} \cap \overline{E_\ell^2}} + \PROL{E_\ell^1} + \PROL{E_\ell^2} \\
& = \BIGO(1/K_\ell) \\
&= \BIGO(1/K),
\end{align*}
which proves the second property in the theorem. The third property in the theorem can be derived using the method in~\eqref{eq:element_wise} in the proof of Theorem~\ref{thm:main_noiseless}, and we omit the details here. Thus, we have proved the three properties in Theorem~\ref{thm:main_noisy}. The time complexity can be analyzed using the same method as in the noiseless case, provided in Appendix~\ref{sec:prf_time_complexity}. The only difference is that, the bin-level operation takes $\Theta(1)$ time in the noiseless setting, while in the noisy setting it takes $\Theta(\polylog(n))$ time. Therefore, the time complexity of the Robust Mixed-Coloring algorithm is $\Theta(K\polylog(n))$.

\section{Proof of Lemma~\ref{lem:verification}}\label{sec:prf_verification}
We first provide a simplified interpretation of Lemma~\ref{lem:verification}.
\begin{lemma}\label{lem:verification_2}
Let $\mat{V}\in\{0,1\}^{P_2\times n}$ be a Rademacher matrix with $P_2=\Theta(\log(n))$. Denote the $j$-th column of $\mat{V}$ by $\vect{v}_j$. Suppose that $\vect{h}\in\{0,1\}^n$ and $\vect{\beta}\in\mathbb{D}^n$, where $\mathbb{D} = \{\pm\Delta, \pm2\Delta,\ldots,\pm b\Delta\}$.
Let $\vect{y} = \mat{V}\DIAG{\vect{h}}\vect{\beta}$. Suppose that $\DIAG{\vect{h}}\vect{\beta} \neq a\Delta\vect{e}_j$ for some $a\Delta \in \mathbb{D}$ and canonical basis vector $\vect{e}_j$. Then, with probability $1-\BIGO(1/\poly(n))$, $\vect{y} \neq a\Delta \vect{v}_j$.
\end{lemma}
Here, $\vect{\beta}$ is the parameter vector that generates the consistent set of measurements, and $\vect{h}$ denotes the association between the bin and the coordinates. Define $\tilde{\vect{\beta}} = \DIAG{\vect{h}}\vect{\beta} \in \mathbb{D}^n$, and we have $\vect{y} = \mat{V}\tilde{\vect{\beta}} $. Our goal is to justify that, when $\tilde{\vect{\beta}}\neq a\Delta\vect{e}_j$, with high probability, $\vect{y} \neq a\Delta \vect{v}_j$.

Suppose that $\tilde{\vect{\beta}} \neq a\Delta\vect{e}_j$ but $\vect{y} = a\Delta \vect{v}_j$. Then we have $\mat{V}(\tilde{\vect{\beta}} - a\Delta\vect{e}_j) = \vect{0}$. According to a corollary of the Johnson-Lindenstrauss Lemma~\cite{johnson1984extensions} (one can also refer to Section 4 in~\cite{baraniuk2008simple}), we know that
\begin{align*}
\PROL{\mat{V}(\tilde{\vect{\beta}} - a\Delta\vect{e}_j) = \vect{0}} \le & \PRO{\ABS{ \|\frac{1}{\sqrt{P_2}}\mat{V}(\tilde{\vect{\beta}} - a\Delta\vect{e}_j)\|_2^2 -  \|\tilde{\vect{\beta}} - a\Delta\vect{e}_j\|_2^2} \ge \frac{1}{2}\|\tilde{\vect{\beta}} - a\Delta\vect{e}_j\|_2^2} \le 2e^{-P_2/24}.
\end{align*}
Therefore, we can see that by having $P_2=\Theta(\log(n))$ verification query vectors, we can guarantee that with probability at least $1-\BIGO(1/\poly(n))$, we won't identify $\tilde{\vect{\beta}}$ as $a\Delta\vect{e}_j$, and this means that we won't misclassify a non-singleton as a singleton.

\section{Proof of Theorem~\ref{thm:em}}\label{sec:em}
In this section, we provide a method to estimate the parameters of a mixture of two Gaussian random variables, and give the theoretical analysis to prove Theorem~\ref{thm:em}. This estimation method is based on EM algorithm with method of moments initialization. 

Recall the setting of Theorem~\ref{thm:em}. Let $z_i$'s be i.i.d. samples of Bernoulli$(\frac{1}{2})$ distribution, and $w_i$'s be i.i.d. samples of Gaussian distribution with mean zero and variance $\sigma^2$, independently of $z_i$'s, $i\in[N]$. Suppose that random variables $y_i$'s are generated in the following way:
$$
y_i=\mu_1(1-z_i)+\mu_2z_i+w_i,~i\in[N].
$$
Then, we can consider $y_i$ as a mixture of two Gaussian random variables with means $\mu_1$ and $\mu_2$, respectively. We assume that $\sigma^2$ is known, and the parameters $\mu_1$ and $\mu_2$ are unknown and take value in a finite and quantized set
$\mathbb{D}=\{k\Delta:k\in\mathbb{Z}, \ABSL{k}\le b\}$, for some $\Delta>0$. Without loss of generality, we assume that $\mu_1\le\mu_2$. (Note that we allow $\mu_1=\mu_2$ here.) Our goal is to get accurate estimation of $\mu_1$ and $\mu_2$.

The first step is to compute the sample mean of the first $N_1$ samples, i.e., $\bar{y}=\frac{1}{N_1}\sum_{i=1}^{N_1} y_i$. Since we know that the mean of $y_i$'s takes value in the set $\mathbb{D}^+ = \{\frac{k}{2}\Delta:k\in\mathbb{Z}, \ABSL{k}\le 2b\}$, we find the element in $\mathbb{D}^+$ which is the closest one to $\bar{y}$ as the estimator of the mean of $y_i$, i.e., $\frac{1}{2}(\mu_1+\mu_2)$,
$$
\hat{\mu} = \arg\min_{\mu\in\mathbb{D}^+}\ABSL{\bar{y}-\mu}.
$$
We have the following result on the accuracy of the estimator $\hat{\mu}$.
\begin{lemma}\label{lem:hatmu}
There exist universal constants $c_1$ and $c_2$ such that for any $\delta>0$, if 
\begin{equation}\label{eq:lem1n1}
N_1\ge \max\{c_1b^2, c_2\frac{\sigma^2}{\Delta^2}\}\log(\frac{1}{\delta}),
\end{equation}
we have $\hat{\mu}=\frac{1}{2}(\mu_1+\mu_2)$ with probability at least $1-6\delta$. 
\end{lemma}
We prove Lemma \ref{lem:hatmu} in Appendix \ref{prf:hatmu}. In the second step, we subtract $\hat{\mu}$ from the other $N-N_1$ samples, and get centered random variables $\tilde{y}_i=y_i-\hat{\mu}$, $i=N_1+1,\ldots,N$. We assume that $\hat{\mu}$ is the actual mean of the $y_i$'s, meaning that $\hat{\mu}=\frac{1}{2}(\mu_1+\mu_2)$. Then, we know that if $\mu_1=\mu_2$, the centered random variables $\tilde{y}_i$'s are i.i.d. $\mathcal{N}(0, \sigma^2)$ distributed; otherwise $\tilde{y}_i$'s are i.i.d. mixtures of two Gaussian distributions
$$
\tilde{y}_i\sim\begin{cases}
\mathcal{N}(\theta_*, \sigma^2)&\quad\text{with probability }\frac{1}{2}\\
\mathcal{N}(-\theta_*, \sigma^2)&\quad\text{with probability }\frac{1}{2},
\end{cases}
$$
where $\theta_*=\frac{1}{2}(\mu_2-\mu_1)\ge 0$. Then, we make an initial estimation of $\theta_*$ using $N_2$ of the $N-N_1$ centered random variables. Specifically, we compute
$$
\theta_0 = \begin{cases}
\sqrt{\frac{1}{N_2}\sum_{i=N_1+1}^{N_1+N_2} \tilde{y}_i^2 - \sigma^2} &\quad\text{if }\frac{1}{N_2}\sum_{i=N_1+1}^{N_1+N_2} \tilde{y}_i^2 - \sigma^2>0\\
0 &\quad\text{otherwise}.
\end{cases}
$$
We have the following result on $\theta_0$:
\begin{lemma}\label{lem:theta0}
Condition on the event that $\hat{\mu}=\frac{1}{2}(\mu_1+\mu_2)$. There exist universal constants $c_3$ and $c_4$, such that for any $\delta>0$, when
\begin{equation}\label{eq:lem2n2}
N_2\ge\max\{c_3\frac{\sigma^2}{\Delta^2}(1+\frac{4\sigma^2}{\Delta^2}), c_4\}\log(\frac{1}{\delta}),
\end{equation}
then $\theta_0$ satisfies:\\
(1) if $\mu_1=\mu_2$, $\theta_0<\frac{\Delta}{4}$ with probability at least $1-2\delta$;\\
(2) if $\mu_1\neq\mu_2$, $\ABSL{\theta_0-\theta_*}<\frac{\theta_*}{4}$ with probability at least $1-2\delta$.
\end{lemma}
We prove Lemma \ref{lem:theta0} in Appendix \ref{prf:theta0}. If $\theta_0<\frac{\Delta}{4}$, we claim that $\mu_1=\mu_2$, and give estimators $\hat{\mu}_1=\hat{\mu}_2=\hat{\mu}$. Otherwise, we run a standard EM algorithm with the remaining $N_3:=N-(N_1+N_2)$ samples using $\theta_0$ as an initialization to estimate $\theta_*$. Here, we briefly review the procedures of standard EM algorithm for mixtures of Gaussian distributions. For $t=0,1,2,\ldots$, conduct the following two steps:\\
\textbf{E step:} compute the expected log-likelihood.
$$
L(\theta|\theta_t) = -\frac{1}{2N_3} \sum_{i=N_1+N_2+1}^N [p(\tilde{y}_i|\theta_t) (\tilde{y}_i-\theta_t)^2 + (1-p(\tilde{y}_i|\theta_t))(\tilde{y}_i+\theta_t)^2],
$$
where
$$
p(y|\theta_t) = e^{-\frac{(y-\theta_t)^2}{2\sigma^2}}\left[ e^{-\frac{(y-\theta_t)^2}{2\sigma^2}}+e^{-\frac{(y+\theta_t)^2}{2\sigma^2}} \right]^{-1}.
$$
\textbf{M step:} compute
$$
\theta_{t+1} = \arg\max_{\theta}L(\theta|\theta_t)=\frac{1}{N_3}\left[ 2\sum_{i=N_1+N_2+1}^N p(\tilde{y}_i|\theta_t)\tilde{y}_i - \sum_{i=N_1+N_2+1}^N \tilde{y}_i \right].
$$
We run the EM algorithm for $T$ iterations, and find the element in $\mathbb{D}^+$ which is the closest one to $\theta_t$ as the estimator of the mean of $\theta_*$, i.e., $\hat{\theta}_*=\arg\min_{\theta\in\mathbb{D}^+}\ABSL{\theta-\theta_T}$. Then, we output the estimation of $\mu_1$ and $\mu_2$ by $\hat{\mu}_1 = \hat{\mu}-\hat{\theta}_*$ and $\hat{\mu}_2 = \hat{\mu}+\hat{\theta}_*$.

Here, we review the results in \cite{balakrishnan2014statistical} which characterizes the performance of the EM algorithm.
\begin{lemma}\label{lem:em_siva}\cite{balakrishnan2014statistical}
Suppose that $\mu_1<\mu_2$. Conditioned on the event that $\hat{\mu}=\frac{1}{2}(\mu_1+\mu_2)$ and the event that $\ABSL{\theta_0-\theta_*}<\frac{\theta_*}{4}$. Suppose that $\eta:=\frac{\theta_*}{\sigma} \ge \frac{4}{\sqrt{3}}$. Then, there exist universal constants $c_5$, $c_6$, and $c_7$, such that when $N_3\ge c_5\log(\frac{1}{\delta})$, for any $\delta>0$, we have
$$
\ABSL{\theta_t-\theta_*} \le \kappa^t \ABSL{\theta_0-\theta_*} + \frac{c_6}{1-\kappa}\theta_*\sqrt{\theta_*^2+\sigma^2}\sqrt{\frac{1}{N_3}\log(\frac{1}{\delta})},
$$
with probability at least $1-\delta$, where $\kappa\le \exp(-c_7\eta^2)$.
\end{lemma}
Then, we have the direct corollary:
\begin{corollary}\label{cor:em}
Under the same condition that $\hat{\mu}=\frac{1}{2}(\mu_1+\mu_2)$, $\ABSL{\theta_0-\theta_*}<\frac{\theta_*}{4}$, and that $\eta=\frac{\theta_*}{\sigma} \ge \frac{4}{\sqrt{3}}$ as in Theorem \ref{thm:em}, then, when 
\begin{equation}\label{eq:n3}
N_3>\max\{c_5,\frac{16c_6^2}{(1-\kappa)^2}b^2(b^2\Delta^2+\sigma^2)\}\log(\frac{1}{\delta}),
\end{equation}
and 
\begin{equation}\label{eq:t}
T>\frac{\log(b)}{\log(1/\kappa)},
\end{equation}
we have $\hat{\theta}_*=\theta_*$ with probability at least $1-\delta$, for any $\delta>0$.
\end{corollary}
We prove Corollary \ref{cor:em} in Appendix \ref{prf:cor}. We have the following theorem to characterize the performance of the proposed estimation algorithm.
\begin{theorem}
If $N_1$, $N_2$, $N_3$, and $T$ satisfy (\ref{eq:lem1n1}), (\ref{eq:lem2n2}), (\ref{eq:n3}), and (\ref{eq:t}), respectively, and $\frac{\Delta}{\sigma}  \ge \frac{4}{\sqrt{3}} $, then the proposed estimation algorithm outputs correct estimations $\hat{\mu}_1=\mu_1$ and $\hat{\mu}_2=\mu_2$ with probability at least $1-9\delta$, for any $\delta>0$.
\end{theorem}

\begin{proof}
Let $A_1$ and $A_2$ be the events that $\hat{\mu}=\frac{1}{2}(\mu_1+\mu_2)$ and that $\hat{\theta}_*=\theta_*$, respectively, and $A$ be the event that $\hat{\mu}_1=\mu_1$ and $\hat{\mu}_2=\mu_2$. Then, by Lemma \ref{lem:hatmu}, we know that $\PROL{A_1}\ge1-6\delta$. 

If $\mu_1=\mu_2$, by Lemma \ref{lem:theta0}, we know that $\PROL{A|A_1}\ge1-2\delta$. Then $\PROL{A}\ge\PROL{A|A_1}\PROL{A_1}\ge1-8\delta$.
If $\mu_1<\mu_2$, by Lemma \ref{lem:theta0}, we know that $\PROL{A_2|A_1}\ge1-2\delta$, and by Corollary \ref{cor:em}, we know that $\PROL{A_3|A_2,A_1}\ge1-\delta$. Then, $\PROL{A}\ge\PROL{A_1}\PROL{A_2|A_1}\PROL{A_3|A_2,A_1}\ge1-9\delta$.
\end{proof}

Then, we can derive Theorem~\ref{thm:em} in the main paper by setting $\delta=\BIGO(1/\poly(n))$ and $N=N_1+N_2+N_3$.

\subsection{Proof of Lemma \ref{lem:hatmu}}\label{prf:hatmu}
First, we can see that to get an accurate estimation, it suffices to have $\ABSL{\bar{y}-\frac{1}{2}(\mu_1+\mu_2)}<\frac{\Delta}{4}$. Let $N_{11}=\sum_{i=1}^{N_1}1-z_i$, and $N_{12}=\sum_{i=1}^{N_1}z_i$. We have
$$
\bar{y} = \frac{N_{11}}{N_1}\mu_1+\frac{N_{12}}{N_1}\mu_2+\frac{1}{N_1}\sum_{i=1}^{N_1}w_i.
$$
By Hoeffding's inequality, we have
\begin{equation}\label{eq:hfd1}
\PRO{\ABSL{\frac{N_{11}}{N_1}\mu_1-\frac{\mu_1}{2}}<\frac{\Delta}{12}} \ge 1-2\exp(-\frac{\Delta^2N_1}{72\mu_1^2})\ge 1-2\exp(-\frac{N_1}{72B^2}),
\end{equation}
and similarly
\begin{equation}\label{eq:hfd2}
\PRO{\ABSL{\frac{N_{12}}{N_1}\mu_2-\frac{\mu_2}{2}}<\frac{\Delta}{12}} \ge 1-2\exp(-\frac{\Delta^2N_1}{72\mu_2^2})\ge 1-2\exp(-\frac{N_1}{72b^2}).
\end{equation}
By Chernoff's inequality, we have
\begin{equation}\label{eq:chernoff}
\PRO{\ABSL{\frac{1}{N_1}\sum_{i=1}^{N_1}w_i}<\frac{\Delta}{12}}\ge 1-2\exp(-\frac{N_1\Delta^2}{288\sigma^2}).
\end{equation}
By triangle inequality and union bound, we get
$$
\PRO{\ABSL{\bar{y}-\frac{1}{2}(\mu_1+\mu_2)}<\frac{\Delta}{4}}\ge 1-4\exp(-\frac{N_1}{72b^2})-2\exp(-\frac{N_1\Delta^2}{288\sigma^2}),
$$
which completes the proof.

\subsection{Proof of Lemma \ref{lem:theta0}}\label{prf:theta0}
Let $A_1$ be the event that $\hat{\mu}=\frac{1}{2}(\mu_1+\mu_2)$. In this lemma, all the probabilities are conditioned on the event $A_1$. 

First, consider the case when $\mu_1=\mu_2$, i.e., $\theta_*=0$. Let $\tilde{y}:=\frac{1}{\sigma^2}\sum_{i=N_1+1}^{N_1+N_2} \tilde{y}_i^2$. Then, we know that $\tilde{y}$ is $\chi^2$ distributed with $N_2$ degrees of freedom. By the concentration result of $\chi^2$ distribution, we have for any $\epsilon>0$,
$$
\PRO{\ABSL{\frac{1}{N_2}\tilde{y}-1}\ge \epsilon | A_1}\le 2\exp(-\frac{N_2}{8}\min\{1, \epsilon^2\}).
$$
Then, we have
$$
\PRO{\theta_0<\frac{\Delta}{4} | A_1}\ge\PRO{\ABSL{\frac{\tilde{y}^2}{N_2} - 1} < \frac{\Delta^2}{16\sigma^2} | A_1}\ge 1-2\exp(-\frac{N_2}{8}\min\{1,\frac{\Delta^2}{16\sigma^2}\}),
$$
which implies that if 
\begin{equation}\label{eq:n2eq1}
N_2\ge 8\max\{1, 16\frac{\sigma^2}{\Delta^2}\}\log(\frac{1}{\delta}),
\end{equation}
conditioned on $A_1$, the probability that $\theta_0<\frac{\Delta}{4}$ is at least $1-2\delta$.

Then, we consider the case when $\mu_1\neq\mu_2$. In this case, we have $\theta_*\ge \frac{\Delta}{2}$, and we study the probability that $\ABSL{\theta_0-\theta_*}\le\frac{\theta_*}{4}$. We still define $\tilde{y}:=\frac{1}{\sigma^2}\sum_{i=N_1+1}^{N_1+N_2} \tilde{y}_i^2$. We can see that $\tilde{y}$ has noncentral $\chi^2$ distribution with $N_2$ degrees of freedom and noncentrality parameter $\nu=N_2\frac{\theta_*^2}{\sigma^2}$. According to the results of concentrations of non-central $\chi^2$ distribution, 
we have for all $\epsilon>0$,
\begin{equation}\label{eq:chi1}
\PRO{\tilde{y}\ge (N_2+\nu)+2\sqrt{(N_2+\nu)\epsilon}+2\epsilon | A_1} \le \exp(-\epsilon),
\end{equation}
\begin{equation}\label{eq:chi2}
\PRO{\tilde{y}\le (N_2+\nu)-2\sqrt{(N_2+2\nu)\epsilon} | A_1} \le \exp(-\epsilon).
\end{equation}
We analyze the probability that $\frac{\theta_0}{\theta_*}<\frac{5}{4}$. We substitute $\tilde{y}$ and $\nu$ in (\ref{eq:chi1}) with $N_2(\frac{\theta_0^2}{\sigma^2}+1)$ and $N_2\frac{\theta_*^2}{\sigma^2}$, respectively. By some rearrangements, we get
$$
\PRO{\frac{\theta_0^2}{\theta_*^2} \ge 1+2\frac{\sigma}{\theta_*^2}\sqrt{\frac{(\theta_*^2+\sigma^2)\epsilon}{N_2}} + \frac{2\sigma^2\epsilon}{\theta_*^2N_2} | A_1}\le\exp(-\epsilon).
$$
Then, we know that if $N_2$ is large enough such that $\frac{\sigma}{\theta_*^2}\sqrt{\frac{(\theta_*^2+\sigma^2)\epsilon}{N_2}} \le \frac{9}{64}$ and $\frac{\sigma^2\epsilon}{\theta_*^2N_2} \le \frac{9}{64}$, then we have
$$
\PRO{\frac{\theta_0^2}{\theta_*^2} \ge \frac{25}{16} | A_1} = \PRO{\frac{\theta_0}{\theta_*} \ge \frac{5}{4} | A_1} \le\exp(-\epsilon)
$$
By simple algebra and the fact that $\theta_*\ge\frac{\Delta}{2}$, one can see that there exists universal constants $c_3$ such that if $N_2$ satisfies
\begin{equation}\label{eq:n2eq2}
N_2\ge c_3\frac{\sigma^2}{\Delta^2}(1+\frac{4\sigma^2}{\Delta^2})\log(\frac{1}{\delta}),
\end{equation}
then the probability that $\frac{\theta_0}{\theta_*}<\frac{5}{4}$ conditioned on $A_1$ is at least $1-\delta$. Similarly, using (\ref{eq:chi2}), we know that when (\ref{eq:n2eq2}) is satisfied, we can guarantee that $\frac{\theta_0}{\theta_*}>\frac{3}{4}$ with probability at least $1-\delta$. We can complete the proof by union bound.

\subsection{Proof of Corollary \ref{cor:em}}\label{prf:cor}
To guarantee that $\hat{\theta}_*=\theta_*$, we need $\ABSL{\theta_T-\theta_*}<\frac{\Delta}{2}$. By Lemma \ref{lem:theta0}, it suffices to guarantee two facts:
$$
\kappa^T\ABSL{\theta_0-\theta_*}<\frac{\Delta}{4},
$$
and 
$$
\frac{c_6}{1-\kappa}\theta_*\sqrt{\theta_*^2+\sigma^2}\sqrt{\frac{1}{N_3}\log(\frac{1}{\delta})}<\frac{\Delta}{4}.
$$
Conditioning on the event that $\ABSL{\theta_0-\theta_*}<\frac{\theta_*}{4}$ and $\theta_*<b\Delta$, we know that it is sufficient to have $T>\frac{\log(b)}{\log(1/\kappa)}$ and $N_3>\frac{16c_6^2}{(1-\kappa)^2}b^2(b^2\Delta^2+\sigma^2)\log(\frac{1}{\delta})$.

\end{document}